\pdfoutput=1
\documentclass[11pt]{article}
\usepackage[left=1in,top=1in,right=1in,bottom=1.1in,head=.1in,nofoot]{geometry}
\usepackage{setspace,url,bm,amsmath}

\usepackage{titlesec} 
\titlelabel{\thetitle.\quad} 

\titleformat*{\section}{\bf\Large\center}
\usepackage{graphicx}
\usepackage{bbm}
\usepackage{latexsym}
\usepackage{caption}
\usepackage{subcaption}
\usepackage[margin=20pt]{subcaption}
\usepackage{hyperref}
\usepackage{booktabs}
\usepackage{multirow}

\usepackage{enumitem}

\usepackage[table]{xcolor}

\newcommand{\GG}[1]{}

\usepackage{amsthm}
\usepackage{amssymb}
\usepackage{color}

\usepackage{comment}
\theoremstyle{definition}

\newtheorem*{theorem*}{Theorem}
\newtheorem{theorem}{Theorem}
\newtheorem*{rmk*}{Remark}

\newtheorem{proposition}{Proposition}
\newtheorem{lemma}{Lemma}
\newtheorem{algo}{Algorithm}

\newtheorem{condition}{Condition}

\newtheorem{definition}{Definition}
\newtheorem{remark}{Remark}

\newtheorem*{corollary*}{Corollary}

\def\rank{\text{rank}}

\def\dypr{m}

\def\strata{b}
\def\c{\complement}

\usepackage{natbib} 
\bibpunct{(}{)}{;}{a}{}{,}

\usepackage{etoolbox} 
\apptocmd{\sloppy}{\hbadness 10000\relax}{}{} 

\usepackage{color}
\usepackage{listings}

\DeclareMathOperator*{\argmin}{arg\,min}
\DeclareMathOperator*{\argmax}{arg\,max}

\newcommand{\lxr}{\color{black}}

\def\Var{\text{Var}}

\def\converged{\stackrel{d}{\longrightarrow}}

\def\I{\mathbbm{1}}
\def\E{\mathbb{E}}

\def\Pr{\text{Pr}}

\def\bs{}

\def\GT{\text{GT}}
\def\LP{\prime}

\def\c{\complement}
\def\rev{\color{black}}

\allowdisplaybreaks

\usepackage{textcomp}

\usepackage{soul}

\def\Unif{\text{Unif}}

\usepackage[textsize=scriptsize, textwidth = 2cm, shadow]{todonotes}
\usepackage[plain,noend]{algorithm2e}

\RequirePackage[normalem]{ulem}

\usepackage{tikz}
\usetikzlibrary{shapes.geometric, arrows}
\tikzstyle{io} = [trapezium, trapezium left angle=70, trapezium right angle=110, minimum width=1cm, minimum height=1cm, text centered, draw=black,  trapezium stretches=true, thick]
\tikzstyle{process} = [rectangle, minimum width=3cm, minimum height=1cm, text centered, draw=black, thick]
\tikzstyle{decision} = [diamond, minimum width=1cm, minimum height=1cm, text centered, draw=black, thick]
\tikzstyle{arrow} = [thick,->,>=stealth]

\usepackage{subcaption}

\begin{document}
\onehalfspacing
\allowdisplaybreaks

\title{\bf 
Treatment Effect Quantiles in Stratified Randomized Experiments and Matched Observational Studies
}

\author{Yongchang Su and Xinran Li
\footnote{
Yongchang Su is Doctoral Candidate, Department of Statistics, University of Illinois 
at Urbana-Champaign, Champaign, IL 61820 (e-mail: \href{mailto:ysu17@illinois.edu}{ysu17@illinois.edu}). 
Xinran Li is Assistant Professor, Department of Statistics, University of Illinois at Urbana-Champaign, Champaign, IL 61820 (e-mail: \href{mailto:xinranli@illinois.edu}{xinranli@illinois.edu}). 
}
}
\date{}
\maketitle

\vspace{-0.8em}

\begin{abstract}	
Evaluating {\rev the treatment effect} has become an important topic for many applications.  
However, most existing literature focuses mainly on average treatment effects. 
When the individual effects are heavy-tailed 
or have outlier values, 
not only may the average effect not be appropriate for summarizing treatment effects, 
but also the conventional inference for it can be sensitive and possibly invalid due to poor large-sample approximations. 
In this paper we focus on quantiles of individual treatment effects, 
which can be more robust 
in the presence of extreme individual effects. 
Moreover, our inference for 
them 
is purely randomization-based, 
avoiding any 
distributional assumptions on the 
units. 
We first consider inference in stratified randomized experiments, 
extending the recent work by \citet{li2020quantile}.  
We show that 
the computation of 
valid $p$-values for testing null hypotheses on quantiles of individual effects can be 
transformed into instances of the
multiple-choice knapsack problem, which can be efficiently solved exactly or slightly conservatively. 
We then extend our approach to matched observational studies and propose sensitivity analysis to investigate to what extent our inference on quantiles of individual effects is robust to unmeasured confounding.
%
The proposed 
randomization inference and sensitivity analysis are simultaneously valid for all quantiles of individual effects, 
noting that the analysis for the maximum {\rev or minimum} individual effect coincides with the conventional analysis assuming constant treatment effects. 

%
%
%
%
%
%
%
    
\end{abstract}

\section{Introduction}

Evaluating the causal or treatment effect is important in many applications, such as drug development, policy evaluation and design of educational intervention; see, e.g., \citet{Imbens15a}, \citet{angrist2008mostly}, \citet{robins2020}. 
However, most existing literature focuses mainly on the average or certain subgroup average treatment effects. 
First, the average effects may not be the most appropriate summary for treatment effects, especially when individual effects are heterogeneous, heavy-tailed or have outlier values. 
For example,
a treatment has a negative effect for most of the units, but its effect averaging over all units is positive, driven by extremely large effects on a tiny proportion of the population that may be due to outliers. 
Hence, 
simply summarizing by {\rev the average} can be misleading.
Second, inferences that rely on large-sample approximations may perform poorly when data are heavy-tailed or have limited sample size. 
Third, with discrete or ordinal outcomes, the average treatment effect may be difficult to interpret, especially when the numerical values for the outcome do not have obvious meaning and only their relative values matter; see, e.g., \citet{Lu2018} and references therein. 

Motivated by the above concerns, we propose to consider quantiles of individual treatment effects, 
which are closely related to proportions of units with effects greater {\rev or less} than certain thresholds. 
These measures of treatment effects are generally more robust to extreme or outlier values in individual effects. 
Different from usual average effects, 
the quantiles of individual effects and proportions of units with effects passing any thresholds are generally not point identifiable from the observed data \citep[see, e.g.,][]{Fan2010}. 
Because they depend on the joint distribution of the two potential outcomes, which can never be observed simultaneously for the same unit. 
Nevertheless, we can still construct confidence intervals for them. 
Recently, \citet{Lu2018} and \citet{lu2020} studied the sharp bound for the proportion of units benefited from the treatment with ordinal outcomes, and \citet{Rosenblum2019} studied statistical inference for the same estimand in randomized experiments with binary outcomes. 
Importantly, our proposed method can infer the proportion of units with effects passing any threshold and works for general outcomes, which can be continuous, discrete or even a mixture of them. 

In this paper we will consider randomization-based inference for treatment effects, which is also called design-based or finite population inference. 
In particular, we will use solely the random treatment assignments as the reasoned basis for inference \citep{Fisher35a}, and focus on the units in hand without imposing any model or distributional assumptions on them, such as independent and identically distributed (i.i.d.) sampling from some superpopulation. 
Besides, the proposed inference utilizes rank-based test statistics and can be finite-sample valid.
Therefore, 
it is robust to heavy-tailed data and can avoid poor asymptotic approximations.

We first consider inference for quantiles of individual treatment effects in stratified randomized experiments, which builds upon and generalizes the recent work by \citet{li2020quantile}. 
In particular, \citet{li2020quantile} requires exchangeable treatment assignments across all units (e.g., completely randomized experiments), while we require only exchangeable assignments within each stratum and allow units in different strata to have different propensity scores \citep{pscore1983}. 
{\rev The} stratified experiment is not only a popular design in practice \citep{Fisher35a, box2005statistics}, but also can be used to approximate observational studies through careful designs such as substratification or  matching \citep{Imbens15a, Rosenbaum:2010}. 
The inference in stratified experiments encounters additional computation difficulty, especially when the number of strata and stratum sizes are large,  since the calculation of valid $p$-values for testing null hypotheses on quantiles of individual effects involves integer linear programming. 

Fortunately,
the optimization for computing these valid $p$-values 
can be transformed into the classical multiple-choice knapsack problems, 
for which many efficient algorithms have been developed and will be utilized for our purpose \citep{kellerer2004multiple}. 
{\rev The} multiple-choice knapsack problems are generally NP-hard. 
However, in our specific context, they can be solved in polynomial time. 
Specifically, for a sample of size $N$,  
the multiple-choice knapsack problem for calculating one valid $p$-value 
can be solved exactly in $O(N^2)$ time and conservatively in $O(N)$ time, 
where the latter will lead to a conservative $p$-value. 


We then consider sensitivity analysis for quantiles of individual effects in matched observational studies. 
When all confounding is measured and matched perfectly, the matched observational study reduces to a stratified randomized experiment, for which our randomization inference discussed before can be directly applied. 
However, when matching is not exact and more importantly there exists unmeasured confounding, which is inevitable in most observational studies, 
units within the same matched set can have different unknown propensity scores, under which the randomization inference pretending completely randomized assignments within each matched set may provide biased inference for the true treatment effects. 
In this case, a sensitivity analysis is often invoked, which investigates to what extent our inferred causal conclusion is robust to unmeasured confounding \citep{cornfield1959smoking, Rosenbaum02a, ding2016sensitivity, zhaosenIPW2019, Fogarty2020}. 
Here 
we adopt \citet{Rosenbaum02a}'s sensitivity analysis framework, assuming that the odds ratio of propensity scores for any two units within the same matched set is bounded between $1/\Gamma$ and $\Gamma$ for some $\Gamma\ge 1$. 
We then  
propose conservative but still valid inference for quantiles of individual effects under each value of $\Gamma$, and investigate how the inference results vary as $\Gamma$ increases. 
If the effect of interest is significant for a wide range of $\Gamma$, then 
the study 
is considered insensitive to hidden confounding; otherwise, it is sensitive to hidden confounding. 

Furthermore, we demonstrate that the proposed randomization inference and sensitivity analysis for all quantiles of individual effects are simultaneously valid, without the need for any adjustment due to multiple analyses. 
As discussed later, our 
analysis for the maximum {\rev or minimum} individual effects coincides with 
the conventional 
randomization inference and sensitivity analysis 
under the constant treatment effect assumption. 
Therefore, our analysis on all quantiles of individual effects is essentially a costless addition to the conventional analysis, and should thus always be encouraged in practice. 


\section{Framework and Notation}\label{sec:framework}

\subsection{Potential outcomes, individual effects and treatment assignments
}\label{sec:potential_outcome}

We consider an experiment with $S$ strata of $N$ units and two treatment arms, 
where there are $n_s$ units in stratum $s$ 
($1\le s\le S$)
and $N = \sum_{s=1}^S n_s$. 
We 
label the two treatment arms as treatment and control, and 
invoke the potential outcome framework \citep{Neyman23a, Rubin:1974wx}. 
For each unit $i$ in stratum $s$,  
we use $Y_{si}(1)$ and $Y_{si}(0)$ to denote its treatment and control potential outcomes, respectively, 
and $Z_{si}$ to denote its treatment assignment, which equals 1 if the unit receives treatment and 0 otherwise. 
Its individual treatment effect 
is then 
$\tau_{si} =  Y_{si}(1) - Y_{si}(0)$, 
and its 
observed outcome is $Y_{si} = Z_i Y_{si}(1) + (1-Z_i) Y_{si}(0)$. 
We define $\bs{Y}(1) = (Y_{11}(1), \ldots, Y_{Sn_S}(1))^\top$ 
and 
$\bs{Y}(0) = (Y_{11}(0), \ldots, Y_{Sn_S}(0))^\top$ as the treatment and control potential outcome vectors for all units, 
$\bs{\tau} = (\tau_{11}, \ldots, \tau_{Sn_S})^\top$ as the individual treatment effect vector, 
$\bs{Z} = (Z_{11}, \ldots, Z_{Sn_S})^\top$ as the treatment assignment vector, 
and 
$\bs{Y} = 
(Y_{11}, \ldots, Y_{Sn_S})$ as the observed outcome vector. 
Analogously, for each stratum $s$, 
we define $\bs{Y}_s(1), \bs{Y}_s(0), \bs{\tau}_s, \bs{Z}_s$ and $\bs{Y}_s$ to denote vectors of potential outcomes, individual effects, treatment assignments and observed outcomes. 
We further introduce $\tau_{(1)} \le \tau_{(2)} \le \ldots \le \tau_{(N)}$ to denote the individual effects $\tau_{si}$'s 
sorted in an increasing way, 
and 
$n(c) = \sum_{s=1}^S\sum_{i=1}^{n_s} \I(\tau_{si} > c)$ to denote the number of units with treatment effects greater than $c$ for $c\in \mathbb{R}$. 
We will focus on inferring the quantiles of individual effects $\tau_{(k)}$'s 
and the number (or equivalently proportion) of units with effects greater (or smaller) than any threshold.  

Throughout the paper, 
we will conduct finite population inference, also called design-based or randomization-based inference. 
Specifically, all the potential outcomes are viewed as fixed constants or equivalently being conditioned on. 
Therefore, the randomness in the observed data $(\bs{Y}, \bs{Z})$ comes solely from the random treatment assignment $\bs{Z}$, 
whose distribution is often called the treatment assignment mechanism and governs the statistical inference. 
We will focus on a special class of treatment assignment mechanisms, formally defined as follows. 

\begin{definition}\label{def:ERBE}
The experiment is called 
a stratified exchangeable randomized experiment (SERE)
if its treatment assignment mechanism satisfies that 
(i) the assignments are mutually independent across all strata, 
i.e., 
$(\bs{Z}_1, \ldots,\bs{Z}_S)$ are mutually independent; 
(ii) the assignments for units within the same stratum are exchangeable, 
i.e., 
$(Z_{s1}, \ldots, Z_{sn_s}) \sim 
(Z_{s\psi(1)}, \ldots, Z_{s\psi(n_s)})$ 
for any $1\le s \le S$ and any fixed permutation $\psi$ of $\{1,2, \ldots, n_s\}$.  
\end{definition}

If the assignments are mutually independent across all strata, 
and the assignments within each stratum are from either a completely randomized experiment {\rev (CRE), where fixed numbers of units are randomly assigned to treatment and control,}
or a Bernoulli randomized experiment, {\rev where assignments are i.i.d. Bernoulli distributed}, 
then the corresponding stratified experiment is a SERE. 
In particular, we call a stratified experiment as a stratified completely randomized experiment (\textsc{SCRE}) if the assignments are mutually independent across all strata and are from a CRE within each stratum. 

\subsection{Sharp null hypothesis and Fisher randomization test}\label{sec:sharp_null}

To conduct randomization inference for quantiles of individual effects, we will utilize the Fisher randomization test (\textsc{FRT}) for sharp null hypotheses \citep{Fisher35a}. 
The sharp null hypothesis refers to a hypothesis that stipulates all individual treatment effects, 
and has the following general form:
\begin{align}\label{eq:H_delta}
    H_{\bs{\delta}}: \tau_{si} = \delta_{si}, \quad 1\le i \le n_s, \ 1\le s \le S, 
    \ \ \text{or equivalently} \ \ 
    \bs{\tau} = \bs{\delta}, 
\end{align}
where $\bs{\delta} = (\delta_{11}, \ldots, \delta_{Sn_S})^\top \in \mathbb{R}^N$ is a predetermined constant vector. 
When $\bs{\delta}=0$, 
$H_{\bs{\delta}}$ 
is often called Fisher's null of no effect. 
Below we briefly describe the procedure of FRT. 

Under $H_{\bs{\delta}}$ in \eqref{eq:H_delta}, 
we can impute the potential outcomes from the observed data: 
$
\bs{Y}_{\bs{Z}, \bs{\delta}}(1) = \bs{Y} + (\bs{1}-\bs{Z}) \circ \bs{\delta}
$
and 
$
\bs{Y}_{\bs{Z}, \bs{\delta}}(0) = \bs{Y} - \bs{Z} \circ \bs{\delta}, 
$
where $\circ$ denotes element-wise multiplication. 
These \textit{imputed} 
potential outcomes 
are the same as the \textit{true} ones 
if and only if $H_{\bs{\delta}}$ holds. 
Following \citet{Rosenbaum02a}, 
we consider test statistic of the form $t(\bs{Z}, \bs{Y}_{\bs{Z}, \bs{\delta}}(0))$, 
where $t(\bs{z}, \bs{y})$ is a generic function of the treatment assignment vector $\bs{z} \in \{0,1\}^N$ and outcome vector $\bs{y} \in \mathbb{R}^N$. 
When the null $H_{\bs{\delta}}$ holds, the imputed potential outcome $\bs{Y}_{\bs{Z}, \bs{\delta}}(0)$ becomes the same as $\bs{Y}(0)$, 
no longer 
depending on the random treatment assignment. 
Thus, under $H_{\bs{\delta}}$, the randomization distribution of the test statistic $t(\bs{Z}, \bs{Y}_{\bs{Z}, \bs{\delta}}(0))$ is 
\begin{align}\label{eq:tail_prob_tilde}
G_{\bs{Z}, \bs{\delta}} (c) \equiv 
\Pr\left\{
t(\bs{A}, \bs{Y}_{\bs{Z}, \bs{\delta}}(0) )
\ge c
\right\}
= \sum_{\bs{a} \in \{0,1\}^N} \Pr(\bs{A} = \bs{a}) {\rev \I\left \{
t(\bs{a}, \bs{Y}_{\bs{Z},\bs{\delta}}(0))
\ge 
c
\right \},}
\end{align}
where $\bs{A}\sim Z$ is a generic random vector. 
The corresponding randomization $p$-value is 
{\rev \begin{equation}\label{eq:p_val_imp_control}
	p_{\bs{Z}, \bs{\delta}} 
	\equiv  G_{\bs{Z}, \bs{\delta}} \left( t(\bs{Z},\bs{Y}_{\bs{Z}, \bs{\delta}}(0))\right)
	= 
    \sum_{\bs{a}}
    \Pr(\bs{A} = \bs{a}) 
    \I\left \{
	t(\bs{a}, \bs{Y}_{\bs{Z}, \bs{\delta}}(0))
	\ge 
	t(\bs{Z}, \bs{Y}_{\bs{Z}, \bs{\delta}}(0))
	\right \}. 
\end{equation} }


%


\subsection{Stratified rank sum statistic}\label{sec:strat_rank}

The FRT has the advantage of imposing no distributional assumption on the potential outcomes, and uses only the random treatment assignment as the ``reasoned basis'' \citep{Fisher35a}. 
However, FRT can be limited since the sharp null hypothesis stipulates all individual effects and is generally false \textit{a priori} in practice \citep{Hill2002, Gelman13a}. 
Below we introduce a special class of test statistics that can help us 
test 
more general composite null hypotheses.

Specifically, 
we will utilize general stratified rank sum statistics of the following form  \citep[see, e.g.,][]{van1960combination, lehmann1975statistical}:
\begin{align}\label{eq:strat_rank_sum}
    t (\bs{z}, \bs{y}) = \sum_{s=1}^S t_s(\bs{z}_s, \bs{y}_s ) = 
    \sum_{s=1}^S \sum_{i=1}^{n_s} z_{si} \phi_s(\rank_i(\bs{y}_s)), 
\end{align}
where 
$\bs{z}_s$ and $\bs{y}_s$ are subvectors of $\bs{z}$ and $\bs{y}$ corresponding to stratum $s$, 
$\phi_s$ denotes a transformation for the ranks, 
and $\rank_i(\bs{y}_s)$ denotes the rank of the $i$th coordinate of $\bs{y}_s$. We assume index ordering is used to break ties, 
a matter discussed further in \S \ref{sec:ordering}.
When $\phi_s$'s are identity functions, \eqref{eq:strat_rank_sum} reduces to 
the 
stratified Wilcoxon rank sum statistic. 
For descriptive convenience, 
we formally introduce the following class of stratified rank score statistics. 
\begin{definition}\label{def:strat_rank_score}
A statistic $t(\cdot, \cdot)$ is said to be a stratified rank score statistic 
if it has the form in \eqref{eq:strat_rank_sum} with 
increasing 
$\phi_s$'s and uses index ordering for ties assuming 
the order of units has been randomly permuted. 
\end{definition}

Importantly, the stratified rank score statistic is distribution free under the SERE, a key property that will be utilized later.
Moreover, 
using the {\rev \textit{random}} method rather than the common {\rev \textit{average}} method for ties is crucial for this property. 

\begin{proposition}\label{prop:dist_free}
Under the SERE, the stratified rank score statistic $t(\cdot, \cdot)$ is distribution free, in the sense that 
$t(\bs{Z}, \bs{y}) \sim 
t(\bs{Z}, \bs{y}')$ for any $\bs{y}, \bs{y}'\in \mathbb{R}^N$. 
\end{proposition}

Below 
we introduce a special class of statistics satisfying Definition \ref{def:strat_rank_score}, the stratified Stephenson rank sum statistics \citep{Stephenson85a}, 
which 
is of form 
$(\ref{eq:strat_rank_sum})$ with 
\begin{align}\label{eq:stephen_transform}
    \phi_s(r) = 
    \binom{r-1}{h_s-1}
    \text{  if }r > h_s-1, \text{  and }  \phi_s(r)=0 \text{  otherwise, }
    \ (1\le s\le S; 1\le r \le n_s)
\end{align}
for some fixed integers $h_1, \ldots, h_S \ge 2$.  
When $h_1 =  \ldots = h_S = 2$, the statistic is almost equivalent to the stratified Wilcoxon rank sum statistic. 
When $h_s$'s increase, the Stephenson ranks place greater weights on responses with larger ranks, 
making the rank sum statistic more sensitive to larger outcomes as well as larger 
individual treatment effects.
\citet{Conover1988} showed that, under the alternative where the treatment has large effect on a small fraction of the population, 
the Stephenson ranks can lead to the locally most powerful test. 
Such an alternative can be especially relevant in our context for inferring quantiles of individual effects, 
under which we will deliberately remove a certain amount of large effects from the units. 
In the supplementary material, we conduct a simulation study to illustrate the superior performance of the Stephenson rank statistics compared to the Wilcoxon rank statistic.
Relatedly, 
\citet{Rosenbaum2007dramatic} and \citet[][Chapter 16]{Rosenbaum:2010} have  used the Stephenson ranks to weight matched pairs to better detect uncommon but dramatic responses to treatment in observational studies, 
and 
\citet{Rosenbaum2011U, Rosenbaum2014} studied more general rank scores in matched observational studies with multiple controls to improve  the power of a sensitivity analysis. 
 
In this paper, we focus mainly on monotone rank scores as in Definition \ref{def:strat_rank_score}. It is worth noting that these scores do not encompass all the general rank scores discussed in \citet{Rosenbaum2011U, Rosenbaum2014}. 
It will be interesting to further extend the discussion to non-monotone rank scores, which may present additional computation challenges when conducting tests on quantiles of individual effects. We defer this to future research.

\section{Randomization Tests for Quantiles of Individual Treatment Effects}\label{sec:rand_test_quantiles}

\subsection{Composite null hypotheses on quantiles of individual treatment effects}\label{sec:comp_null_quant}
We will focus on the one-sided hypothesis testing with 
alternative hypotheses favoring larger effects, 
and defer the other one-sided and two-sided testing to the end of the paper.
Specifically, we consider the following class of null hypotheses on quantiles of individual effects:
\begin{align}\label{eq:H_Nkc}
    H_{k,c}: \tau_{(k)} \le c 
    \ \Longleftrightarrow \ 
    n(c) \le N-k 
    \ \Longleftrightarrow \ 
    \bs{\tau} \in \mathcal{H}_{k,c}, 
\end{align}
where $1\le k \le N$, $c\in \mathbb{R}$, and $\mathcal{H}_{k,c}$ is defined as 
\begin{align}\label{eq:set_H_Nkc}
    \mathcal{H}_{k,c}
    & \equiv
    \big\{\bs{\delta} \in \mathbb{R}^N: \delta_{(k)} \le c\big\} = 
    \Big\{
    \bs{\delta} \in \mathbb{R}^N: 
    \sum_{i=1}^N \I(\delta_i > c) \le N-k
    \Big\} \subset \mathbb{R}^N,
\end{align}
with $\delta_{(k)}$ denoting the coordinate of $\bs{\delta}$ at rank $k$. 
The equivalence in \eqref{eq:H_Nkc} and 
\eqref{eq:set_H_Nkc} follows 
from definition. 
We further 
define $\delta_{(0)} = -\infty$. Then 
\eqref{eq:H_Nkc} and 
\eqref{eq:set_H_Nkc} still hold 
when $k=0$. 

In contrast to sharp null hypotheses discussed in \S \ref{sec:sharp_null}, 
the hypotheses 
of form \eqref{eq:H_Nkc} are composite, 
i.e., 
some potential outcomes or individual effects may still be unknown 
under 
$H_{k,c}$. 
Thus, the FRT is no longer applicable. 
We can still obtain a valid $p$-value for testing $H_{k,c}$ by maximizing the randomization $p$-value $p_{\bs{Z}, \bs{\delta}}$ in \eqref{eq:H_Nkc} over all $\bs{\delta} \in \mathcal{H}_{k,c}$.
However, 
this is generally not feasible, 
since the optimization 
can be
even NP hard. 
Following \citet{li2020quantile}, we can 
ease 
the optimization by using distribution free test statistics 
in 
Definition \ref{def:strat_rank_score}, as illustrated below.

From Proposition \ref{prop:dist_free}, 
under the SERE and 
for the stratified rank score statistic, 
the imputed randomization distribution in \eqref{eq:tail_prob_tilde} reduces to a distribution that does not depend on the observed treatment assignment $\bs{Z}$ or the hypothesized treatment effect $\bs{\delta}$, i.e., 
\begin{align}\label{eq:G}
    G_{\bs{Z}, \bs{\delta}} (c) \equiv 
    \Pr\left\{
    t(\bs{A}, \bs{Y}_{\bs{Z}, \bs{\delta}}(0) )
    \ge c
    \right\}
    = 
    \Pr\left\{
    t(\bs{A}, \bs{y} )
    \ge c
    \right\}
    \equiv G(c), 
\end{align}
where $\bs{y}$ can be any vector in $\mathbb{R}^N$.
Consequently, 
the valid $p$-value 
for testing $H_{k,c}$ in \eqref{eq:H_Nkc} is 
\begin{align}\label{eq:pNkc}
    p_{k,c} \equiv \sup_{\bs{\delta} \in \mathcal{H}_{k,c}} p_{\bs{Z}, \bs{\delta}}
    & = 
    {\rev\sup_{\bs{\delta} \in \mathcal{H}_{k,c}} G \left( t(\bs{Z},\bs{Y} - \bs{Z}\circ \bs{\delta} )\right)}
    =
     G \big( \inf_{\bs{\delta} \in \mathcal{H}_{k,c}} t(\bs{Z},\bs{Y} - \bs{Z} \circ \bs{\delta}) \big),
\end{align}
where the last equality holds because $G$ is decreasing in $c$ and 
$t(\bs{Z},\bs{Y} - \bs{Z} \circ \bs{\delta})$ can achieve its infimum over $\bs{\delta} \in \mathcal{H}_{k,c}$. 
We summarize the results in the following theorem. 

\begin{theorem}\label{thm:pNkc}
Under the SERE and using a stratified rank score test statistic, 
the $p$-value $p_{k,c} 
$ with equivalent forms in \eqref{eq:pNkc} is valid for testing the 
null hypothesis $H_{k,c}$ in \eqref{eq:H_Nkc}. 
\end{theorem}

Theorem \ref{thm:pNkc} provides a valid $p$-value for testing $H_{k,c}$, whose computation involves only minimization of $t(\bs{Z},\bs{Y} - \bs{Z} \circ \bs{\delta})$ over $\bs{\delta} \in \mathcal{H}_{k,c}$. 
Although much simpler than 
that with a general test statistic, 
the optimization 
in \eqref{eq:pNkc} 
can still be challenging, especially
when the strata number $S$ and sizes $n_s$'s are large, 
as discussed shortly. 

\subsection{Simplifying the optimization for the valid $p$-value}\label{sec:simp_opt}

Below we simplify the optimization in \eqref{eq:pNkc} for calculating the valid $p$-value. 
We first introduce some notation. 
Analogous to \eqref{eq:set_H_Nkc}, 
for $1\le s\le S$, $0\le k_s \le n_s$ and $c\in \mathbb{R}$, 
define
\begin{align*}
    \mathcal{H}_{k_s, c}^s = 
    \big\{\bs{\delta} \in \mathbb{R}^{n_s}: \delta_{(k_s)} \le c \big\} = 
    \Big\{
    \bs{\delta} \in \mathbb{R}^{n_s}: 
    \sum_{i=1}^{n_s} \I(\delta_i > c) \le n_s-k_s
    \Big\} \subset \mathbb{R}^{n_s}, 
\end{align*}
whose elements have at most $n_s-k_s$ coordinates greater than $c$. 
Let $\bs{n} = (n_1, n_2, \ldots, n_S)$ and 
$\mathbb{Z}$ be the set of all integers. 
For $0\le k \le N$, we further define 
\begin{align}\label{eq:setKn}
    \mathcal{K}_{\bs{n}}(k) = 
    \Big\{ (k_1, k_2, \ldots, k_S) \in \mathbb{Z}^S: \ \sum_{s=1}^S k_s = k, \text{ and }\  0\le k_s \le n_s, 1\le s \le S
    \Big\}.
\end{align}
We can then decompose the null set $\mathcal{H}_{k_s,c}$ in \eqref{eq:set_H_Nkc} in the following way:
\begin{align*}
    \mathcal{H}_{k,c} 
    & = \bigcup_{ (k_1, \ldots, k_S) \in \mathcal{K}_{\bs{n}}(k)} 
    \big\{
    \bs{\delta}\in \mathbb{R}^N: 
    \bs{\delta}_s \in \mathcal{H}_{k_s, c}^s, \forall s
    \big\}
    = 
    \bigcup_{ (l_1, \ldots, l_S) \in \mathcal{K}_{\bs{n}}(N-k)} 
    \big\{
    \bs{\delta}: 
    \bs{\delta}_s \in \mathcal{H}_{n_s-l_s, c}^s, \forall s
    \big\}. 
\end{align*}
Recall that $\bs{Z}_s, \bs{Y}_s$ and $\bs{\delta}_s$ are the subvectors of $\bs{Z}, \bs{Y}$ and $\bs{\delta}$ corresponding to stratum $s$, and 
$t_s(\cdot, \cdot)$ is the rank sum statistic for stratum $s$. 
Let 
\begin{align}\label{eq:t_s,c}
    t_{s, c}(l_s)
    & \equiv \inf_{\bs{\delta}_s \in \mathcal{H}_{n_s - l_s, c }^s }
    t_s(\bs{Z}_s, \bs{Y}_s - \bs{Z}_s \circ \bs{\delta}_s)
    \quad
    (1\le s\le S, 0 \le l_s \le n_s, c\in \mathbb{R})
\end{align}
denote the infimum of the rank sum statistic for stratum $s$ when at most $l_s$ units within the stratum can have individual effects greater than $c$. 
We can verify that the optimization for the stratified rank score statistic in \eqref{eq:pNkc}
has the following equivalent form: 
\begin{align}\label{eq:equiv_min_test_stat}
    \inf_{\bs{\delta} \in \mathcal{H}_{k,c}} t(\bs{Z},\bs{Y} - \bs{Z}\circ \bs{\delta}) 
    & 
    = 
    \min_{(l_1, \ldots, l_S) \in \mathcal{K}_{\bs{n}}(N-k)} 
    \sum_{s=1}^S 
    t_{s, c}(l_s). 
\end{align}

Importantly, the computation for $t_{s, c}(l_s)$ in \eqref{eq:t_s,c} has a closed-form solution as demonstrated in \citet{li2020quantile}. 
Here we give some intuitive explanation. 
Under $\mathcal{H}_{n_s - l_s, c}^s$, 
among the $n_s$ units in stratum $s$, 
at most $l_s$ of them can have individual effects greater than $c$. 
To minimize $t_s(\bs{Z}_s, \bs{Y}_s - \bs{Z}_s \circ \bs{\delta}_s)$, 
the $l_s$ treated units with the largest observed outcomes (or all treated units if $l_s > \sum_{i=1}^{n_s}Z_{si}$) 
are hypothesized to have positive infinite individual effects, and the remaining units are hypothesized to have effects of size $c$. 
Specifically, 
$t_{s, c}(l_s) = t_s(\bs{Z}_s, \bs{Y}_s - \bs{Z}_s \circ \bs{\xi}_{s,c}(l_s))$, 
where $\bs{\xi}_{s,c}(l_s) = (\xi_{s,c,1}(l_s), \ldots, \xi_{s,c,n_s}(l_s))^\top$, 
and $\xi_{s, c, i}(l_s)$ equals $\infty$
if $Z_{si}=1$ and $\rank_i(\bs{Y}_s)$ is among the $l_s$ largest elements  of  $\{\rank_j(\bs{Y}_s): Z_{sj} = 1, 1\le j \le n_s\}$, 
and $c$ otherwise. 
Here we define $\infty \cdot 0 = 0$. 
In practice, we can replace $\infty$ by any constant greater than the difference between the maximum observed treated 
and minimum observed control outcomes. 


\subsection{Integer linear and linear programmings for computing the valid $p$-value}\label{sec:ILP_LP}

We can verify that, 
with the equivalent form in \eqref{eq:equiv_min_test_stat}, 
the optimization for the valid $p$-value $p_{k,c}$ in \eqref{eq:pNkc} simplifies to an integer linear programming problem: 
\begin{align}\label{eq:integer_program}
    \min 
    \sum_{s=1}^S \sum_{l=0}^{n_{s}} x_{sl} t_{s, c}(l)   
    \text{ subject to } & \sum_{l=0}^{n_{s}} x_{sl} = 1, 
    \sum_{s=1}^S \sum_{l=0}^{n_{s}} x_{sl} l  = N-k, 
    x_{sl} \in \{0, 1\}, 
     \forall l, s
\end{align}
In \eqref{eq:integer_program}, we are essentially using 
the  dummy variables 
$(x_{s1}, \ldots, x_{sn_s})$ 
to represent the choice of $l_s$ for stratum $s$, 
where $l_s$ corresponds to the only index $l$ such that $x_{sl} = 1$. 
We denote the minimum value of the objective function in \eqref{eq:integer_program} by $t_{k,c}$. 
By the definition in \eqref{eq:pNkc}, 
$p_{k,c} = G(t_{k,c})$. 

It is not difficult to see that 
the solution to the integer programming in (13) is identical to the one with the equality constraint $\sum_{s=1}^S \sum_{l=0}^{n_{s}} x_{sl} l = N-k$ replaced by the inequality constraint $\sum_{s=1}^S \sum_{l=0}^{n_{s}} x_{sl} l \le N-k$, 
since $t_{s,c}(l)$ is 
decreasing 
in $l$ for any $s$ and $c$.
The integer programming in \eqref{eq:integer_program} 
is 
therefore 
essentially 
an instance of the
multiple-choice knapsack problem \citep[see, e.g.,][]{kellerer2004multiple}. 
In particular, we can view $w_{sl} \equiv l$ as the weight and $p_{sl} = -t_{s,c}(l)$ as the profit, and we want to choose exactly one item from each stratum such that the total profit is maximized without exceeding the capacity $N-k$ in the corresponding total weight. 
The general multiple-choice knapsack problem is NP-hard \citep{kellerer2004multiple}. 
However, since the weight capacity in our context is always bounded by $N$, the problem in \eqref{eq:integer_program} can be solved in polynomial time; see the next subsection for more details.


It is straightforward to ease the computation in \eqref{eq:integer_program} by relaxing the integer constraints, which then transforms the optimization into the following linear programming problem: 
\begin{align}\label{eq:linear_program}
    \min 
    \sum_{s=1}^S \sum_{l=0}^{n_{sk}} x_{sl}  t_{s, c}(l)  
    \ 
    \text{subject to } & \sum_{l=0}^{n_{sk}} x_{sl} = 1, 
    \sum_{s=1}^S \sum_{l=0}^{n_{sk}} x_{sl} l  = N-k, 
    0 \le x_{sl} \le 1, 
    \forall l, s
\end{align}
where $n_{sk} = \min\{n_s, N-k\}$. 
We denote the minimum value of the objective function in \eqref{eq:linear_program} by $t_{k,c}^{\LP}$, and  
define 
$p^{\LP}_{k,c} \equiv G(t_{k,c}^{\LP})$. 
Because the feasible region for $x_{sl}$'s in \eqref{eq:linear_program} covers that in \eqref{eq:integer_program} and $G(\cdot)$ is a 
decreasing 
function,  
we must have 
$t_{k,c}^{\LP} \le t_{k,c}$ and $p^{\LP}_{k,c} \ge p_{k,c}$. 
Consequently, 
$p^{\LP}_{k,c}$ is also a valid $p$-value for testing 
$H_{k,c}$ in \eqref{eq:H_Nkc}. 
In sum, the $p$-value $p^{\LP}_{k,c}$ eases the computation from integer programming 
to linear programming, 
while still remaining valid for testing 
$H_{k,c}$. 


\subsection{Algorithms for calculating valid $p$-values}\label{sec:algorithms}

Below we discuss several algorithms to solve the integer and linear programming problems in \eqref{eq:integer_program} and  \eqref{eq:linear_program}, both of which can lead to valid $p$-values for testing the null hypothesis $H_{k,c}$ in \eqref{eq:H_Nkc}. 

We first consider the integer programming in \eqref{eq:integer_program}. 
We can use standard software, such as Gurobi optimizer \citep{gurobi}, to solve it. 
Alternatively, we can reformulate it as a piecewise-linear optimization problem, which can also be solved by Gurobi. 
Due to its connection with the multiple-choice knapsack problem, 
we can also use the dynamic programming algorithm proposed by \citet{DUDZINSKI19873}, 
whose computational complexity is at most of order $N(N-k) \le N^2$ for our specific problem. 

We then consider the linear programming in \eqref{eq:linear_program}. 
We can again use the Gurobi optimizer to solve it. 
From the extensive research on multiple-choice knapsack problems, 
we also use the greedy algorithm developed independently by \citet{dyer1984n} and \citet{zemel1984n}, whose computational complexity is $O(N)$. 
Moreover, 
for stratified rank score statistics with concave rank transformations, which include the stratified Wilcoxon rank sum statistic as a special case, the greedy algorithm solves exactly the integer programming in \eqref{eq:integer_program}. 
However, the rank transformations in stratified Stephenson rank sum statistics are generally convex, under which the greedy algorithm can provide only a conservative solution for \eqref{eq:integer_program}. 

For conciseness, we relegate all the details regarding the above algorithms to the supplementary material. 
We have also implemented these algorithms in our R package, 
and conduct simulation studies to compare their computation time, as detailed in the supplementary material. 
We briefly summarize our findings below. 
The greedy algorithm can be orders of magnitude faster than other algorithms and often provides only slightly conservative solutions; it is thus particularly favored for large sample size and concave rank transformations (for which it is exact for \eqref{eq:integer_program}). 
In cases where we are particularly interested in the number of units with effects passing a given threshold $c$, 
the dynamic programming can be preferred since a single run of it can provide multiple $p$-values, which can further lead to confidence sets for $n(c)$ discussed later in \S \ref{sec:inv_test_conf_set}. 
In general cases,  the Gurobi optimization for the integer programming in \eqref{eq:integer_program} has a superior performance compared to other algorithms for \eqref{eq:integer_program}. 

\subsection{The dependence of the $p$-value on the random ranking of ties}\label{sec:ordering}

As discussed in Definition \ref{def:strat_rank_score}, we invoke the {\rev \textit{random}} method to deal with ties, i.e., units with equal outcomes 
are ranked in a random order, in contrast to the usual {\rev\textit{average}} method.  
For example, we can use the {\rev \textit{first}} method 
under which units with equal outcomes 
are ranked based on their indices, assuming the order of units has been randomly permuted independently of the treatment assignment;
see the R function \textit{rank} for 
details
of these ranking methods \citep{Rsoftware}.
Specifically, with the {\rev \textit{first}} method, 
$\rank_i(\bs{y}_s) < \rank_j(\bs{y}_s)$ 
if and only if (a) $y_{si} < y_{sj}$ or (b) $y_{si} = y_{sj}$ and $i < j$. 
Importantly, 
our approach relies crucially on the {\rev \textit{random}} method 
{\rev or equivalently the {\rev \textit{first}} method under a random ordering of units} 
for ranking ties,
in order to ensure the distribution free property of the rank-based test statistics. 
This can be unsatisfactory in practice compared to other deterministic methods for dealing with ties, 
such as the average statistics and scores discussed in \citet[][Pages 132--133]{HAJEK1999rank}, 
since the resulting inference may depend on the ordering of units. 
Fortunately, as we demonstrate below,
the inference will not be too sensitive to the ordering of units, and moreover, 
we can
conveniently calculate both the maximum and minimum of the $p$-value over all possible ordering of units. 

Let $\overline{\rank}(\cdot)$ denote the rank function that rank ties based on the treatment assignments, i.e., $\overline{\rank}_i(y_s)< \overline{\rank}_j(y_s)$ 
for any $s, i,j$ such that $y_{si}=y_{sj}$, $z_{si}=0$ and $z_{sj}=1$. 
Note that the value of the stratified rank sum statistic 
does not depend on how we rank ties within treated or control groups. 
Analogously, we define $\underline{\rank}(\cdot)$ as the rank function that rank ties based on minus treatment assignments, i.e., within ties control units have greater ranks than treated units. 
Define $\underline{p}_{k,c}$ the same as in \eqref{eq:pNkc} but using the rank function $\overline{\rank}(\cdot)$, 
and analogously $\overline{p}_{k,c}$ using the rank function $\underline{\rank}(\cdot)$. 
The theorem below establishes the relation among these $p$-values.  

\begin{theorem}\label{thm:bound_pval}
The $p$-values $p_{k,c}$, $\underline{p}_{k,c}$ and $\overline{p}_{k,c}$ using different methods for ranking ties (or equivalently the {\rev first} method but under different orderings of units) satisfy that:
\begin{enumerate}[label = (\roman*)]
    \item $\underline{p}_{k,c} \le p_{k,c} \le \overline{p}_{k,c}$ for all $0\le k \le N$ and $c\in \mathbb{R}$; 
    \item $\overline{p}_{k,c}\le \underline{p}_{k,c'}$ for all $0 \le k\le N$ and any $c<c'$. 
\end{enumerate}
\end{theorem}

Theorem \ref{thm:bound_pval} has several implications. 
First, $\underline{p}_{k,c}$ and $\overline{p}_{k,c}$ give deterministic lower and upper bounds of the $p$-value $p_{k,c}$ that depends on the random ordering of units, and these bounds are sharp in the sense that they can be achieved when it happens all control units are ordered before or after treated units. 
Besides, $\overline{p}_{k,c}$ must also be a valid $p$-value for testing $H_{k,c}$, which no longer depends on the ordering of units. 
Second, using the upper bound $\overline{p}_{k,c}$ instead of the original $p_{k,c}$ will lead to little power loss. 
Specifically, 
as discussed later in Remark \ref{rmk:ci_order}, test inversion using either of these three $p$-values will lead to almost the same confidence sets for quantiles of individual effects. 
Finally, in practice, we can also use Theorem \ref{thm:bound_pval} to check sensitivity of our inference results to the ordering of units. 

\begin{remark}\label{rmk:bound_p_lp}
Theorem \ref{thm:bound_pval} also holds when we consider the $p$-value $p_{k,c}^{\LP}$ from the relaxed linear programming in \eqref{eq:linear_program}. 
\end{remark}


\section{Simultaneous Confidence Sets for All Quantiles of Individual Effects}\label{sec:inv_test_conf_set}


From Theorem \ref{thm:pNkc}, $p_{k,c}$ in \eqref{eq:pNkc} is a valid $p$-value for testing $H_{k,c}$ in \eqref{eq:H_Nkc}, which states that the individual effect at rank $k$ is at most $c$ or equivalently the number of units with effects greater than $c$ is at most $N-k$. 
By inverting the tests 
over $c$ and $k$, we can then obtain confidence sets for $\tau_{(k)}$'s and $n(c)$'s. 
Because $p_{k,c}$ is monotone in $k$ and $c$, 
these confidence sets 
have simpler forms. 
More importantly, these confidence sets for $\tau_{(k)}$'s with $1\le k \le N$ and $n(c)$'s with $c\in \mathbb{R}$ are simultaneously valid, 
in the sense that no adjustment due to multiple analyses is needed. 
We summarize the results in the following theorem. 

\begin{theorem}\label{thm:conf_set}
Under the SERE and using a stratified rank score test statistic, 
the 
randomization 
$p$-value 
$p_{k,c}$ 
is 
increasing in $c$ and decreasing in $k$. 
Moreover, for any $\alpha\in (0, 1)$, 
\begin{enumerate}[label=(\roman*)]
    \item $\mathcal{I}_{\alpha}(k) \equiv \{c: p_{k,c} > \alpha, c \in \mathbb{R}\}$ is a $1-\alpha$ confidence set for $\tau_{(k)}$, and it must be an interval of the form $(\underline{c}, \infty)$ or $[\underline{c}, \infty)$, with $\underline{c} = \inf\{c: p_{k,c} > \alpha \}$; 
    \item $\mathcal{S}_{\alpha}(c) \equiv \{N-k: p_{k,c} > \alpha, 0 \le k \le N \}$ is a $1-\alpha$ confidence set for $n(c)$, and it must have the form of $\{j: N-\overline{k} \le j \le N\}$, 
    with 
    $\overline{k} = \sup\{k: p_{k,c} > \alpha \}$;
    \item the intersection of all $1-\alpha$ confidence sets for $\tau_{(k)}$'s, viewed as a confidence set for $\bs{\tau}$, is the same as that 
    for $n(c)$'s, 
    and they have the following equivalent forms:
    \begin{align*}
        \big\{ \bs{\delta}: \delta_{(k)} \in \mathcal{I}_{\alpha}(k), 1\le k \le N \big\}
        & = 
        \big\{
        \bs{\delta}: 
        \sum_{i=1}^N \I(\delta_i > c) \in \mathcal{S}_{\alpha}(c), c \in \mathbb{R}
        \big\}
        =  
        \bigcap_{k, c: p_{k,c} \le \alpha} 
        \mathcal{H}_{k,c}^{\complement},
    \end{align*}
    where $\mathcal{H}_{k,c}^{\complement}$ denotes the complement of $\mathcal{H}_{k,c}$.
    Moreover, the resulting confidence set covers the true individual treatment effect vector $\bs{\tau}$ with probability at least $1-\alpha$, i.e., 
    $
        \Pr(
        \bs{\tau} \in \bigcap_{k, c: p_{k,c} \le \alpha} 
        \mathcal{H}_{k,c}^{\c}
        ) 
        \ge 1-\alpha. 
    $
\end{enumerate}
\end{theorem}



\begin{remark}\label{rmk:conf_set_plus}
Theorem \ref{thm:conf_set} also holds for the easier-to-calculate $p$-value $p_{k,c}^{\LP}$. 
\end{remark}

We can invert the randomization tests for all possible sharp null hypotheses to get confidence sets for the individual treatment effect vector $\tau$. 
However, this is generally computationally infeasible, and the resulting confidence set can be impractical, since it is $N$-dimensional \citep{Rosenbaum:2001}; see, e.g., \citet{Rigdon:2015aa} and \citet{Li:2016tw} for exceptions in the case of binary outcomes.  
To overcome the difficulty, 
\citet{Rosenbaum:2001, Rosenbaum2002attribute} proposed tests that have identical $p$-values for sharp null hypotheses that share the same ``attributable effect'', e.g., total effects on treated units with binary outcomes and displacement effects on treated units with ordered outcomes. Inverting tests can then provide prediction sets for the attributable effect.

Theorem \ref{thm:conf_set}(iii) also provides an $N$-dimensional confidence set for the individual treatment effect vector $\tau$. 
Importantly, it can be efficiently computed and easily visualized. 
In practice, we can plot $k$ against the lower confidence limit for the individual effect at rank $k$ for $1\le k \le N$. 
By construction, for any $c\in \mathbb{R}$, 
we can then count the number of confidence intervals for $\tau_{(k)}$'s that do not cover $c$ to get 
the lower confidence limit for $n(c)$. 
More importantly, 
these confidence sets for $\tau_{(k)}$'s and $n(c)$'s are simultaneously valid.

 Our approach is similar to that in \citet{Rosenbaum:2001, Rosenbaum2002attribute}, in the sense that they both invert randomization tests with carefully designed test statistics to obtain confidence sets for $\tau$. 
Specifically, \citet{Rosenbaum:2001, Rosenbaum2002attribute} used test statistics that can lead to the same $p$-value for a class of sharp null hypotheses, 
while we use the rank statistics to efficiently compute the supremum of the $p$-value over a class of sharp null hypotheses, 
such as those satisfying \eqref{eq:H_Nkc}. 

\begin{remark}\label{rmk:ci_order}
 
From Theorem \ref{thm:bound_pval}, the confidence interval $\mathcal{I}_{\alpha}(k)$ for $\tau_{(k)}$ using either one of $p_{k,c}$, $\overline{p}_{k,c}$ and $\underline{p}_{k,c}$ will be the same, except for inclusion 
of the lower boundary.
This also holds when we consider the $p$-value $p_{k,c}^{\LP}$ from the relaxed linear programming.
\end{remark}

\begin{remark}\label{rmk:conven_rand}
When $k = N$, the $p$-values for testing the null hypothesis of $\tau_{(N)}\le c$ on the maximum individual effect satisfy that 
$p_{N,c} = p_{N,c}^{\LP} = p_{\bs{Z}, c\bs{1}_N}$, which reduces to the usual $p$-value for testing the null hypothesis of constant treatment effect $c$. 
Thus, our inference for the maximum individual effect coincides with the conventional 
randomization 
inference assuming constant treatment effects. 
From Theorem \ref{thm:conf_set} and Remark \ref{rmk:conf_set_plus}, 
our inference on all quantiles of individual effects is actually a costless addition to the conventional randomization inference. 
\end{remark}



\begin{remark}\label{rmk:switch}
The confidence sets 
in Theorems \ref{thm:conf_set} and Remark \ref{rmk:conf_set_plus} are generally less informative when treated group has a smaller size than the control group. 
This asymmetry between treated and control group sizes is actually not surprising, 
since the $p$-value in \eqref{eq:p_val_imp_control} uses only the imputed control potential outcomes. 
In practice, when the control group is expected to have a larger size,  
we prefer to use the imputed treatment potential outcomes for FRT and use the corresponding $p$-values for testing null hypotheses on effect quantiles. 
This can be easily achieved by 
switching treatment labels and changing signs of observed outcomes. 
Moreover, 
we can perform the 
label switching and sign change
for each stratum separately.  
\end{remark}



\section{Sensitivity Analysis for Quantiles of Individual Treatment Effects}\label{sec:sen_ana} 

\subsection{Matched observational studies and sensitivity analysis}

In 
\S 
\ref{sec:rand_test_quantiles}--\ref{sec:inv_test_conf_set} we studied randomization tests for null hypotheses on quantiles of individual effects $\tau_{(k)}$'s as well as numbers of units $n(c)$'s with effects greater than any threshold.  
In this section, we will focus on matched observational studies, which have \textit{unknown} treatment assignment mechanisms due to 
unmeasured confounding. 
Specifically, 
in a matched observational study, 
units within the same stratum may have different {\rev and unknown} propensity scores.  
This obviously violates Definition \ref{def:ERBE} and renders the randomization inference in 
\S 
\ref{sec:rand_test_quantiles}--\ref{sec:inv_test_conf_set} not directly applicable. 




In the following, 
we will invoke \citet{Rosenbaum02a}'s sensitivity analysis framework to investigate to what extent our inference on $\tau_{(k)}$'s and $n(c)$'s is robust to 
unmeasured confounding. 
Specifically, 
we will conduct tests and construct confidence sets for $\tau_{(k)}$'s and $n(c)$'s under a sensitivity model that allows certain amounts of biases in the treatment assignment.  
Below we first introduce Rosenbaum's sensitivity model, 
and then demonstrate how we can utilize the stratified rank score statistic to conduct 
sensitivity analysis for 
quantiles of individual treatment effects. 

We consider a matched observational study with $S$ matched sets, 
where each set $s$ contains 1 treated unit and $n_s-1$ control units; 
see Remark \ref{rmk:switch_sen} for extension to multiple treated or control units from, say, full matching \citep{Rosenbaumfullmatch1991}.
We adopt the same notation in \S \ref{sec:potential_outcome}
to denote the potential outcomes, individual effects, 
observed outcomes and treatment assignments for the $N=\sum_{s=1}^S n_s$ units, where each matched set can be 
viewed as a stratum. 
When matching is exact and takes into account all the confounding,  
units within the same matched set have the same propensity score, 
the corresponding treatment assignment mechanism reduces to a SCRE,  
and thus the randomization inference in \S \ref{sec:rand_test_quantiles}--\ref{sec:inv_test_conf_set} provides valid inference for the treatment effect vector $\bs{\tau}$. 
However, due to inexact matching and more importantly unmeasured confounding, 
the propensity scores for units within the same matched set are generally different. 
Following \citet[][]{Rosenbaum02a}, 
we assume that the odds ratio of the propensity scores for any two units within the same matched set is bounded between $1/\Gamma$ and $\Gamma$ for some constant $\Gamma\ge 1$, i.e., 
\begin{align}\label{eq:Gamma}
    \frac{1}{\Gamma} \le \frac{\pi_{si}/(1-\pi_{si})}{\pi_{sj}/(1-\pi_{sj})} \le \Gamma, 
    \quad (1\le i,j \le n_s, 1 \le s \le S)
\end{align}
where $\pi_{si}$ and $\pi_{sj}$ denote the propensity scores of units $i$ and $j$ in matched set $s$. 
When $\Gamma = 1$, all units within the same matched set have the same propensity score; 
when $\Gamma >  1$, the units can have different propensity scores but their difference is constrained by $\Gamma$ 
as in \eqref{eq:Gamma}. 
The sensitivity analysis investigates how our inference changes as $\Gamma$ changes, and to what extent in terms of 
$\Gamma$ our treatment effect of interest is still significant. 

Let 
$\mathcal{Z}=\{\bs{z}: \sum_{i=1}^{n_s} z_{si} = 1, 1\le s\le S\}$ be the set of all assignments such that there is exactly one 
treated
unit 
within each matched set. 
Then under the constraint \eqref{eq:Gamma} and conditional on $\bs{Z} \in \mathcal{Z}$, 
the assignment mechanism has the 
form in \eqref{eq:dist_Z_gamma}, 
as shown in \citet{Rosenbaum02a}. 
For convenience, we call it a sensitivity model with bias at most $\Gamma$. 
\begin{definition}[Sensitivity model]\label{def:sen_model}
The treatment assignment mechanism is said to follow a sensitivity model with bias at most $\Gamma$, 
if 
it has the following form: 
\begin{align}\label{eq:dist_Z_gamma}
    \Pr_{\bs{u}, \Gamma}(\bs{Z} = \bs{z}) =
    \prod_{s=1}^S
    \frac{\exp(\gamma \sum_{i=1}^{n_s} z_{si} u_{si})}{
    \sum_{i=1}^{n_s} \exp(\gamma u_{si})} 
    \cdot 
    \I(\bs{z} \in \mathcal{Z})
\end{align}
for $\gamma = \log(\Gamma)$ and some (unknown) $\{u_{si}: 1\le i \le n_s, 1\le s\le S\} \in \mathcal{U} \equiv [0, 1]^N$. 
\end{definition}

Obviously, 
the sensitivity model with bias at most $\Gamma=1$ reduces to a SCRE, under which the randomization inference in \S \ref{sec:rand_test_quantiles}--\ref{sec:inv_test_conf_set} is valid. 
For general sensitivity models with $\Gamma> 1$, 
units within each matched set can have different propensity scores, 
and the difference comes from the 
difference in $u_{si}$'s, which can be viewed 
hidden confounding. 
Below we 
study how to 
conduct valid test for 
$H_{k,c}$ in \eqref{eq:H_Nkc} under the sensitivity model. 

\subsection{Sensitivity analysis for quantiles of individual treatment effects}\label{sec:sen_quantile}

We first consider testing sharp null hypotheses using the stratified rank score statistic 
under a sensitivity model. 
Under 
$H_{\bs{\delta}}$ in \eqref{eq:H_delta}, 
the imputed control potential outcome is the same as the true one, 
and 
the tail probability of the randomization distribution of the test statistic under 
a sensitivity model with bias at most $\Gamma$ 
has the following 
forms: 
\begin{align}\label{eq:tail_prob_sen}
G_{\bs{Z}, \bs{\delta}, \bs{u}, \Gamma} (c) 
= 
\sum_{\bs{a} \in \mathcal{Z}} 
\prod_{s=1}^S
    \frac{\exp(\gamma \sum_{i=1}^{n_s} a_{si} u_{si})}{
    \sum_{i=1}^{n_s} \exp(\gamma u_{si})}
\I\left\{
t(\bs{a}, \bs{Y}_{\bs{Z},\bs{\delta}}(0))
\ge 
c
\right\}, 
\end{align}
where 
$\gamma = \log(\Gamma)$ and $\bs{u} = \{u_{si}: 1\le i \le n_s, 1\le s\le S\} \in \mathcal{U}$.  
The resulting $p$-value is then
\begin{align}\label{eq:p_Z_delta_u}
    p_{\bs{Z}, \bs{\delta}, \bs{u}, \Gamma} \equiv G_{\bs{Z}, \bs{\delta}, \bs{u}, \Gamma} \left( t(\bs{Z}, \bs{Y}_{\bs{Z}, \bs{\delta}}(0)) \right)  
\end{align}
Since $\bs{u}$ is unknown, 
this $p$-value 
is not calculable. 
\citet{Rosenbaum02a} proposed to take supremum of the $p$-value over 
$\bs{u}\in \mathcal{U}$ 
to ensure 
its validity for testing
the sharp null $H_{\bs{\delta}}$.

We then 
consider testing for the composite null $H_{k,c}$ in \eqref{eq:H_Nkc} under a sensitivity model.  
To ensure the validity of the test, 
we take the supremum of the $p$-value in \eqref{eq:p_Z_delta_u} 
over both $\bs{\delta} \in \mathcal{H}_{k,c}$ and $\bs{u} \in \mathcal{U}$. 
Unlike the SCRE, 
for general sensitivity models with $\Gamma > 1$, 
the stratified rank score statistic is no longer distribution free, and the imputed tail probability in \eqref{eq:tail_prob_sen} 
generally 
depends on both the imputed potential outcome $\bs{Y}_{\bs{Z},\bs{\delta}}(0)$ and the hidden confounding $\bs{u}$. 
Fortunately, 
as demonstrated in the supplementary material, 
with the stratified rank score statistic, 
the supremum of the tail probability in \eqref{eq:tail_prob_sen} over $\bs{u} \in \mathcal{U}$ does not depend on the imputed potential outcome $Y_{\bs{Z}, \bs{\delta}}(0)$, 
and it can be achieved at some $\bs{u} \in \mathcal{U}$. 
Intuitively, 
the distribution free property holds when we consider the worst-case scenario. 
We then define
\begin{equation}\label{eq:G_Gamma}
    G_{\Gamma}(c) 
    = 
    \max_{\bs{u} \in \mathcal{U}} 
    G_{\bs{Z}, \bs{\delta}, \bs{u}, \Gamma} (c) 
    = 
    \max_{\bs{u} \in \mathcal{U}} 
    \sum_{\bs{a} \in \mathcal{Z}} 
    \prod_{s=1}^S
    \frac{\exp(\gamma \sum_{i=1}^{n_s} a_{si} u_{si})}{
    \sum_{i=1}^{n_s} \exp(\gamma u_{si})}
    \I\left\{
    t( \bs{a}, \bs{y} )
    \ge 
    c
    \right\}, 
\end{equation}
where $\bs{y}\in \mathbb{R}^N$ can be any constant vector.
We can then 
simplify the supremum of the $p$-value
in \eqref{eq:p_Z_delta_u}
over 
all 
sensitivity models with bias at most $\Gamma$ 
and 
all $\bs{\delta}\in \mathcal{H}_{k,c}$. 

\begin{theorem}\label{thm:sen_ana}
Under the sensitivity model with bias at most $\Gamma$ 
as in Definition \ref{def:sen_model}, 
\begin{equation*}
    p_{k,c, \Gamma}  \equiv \sup_{\bs{\delta} \in \mathcal{H}_{k,c}, \bs{u} \in \mathcal{U}} p_{\bs{Z}, \bs{\delta}, \bs{u}, \Gamma} 
    = G_\Gamma \Big( \inf_{\bs{\delta} \in \mathcal{H}_{k,c}} t(\bs{Z}, \bs{Y}_{\bs{Z}, \bs{\delta}}(0)) \Big)
\end{equation*}
is a valid $p$-value 
for testing $H_{k,c}$ in \eqref{eq:H_Nkc}, 
where $G_{\Gamma}(\cdot)$ is defined in \eqref{eq:G_Gamma}. 
\end{theorem}

From \S\ref{sec:rand_test_quantiles}, 
we can find 
the minimum value of the stratified rank score statistic, 
$\inf_{\bs{\delta} \in \mathcal{H}_{k,c}} t(\bs{Z}, \bs{Y}_{\bs{Z}, \bs{\delta}}(0))$, 
by solving the integer programming in \eqref{eq:integer_program}, 
and can find its lower bound by solving the linear programming in \eqref{eq:linear_program}\footnote{With one treated unit per matched set,
the minimization of the test statistic is simple and 
the solution from the 
linear programming
will be exact. 
However, 
this is generally not true after we perform label switching; 
see Remark \ref{rmk:switch_sen} for details. 
}. 
However, for a general matched observational study, 
achieving the exact upper bound $G_\Gamma(\cdot)$ in \eqref{eq:G_Gamma} is challenging both analytically and computationally, except when matched sets are all pairs. 
In the supplementary material, we give the form of $G_\Gamma(\cdot)$ for matched pair studies and construct a finite-sample conservative upper bound of $G_\Gamma(\cdot)$ for general matched studies. 
Below 
we focus on the large-sample approximation of $G_\Gamma(\cdot)$.

\subsection{Large-sample sensitivity analysis}\label{sec:large_sen}

In this subsection we construct a large-sample approximation of $G_{\Gamma}(\cdot)$ 
that can provide asymptotically valid $p$-value 
for testing the null hypothesis $H_{k,c}$ in \eqref{eq:H_Nkc} under the sensitivity model in Definition \ref{def:sen_model}. 
The construction relies crucially on the large-sample Gaussian approximation for the stratified rank score statistic. 
Under the sensitivity model with bias at most $\Gamma$ and unmeasured confounding $\bs{u}$, 
for any $\bs{y} \in \mathbb{R}^N$, 
the 
rank score statistic $t_s(\bs{Z}_s, \bs{y}_s)$ in \eqref{eq:strat_rank_sum} for each matched set $s$ 
has mean and variance: 
\begin{align*}
    \mu_{s, \Gamma, \bs{u}_s}(\bs{y}_s)
    = 
    \frac{
    \sum_{i=1}^{n_s} 
    e^{\gamma u_{si}}
    \phi_s(\rank_i(\bs{y}_s))
    }{
    \sum_{i=1}^{n_s} 
    e^{\gamma u_{si}}
    }, 
    \  
    v_{s, \Gamma, \bs{u}_s}^2(\bs{y}_s)
    = 
    \frac{
    \sum_{i=1}^{n_s} 
    e^{\gamma u_{si}}
    \phi_s^2(\rank_i(\bs{y}_s))
    }{
    \sum_{i=1}^{n_s} 
    e^{\gamma u_{si}}
    } - \mu_{s, \Gamma, \bs{u}_s}^2(\bs{y}_s), 
\end{align*}
where $\bs{u}_s = (u_{s1}, \ldots, u_{sn_s})$. 
By the mutual independence of treatment assignments across all matched sets, 
the stratified rank score statistic $t(\bs{Z}, \bs{y})$ has mean 
$\mu_{\Gamma, \bs{u}}(\bs{y}) = \sum_{s=1}^S \mu_{s, \Gamma, \bs{u}_s}(\bs{y}_s)$ and variance $\sigma^2_{\Gamma, \bs{u}}(\bs{y}) = \sum_{s=1}^S v_{s, \Gamma, \bs{u}_s}^2(\bs{y}_s)$. 
Moreover, as demonstrated in the supplementary material, 
$t(\bs{Z}, \bs{y})$ is asymptotically Gaussian with mean $\mu_{\Gamma, \bs{u}}(\bs{y})$ and variance $v^2_{\Gamma, \bs{u}}(\bs{y})$ under the following regularity condition 
as $S\rightarrow \infty$. 
Define $R_s = \phi_s(n_s) - \phi_s(1)$ for all $s$.  

\begin{condition}\label{cond:match_set_clt}
As $S\rightarrow \infty$, 
$
\max_{1\le s\le S} R_s^2/( \sum_{s=1}^S R_s^2/n_s^3 ) \rightarrow 0. 
$
\end{condition}

Below we give some intuition for Condition \ref{cond:match_set_clt}. 
We consider the case in which all matched sets have bounded sizes, 
i.e., $\max_{1\le s\le S} n_s\le \overline{n}$ for some finite constant $\overline{n}$. 
Suppose the transformation functions $\phi_s$'s are chosen such that, for all $s$, 
$\underline{c} \le R_s 
\le \overline{c}$ for some positive constants $\underline{c}$ and $\overline{c}$. 
For example, 
$\underline{c} = 1$ and $\overline{c} = \overline{n}-1$ when $\phi_s$'s are identity functions. 
In this case, 
$
\max_{1\le s\le S} R_s^2/\sum_{s=1}^S R_s^2/n_s^3 \le (\overline{c}/\underline{c})^2 \cdot \overline{n}^3/S, 
$
and thus 
Condition \ref{cond:match_set_clt} must hold.

Intuitively, 
to maximize the tail probability of the Gaussian approximation of the stratified rank score statistic at values 
no less than 
its maximum possible mean, 
we want to maximize both the mean $\mu_{\Gamma, \bs{u}}(\bs{y})$ and variance $\sigma_{\Gamma, \bs{u}}^2(\bs{y})$. 
Moreover, we want to first maximize the mean and then maximize the variance given that the mean is maximized. 
This is because
the mean is usually much larger than the standard deviation in magnitude. Roughly speaking, 
as $S\rightarrow \infty$, 
if the mean and variance of the rank score statistic for each set is of constant order, then 
$\mu_{\Gamma, \bs{u}}(\bs{y})$ is of order $S$, while $\sigma_{\Gamma, \bs{u}}(\bs{y})$ is of order $\sqrt{S}$. 
From \citet{Rosenbaum2000sep}, 
such an optimization for $\mu_{\Gamma, \bs{u}}(\bs{y})$ and $\sigma_{\Gamma, \bs{u}}^2(\bs{y})$ can be separated into optimization for $\mu_{s, \Gamma, \bs{u}_s}(\bs{y}_s)$ and $v_{s, \Gamma, \bs{u}_s}(\bs{y}_s)$ within each matched set, which can be solved efficiently. 
Furthermore, 
the maximized mean and variance of the stratified rank score statistic will no longer depend on the value of $\bs{y}$. That is, we are able to define 
\begin{align}\label{eq:mu_v2_Gamma}
    \tilde{\mu}_{\Gamma} \equiv \max_{\bs{u}\in \mathcal{U}} \mu_{\Gamma, \bs{u}}(\bs{y}) = \sum_{s=1}^S \tilde{\mu}_{s, \Gamma}
    \quad \text{and} \quad
    \tilde{\sigma}_{\Gamma}^2 \equiv \max_{\bs{u}\in \mathcal{U}: \mu_{\Gamma, \bs{u}}(\bs{y}) = \tilde{\mu}_{\Gamma}} \sigma^2_{\Gamma, \bs{u}}(\bs{y})
    = \sum_{s=1}^S \tilde{v}^2_{s, \Gamma}, 
\end{align}
where $\tilde{\mu}_{s, \Gamma}$ is the maximum mean for rank score statistic in set $s$ and $\tilde{v}^2_{s, \Gamma}$ is the corresponding maximum variance,  
and they can be efficiently computed 
by 
\begin{align}\label{eq:mu_s_Gamma}
   \tilde{\mu}_{s, \Gamma} & \equiv 
   \max_{\bs{u}_s \in [0,1]^{n_s}}\mu_{s, \Gamma, \bs{u}_s}(\bs{y}_s)
   = 
   \max_{1\le j\le n_s}
   \frac{\sum_{i=1}^j \phi_s(i) + \Gamma \sum_{i=j+1}^{n_s} \phi_s(i)}{
   j + \Gamma (n_s - j)
   }, 
\\
\label{eq:v_s_Gamma}
    \tilde{v}_{s, \Gamma}^2 & \equiv \max_{\bs{u}_s:\  \mu_{s, \Gamma, \bs{u}_s}(\bs{y}_s) = \tilde{\mu}_{s, \Gamma}}
    v_{s, \Gamma, \bs{u}_s}^2(\bs{y}_s)
    = 
    \max_{j \in \mathcal{M}_{s, \Gamma} } 
   \frac{\sum_{i=1}^j \phi_s^2(i) + \Gamma \sum_{i=j+1}^{n_s} \phi_s^2(i)}{
   j + \Gamma (n_s - j)
   }
    - \tilde{\mu}_{s, \Gamma}^2, 
\end{align}
with $\mathcal{M}_{s, \Gamma} \subset \{1, 2, \ldots, n_s\}$ being the set of $j$ that can achieve the maximum in \eqref{eq:mu_s_Gamma}.

Let $\tilde{G}_{\Gamma}(\cdot)$ be the tail probability of the Gaussian distribution with mean $\tilde{\mu}_{\Gamma}$ and variance $\tilde{\sigma}^2_{\Gamma}$. 
To ensure the asymptotic validity of the resulting $p$-values, we need $\tilde{G}_{\Gamma}(\cdot)$ to have asymptotically heavier tail than the true distribution of the stratified rank score statistic $t(\bs{Z}, \bs{Y}(0))$ using the true control potential outcomes. 
Let $\mu_s = \E\{ t_s(\bs{Z}_s, \bs{Y}_s(0)) \}$ and $v_s^2 = \Var\{t_s(\bs{Z}_s, \bs{Y}_s(0))\}$ be the true mean and variance of the rank score statistic for each matched set $s$. 
We emphasize that both $\mu_s$'s and $v_s^2$'s are unknown, since they depend on the unknown potential outcomes and unknown treatment assignment mechanism. 
Suppose that the true treatment assignment mechanism satisfies the sensitivity model with bias at most $\Gamma$. 
By the construction in \eqref{eq:mu_s_Gamma} and \eqref{eq:v_s_Gamma}, $\mu_s \le \tilde{\mu}_{s, \Gamma}$, and moreover, if $\mu_s = \tilde{\mu}_{s, \Gamma}$, then $v_s^2 \le \tilde{v}_{s, \Gamma}^2$. 
Let $\mathcal{A} = \{s: v_s^2 > \tilde{v}_{s, \Gamma}^2, 1 \le s \le S \}$. 
Then we must have $\mu_s < \tilde{\mu}_{s, \Gamma}$ for $s\in \mathcal{A}$. 
Define  
$\Delta_{\Gamma}(\mu) = |\mathcal{A}|^{-1} \sum_{s\in \mathcal{A}}(\tilde{\mu}_{s, \Gamma} - \mu_s)$ and 
$\Delta_{\Gamma}(v^2) = |\mathcal{A}|^{-1} \sum_{s\in \mathcal{A}}(v_s^2 - \tilde{v}_{s, \Gamma}^2)$ as the average differences between $\tilde{\mu}_{s, \Gamma}$ and $\mu_s$ and between $v_s^2$ and $\tilde{v}_{s, \Gamma}^2$ for set $s$ in $\mathcal{A}$, respectively. 
Obviously, 
both $\Delta_{\Gamma}(\mu)$ and $\Delta_{\Gamma}(v^2)$ 
are positive when $\mathcal{A} \ne \emptyset$. 
For descriptive convenience, we define $\Delta_{\Gamma}(\mu)/\Delta_{\Gamma}(v^2)$ to be 
$\infty$
when $\mathcal{A} = \emptyset$.  
We invoke the following regularity condition on the relative magnitude of $\Delta_{\Gamma}(\mu)$ and $\Delta_{\Gamma}(v^2)$ as $S\rightarrow \infty$. 

\begin{condition}\label{cond:sen_asymp_conservative}
As $S\rightarrow \infty$, 
$\sqrt{\sum_{s=1}^S R_s^2/n_s^3} \cdot \Delta_{\Gamma}(\mu)/\Delta_{\Gamma}(v^2) \rightarrow \infty$. 
\end{condition}

From the discussion after Condition \ref{cond:match_set_clt}, when all matched sets have bounded sizes and the ranges of the transformed ranks are bounded between two positive constants, 
$\sum_{s=1}^S R_s^2/n_s^3$ is of order $S$, and Condition \ref{cond:sen_asymp_conservative} 
reduces 
to $\sqrt{S} \cdot \Delta_{\Gamma}(\mu)/\Delta_{\Gamma}(v^2) \rightarrow \infty$.  
This intuitively requires the ratio between the average differences in mean and variance 
to be much larger than $1/\sqrt{S}$. 
Under Conditions \ref{cond:match_set_clt} and \ref{cond:sen_asymp_conservative}, we are able to conduct 
asymptotic
sensitivity analysis.  

\begin{theorem}\label{thm:sen_ana_large_sample}
Under the sensitivity model with bias at most $\Gamma$ as in Definition \ref{def:sen_model}, 
if 
Conditions \ref{cond:match_set_clt} and \ref{cond:sen_asymp_conservative} hold, 
then the following two $p$-values
\begin{align}\label{eq:p_tilde}
    \tilde{p}_{k,c, \Gamma} \equiv \tilde{G}_{\Gamma}(t_{k,c})
    \ \le \
    \tilde{p}_{k,c, \Gamma}^{\LP} \equiv \tilde{G}_{\Gamma}(t^{\LP}_{k,c})
\end{align}
are both asymptotically valid for testing the null hypothesis $H_{k,c}$ in \eqref{eq:H_Nkc} at significance level $\alpha\in (0, 0.5]$, 
where $\tilde{G}_{\Gamma}(\cdot)$ is the tail probability of the Gaussian distribution with mean $\tilde{\mu}_{\Gamma}$ and variance $\tilde{\sigma}^2_{\Gamma}$ defined as in \eqref{eq:mu_v2_Gamma}. 
That is, if the null hypothesis $H_{k,c}$ holds and the bias in the treatment assignment is bounded by $\Gamma$, then for any $0<\alpha \le 0.5$, 
\begin{align*}
    \limsup_{S\rightarrow \infty}\Pr(\tilde{p}_{k,c, \Gamma}^{\LP} \le \alpha) 
    \le 
    \limsup_{S\rightarrow \infty} \Pr(\tilde{p}_{k,c, \Gamma} \le \alpha)
    \le \alpha. 
\end{align*}
\end{theorem}

\subsection{Simultaneous sensitivity analysis for all quantiles of individual treatment effects}\label{sec:simu_sen}


From Theorem
\ref{thm:sen_ana_large_sample}, 
we are able to test null hypotheses of form \eqref{eq:H_Nkc} under 
a sensitivity model with bias at most $\Gamma$. 
By the same logic as Theorem \ref{thm:conf_set},
we can 
then invert the tests to get confidence sets for $\tau_{(k)}$'s and $n(c)$'s. 
Moreover, 
these confidence sets for $\tau_{(k)}$'s and $n(c)$'s 
will be 
simultaneously valid. 
We summarize the results in the following theorem.  


\begin{theorem}\label{thm:sen_inv}
Under the sensitivity model with bias at most $\Gamma$, for any $\alpha\in (0,0.5]$, 
    if Conditions \ref{cond:match_set_clt} and \ref{cond:sen_asymp_conservative} hold,
    then 
    the conclusions in Theorem \ref{thm:conf_set} hold for the $p$-values  $\tilde{p}_{k,c,\Gamma}$ and $\tilde{p}^\LP_{k,c,\Gamma}$  in \eqref{eq:p_tilde}, 
    but with the confidence sets being asymptotically valid. 
\end{theorem}

Similar to the discussion after Theorem \ref{thm:conf_set}, 
under each $\Gamma\ge 1$ and the sensitivity model with bias at most $\Gamma$, 
we can visualize the confidence sets in Theorem \ref{thm:sen_inv} by plotting $k$ against the lower confidence limit for $\tau_{(k)}$ for 
$1\le k \le N$, which 
will be 
simultaneously valid.  
Furthermore, 
we 
will 
investigate how these confidence intervals change as $\Gamma$ varies. 
This can tell us how robust our inference for quantiles of individual effects is to hidden confounding. 
As $\Gamma$ increases, the class of sensitivity models becomes larger and the resulting confidence intervals for $\tau_{(k)}$'s will become wider. 
Under a sensitivity analysis, we 
are especially interested in the cutoff value of $\Gamma$ such that the treatment effect of interest becomes insignificant, which measures the degree of its insensitivity to hidden confounding. 
In practice, we can report these cutoff values for quantiles of individual effects; see the application in \S \ref{sec:application} for more details.  
 
\begin{remark}
By the same logic as Theorem \ref{thm:bound_pval}, we can get sharp bounds for  $\tilde{p}_{k,c,\Gamma}$ simply by ordering all control units before or after treated units. 
Moreover, similar to Remark \ref{rmk:ci_order}, 
the confidence intervals for quantiles of individual effects using either $\tilde{p}_{k,c,\Gamma}$ or its bounds will be the same, except for the inclusion of the lower boundaries.
These also hold when we consider the $p$-value $\tilde{p}_{k,c,\Gamma}^{\LP}$ from the relaxed linear programming.
\end{remark}

\begin{remark}
Similar to 
Remark \ref{rmk:conven_rand}, 
our sensitivity analysis for the maximum individual effect $\tau_{(N)}$ reduces to 
the 
conventional sensitivity analysis assuming 
constant treatment effects. 
Thus, 
our sensitivity analysis on quantiles of individual effects is a costless addition to 
the 
conventional sensitivity analysis
{\rev due to the simultaneous validity in Theorem \ref{thm:sen_inv}}, 
and it provides more robust inference for treatment effects 
{\rev since quantiles are generally more robust than the maximum}. 
Recently, 
\citet{Fogarty2020, fogarty2019testing} extended the conventional sensitivity analysis to infer average 
treatment effects, while we extend it to infer all quantiles of individual treatment effects. 
\end{remark}


\begin{remark}\label{rmk:switch_sen}
We focused on matched observational studies with exactly one treated unit within each matched set. 
By the same logic as \citet[][Pages 161--162]{Rosenbaum02a}, 
we can extend our approach to observational studies 
with one treated or one control unit within each matched set. 
This is particularly relevant for our sensitivity analysis on quantiles of individual treatment effects. 
Similar to Remark \ref{rmk:switch}, 
the relative sizes of treated and control groups matter
for the power of our analysis, 
and generally we prefer larger treated group. 
Thus, in practice, we suggest to make each set contain only one control unit, through switching treatment labels and changing outcome signs. 
However, with multiple treated units rather than one within each set, the minimization of the stratified rank score statistic will become more challenging.
As demonstrated by simulation, 
switching can greatly improve the power of our sensitivity analysis for quantiles of individual treatment effects. 
We 
relegate the details 
to the supplementary material. 
\end{remark}

\section{Effect of Smoking on the Blood Cadmium Level}\label{sec:application}

We apply the proposed methods to study the effect of smoking on the blood cadmium level ($\mu g/l$) using data from the 2005-2006 National Health and Nutrition Examination Survey, which are also available in \citet{bigmatch2020}. 
We perform the optimal matching in \citet{YuRosenbaum2019}, taking into account gender, age, race, education level, household income level and body mass index. 
The matched data contain $512$ matched sets, each of which contain 1 daily smoker and 2 matched nonsmokers, 
with in total $1536$ units. 


\begin{figure}[htb]
    \centering
    \begin{subfigure}[htbp]{0.5\textwidth}
    \centering
        \includegraphics[width=.7\textwidth]{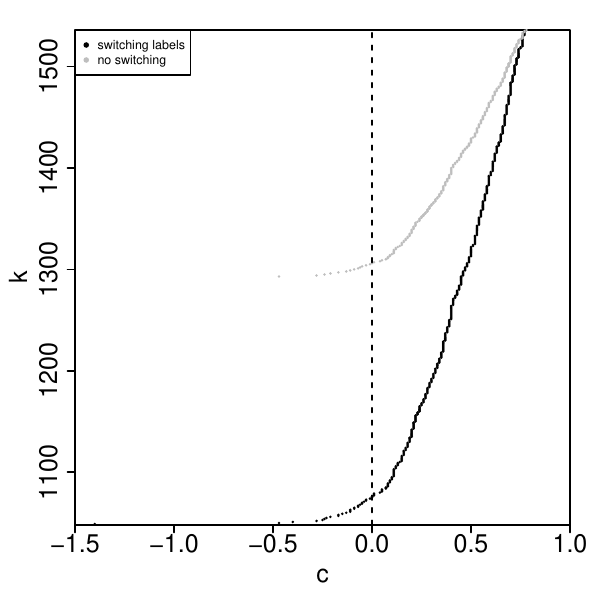}
        \caption{Randomization inference}
    \end{subfigure}%
    \begin{subfigure}[htbp]{0.5\textwidth}
    \centering
        \includegraphics[width=.7\textwidth]{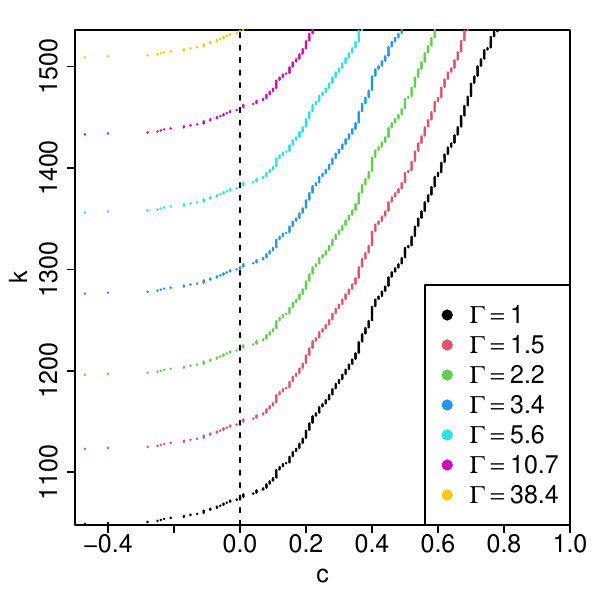}
        \caption{Sensitivity analysis}
    \end{subfigure}
    \caption{$90\%$ lower confidence limits for quantiles of individual effects 
    The left figure shows the lower confidence limits for quantiles of effects assuming there is no hidden confounding, {\rev i.e., the study reduces a SCRE}, using either the original matched data or the one with switched treatment labels and changed outcome signs. 
    The right figure shows the lower confidence limits for quantiles of individual effects under various sensitivity models indexed by different $\Gamma$'s using the data with switched treatment labels and changed outcome signs. 
    }
    \label{fig:CI_nh}
\end{figure}

We first assume that matching has taken into account all confounding, under which the units within each matched set have the same probability to smoke. 
In this case, the matched observational study essentially reduces to a SCRE, for which we apply the randomization inference in \S \ref{sec:rand_test_quantiles}--\ref{sec:inv_test_conf_set}
with the stratified Wilcoxon rank sum statistic. 
Figure \ref{fig:CI_nh}(a) shows the $90\%$ lower confidence limits for all quantiles of individual effects, using both the original dataset and the one with switched treatment labels and changed outcome signs. 
Specifically, for each 
point 
in \ref{fig:CI_nh}(a), its $x$-axis value denotes the lower confidence limit for the individual effect at rank its $y$-axis value.
We 
omit those noninformative minus infinite lower confidence limits for $\tau_{(k)}$'s with small $k$. 
From Figure \ref{fig:CI_nh}(a), it is obvious that treatment label switching helps provide more informative confidence intervals for quantiles of individual effects. 
Moreover, these intervals are simultaneously valid, as implied by Remark
\ref{rmk:conf_set_plus}. 
With label switching, the $90\%$ lower confidence limit for the individual effect at rank $1076$ is 0.01. 
This implies that, with $90\%$ confidence level, at least $30\%$ of the units have positive effects, or equivalently smoking would cause higher blood cadmium level for at least $30\%$ of the units in the study.

We then consider sensitivity analysis, since the previous randomization inference pretending the study is a SCRE may be invalid due to the existence of unmeasured confounding. 
We apply the sensitivity analysis for quantiles of individual effects in \S \ref{sec:sen_ana} to the observed data with switched treatment labels and changed outcome signs.
Figure \ref{fig:CI_nh}(b) shows the $90\%$ lower confidence limits for all quantiles of individual effects under various sensitivity models with biases ranging from $1.0$ to $38.4$. 
These values of $\Gamma$, 
$1.0$, $1.5$, $2.2$, $3.4$, $5.6$, $10.7$ and $38.4$, are actually the largest values of $\Gamma$ such that the resulting confidence intervals for the $70\%$, $75\%$, $80\%$, $85\%$, $90\%$, $95\%$ and $100\%$ quantiles of individual effects do not cover zero. 
For example, 
 when the bias in the treatment assignment is at most 2.2,
the individual effect at rank 1229, {\rev i.e., $80\%$ quantile,} is positive at significance level $0.1$, 
which implies that smoking would cause higher blood cadmium level for at least $20\%$ of the units in the study. 
Equivalently, we say the $80\%$ quantile of individual effect is significant for biases up to 2.2. 
{\rev From \citet[][Chapter 9]{rosenbaum2017book}, intuitively, 
a bias of magnitude 2 or 5 corresponds to an unobserved covariate that could produce a 3-fold or 9 fold increase in the odds of smoking, and 5-fold or 11-fold increase in the odds of a positive pair difference in cadmium levels}. 
Since such a strong confounding is unlikely to exist, especially after 
taking
into account those observed covariates in matching, 
we believe that smoking causes higher blood cadmium level for a non-negligible proportion of units, and the causal conclusion is robust to unmeasured confounding.

\section{Conclusion}\label{sec:discuss}

We studied nonparametric randomization-based inference and sensitivity analysis for quantiles of individual effects. 
We focused on alternative hypotheses that favor larger treatment effects. 
We can straightforwardly extend it to deal with alternative hypotheses that favor smaller treatment effects, by switching treatment labels or changing outcome signs. 
We can also conduct two-sided testing by combining the two one-sided testings using Bonferroni's method \citep{cox1977}.

Our inference utilizes rank-based test statistics, 
and is applicable to both stratified randomized experiments and matched observational studies. 
Our inference is actually a costless addition to the conventional randomization inference and sensitivity analysis assuming constant treatment 
effects, and thus should always be encouraged in practice. 
The inference results can be conveniently visualized and interpreted, and
a 
publicly available
R package 
has also been developed. 


It will be interesting to further incorporate pretreatment covariates to either improve the inference efficiency \citep{Rosenbaum2002cov, lin2013} or consider conditional quantiles of individual treatment effects \citep{Fan2010}. We leave them for future research.  

\section{Acknowledgement}
We are grateful to the editor and reviewers for their invaluable comments and suggestions, which have significantly improved the quality of this work. 

\section{Supplementary material}\label{SM}
Supplementary material available at \textit{Biometrika} online includes simulation studies, proofs of all theorems and additional technical details. 
An R package implementing the proposed methods is available at \texttt{https://github.com/Yongchang-Su/QIoT/}. 

\bibliographystyle{plainnat}
\bibliography{quantile_rbd}

\pagebreak

\begin{center}
\textbf{\large Supplementary Material to ``Treatment Effect Quantiles 
in Stratified Randomized Experiments and Matched Observational Studies''}
\end{center}
\setcounter{equation}{0}
\setcounter{section}{0}
\setcounter{figure}{0}
\setcounter{example}{0}
\setcounter{proposition}{0}
\setcounter{corollary}{0}
\setcounter{theorem}{0}
\setcounter{lemma}{0}
\setcounter{table}{0}
\setcounter{condition}{0}
\setcounter{remark}{0}
\setcounter{page}{1}

\renewcommand {\theproposition} {A\arabic{proposition}}
\renewcommand {\theexample} {A\arabic{example}}
\renewcommand {\thefigure} {A\arabic{figure}}
\renewcommand {\thetable} {A\arabic{table}}
\renewcommand {\theequation} {A\arabic{equation}}
\renewcommand {\thelemma} {A\arabic{lemma}}
\renewcommand {\thesection} {A\arabic{section}}
\renewcommand {\thetheorem} {A\arabic{theorem}}
\renewcommand {\thecorollary} {A\arabic{corollary}}
\renewcommand {\thecondition} {A\arabic{condition}}
\renewcommand {\theremark} {A\arabic{remark}}
\renewcommand {\thepage} {A\arabic{page}}

\S \ref{sec:sen_finite} conducts finite-sample sensitivity analysis, which is exact for matched pair studies but generally conservative. 
{\lxr \S \ref{sec:eff_alg}
gives the 
details for the algorithms used to calculate valid $p$-values.} 
\S \ref{sec:simu_study} conducts simulation studies, comparing various optimization algorithms and illustrating the power gain from Stephenson rank statistics and treatment label switching. 
\S \ref{sec:proof_prop_thm} gives the proofs of all propositions and theorems, as well as some technical remarks. 
\S \ref{sec:add_opi} gives additional technical details for optimizations.

\section{Finite-sample sensitive analysis}\label{sec:sen_finite}


Here we consider performing finite-sample valid sensitivity analysis. 
We first construct a finite-sample upper bound of $G_{\Gamma}(\cdot)$ that has a simple analytical expression and is easy to approximate by Monte Carlo. 
For each set $s$, 
let $\xi_{s1} < \ldots < \xi_{s m_s}$ be the unique values of 
the transformed ranks
$\{\phi_s(1), \ldots, \phi_s(n_s)\}$,
and
$g_{si} = \sum_{j=1}^{n_s} \I\{\phi_s(j)\ge \xi_{si}\}$.
For $\xi_{s, i-1}< c \le \xi_{si}$ and any constant $\bs{y}_s\in \mathbb{R}^{n_s}$, 
we can verify that, under the sensitivity model with bias at most $\Gamma$, 
$t_s(\bs{Z}_s, \bs{y}_s)$ has at most probability $g_{si} \Gamma/\{(n_s-g_{si}) + g_{si}\Gamma\}$ to be
no less than 
$c$. 
This motivates us to define mutually independent random variables 
$(\overline{T}_1, \ldots, \overline{T}_S)$
with probability mass functions: 
\begin{align}\label{eq:T_overline_set}
    \Pr\left( \overline{T}_s = \xi_{si} \right)
    & = 
    \frac{g_{si}\Gamma}{ (n_s - g_{si}) + g_{si}\Gamma }
    - 
    \frac{g_{s,i+1}\Gamma}{ (n_s-g_{s,i+1}) + g_{s,i+1}\Gamma }
    , 
    \quad 
    (i = 1, 2, \ldots, m_s), 
\end{align}
where $g_{s, m_s+1}$ is defined to be zero. 
By the mutual independence of treatment assignments across matched sets, 
the stratified rank score statistic $t(\bs{Z}, \bs{y})$ 
must be 
stochastically smaller than or equal to $\sum_{s=1}^S \overline{T}_s$. 
Therefore, the tail probability of 
$\sum_{s=1}^S \overline{T}_s$, 
denoted by $\overline{G}_{\Gamma}(\cdot)$, 
must be an upper bound of $G_{\Gamma}(\cdot)$ in \eqref{eq:G_Gamma}. 
This further helps us construct finite-sample valid $p$-values for testing null hypotheses on quantiles of individual treatment effects under sensitivity models. 

\begin{theorem}\label{thm:sen_conservative_finte_sample}
Under the sensitivity model with bias at most $\Gamma$ as in Definition \ref{def:sen_model}, 
the following two $p$-values 
\begin{align}\label{eq:p_overline}
    \overline{p}_{k,c, \Gamma} \equiv \overline{G}_{\Gamma}(t_{k,c}) 
    \ \le \ 
    \overline{p}_{k,c, \Gamma}^{\LP} \equiv \overline{G}_{\Gamma}(t^{\LP}_{k,c})
\end{align}
are both valid for testing the null hypothesis $H_{k,c}$ in \eqref{eq:H_Nkc}, 
where $\overline{G}_\Gamma(\cdot)$ is the tail probability of 
$\sum_{s=1}^S \overline{T}_s$ with mutually independent $\overline{T}_s$'s defined in \eqref{eq:T_overline_set}. 
That is, if the null hypothesis $H_{k,c}$ holds and the bias in the treatment assignment is bounded by $\Gamma$, then, for any $\alpha\in (0,1)$, 
\begin{align*}
    \Pr(\overline{p}_{k,c, \Gamma}^{\LP} \le \alpha)
    \le 
    \Pr(\overline{p}_{k,c, \Gamma} \le \alpha) \le \alpha. 
\end{align*}
\end{theorem}

As commented in \S \ref{sec:match_pair_Gamma}, 
in a matched pair study with $n_1 = \ldots = n_S = 2$,
$\overline{G}_{\Gamma}(\cdot)$ becomes the same as $G_{\Gamma}(\cdot)$, i.e., the optimization in \eqref{eq:G_Gamma} has a closed-form analytical solution. 
However, for general matched studies with $n_s\ge 2$, $\overline{G}_{\Gamma}(\cdot)$ is generally different from $G_{\Gamma}(\cdot)$ and may 
be too conservative. 
In \S \ref{sec:large_sen}, we construct an asymptotic approximation for $G_{\Gamma}(\cdot)$ that is sharper than $\overline{G}_{\Gamma}(\cdot)$ and can thus lead to more powerful sensitivity analysis in large samples. 

Analogous to Theorems \ref{thm:conf_set} and \ref{thm:sen_inv}, 
we can then invert the tests to construct simultaneously valid confidence sets for the quantiles of individual effects $\tau_{(k)}$'s and numbers of units with effects greater than any threshold $n(c)$'s. 
We summarize the results in the following theorem. 

\begin{theorem}\label{thm:sen_inv_finite}
Under the sensitivity model with bias at most $\Gamma$,
the 
$p$-value 
$\overline{p}_{k,c,\Gamma}$ 
is increasing in $c$ and decreasing in $k$. 
Moreover, for any $\alpha\in (0, 1)$, 
\begin{enumerate}[label=(\roman*)]
    \item $\overline{\mathcal{I}}_{\alpha}(k) \equiv \{c: \overline{p}_{k,c,\Gamma} > \alpha, c \in \mathbb{R}\}$ is a $1-\alpha$ confidence set for $\tau_{(k)}$, and it must be an interval of the form $(\underline{c}, \infty)$ or $[\underline{c}, \infty)$, with $\underline{c} = \inf\{c: \overline{p}_{k,c,\Gamma} > \alpha\}$; 
    \item $\overline{\mathcal{S}}_{\alpha}(c) \equiv \{N-k: \overline{p}_{k,c, \Gamma} > \alpha, 0 \le k \le N \}$ is a $1-\alpha$ confidence set for $n(c)$, and it must have the form of $\{j: N-\overline{k} \le j \le N\}$, 
    with 
    $\overline{k} = \sup\{k: \overline{p}_{k,c, \Gamma} > \alpha\}$;
    \item the intersection of all $1-\alpha$ confidence sets for $\tau_{(k)}$'s, viewed as a confidence set for $\bs{\tau}$, is the same as that for $n(c)$'s, 
    and they have the following equivalent forms:
    \begin{align*}
        \big\{ \bs{\delta}: \delta_{(k)} \in \overline{\mathcal{I}}_{\alpha}(k), 1\le k \le N \big\}
        & = 
        \Big\{
        \bs{\delta}: 
        \sum_{i=1}^N \I(\delta_i > c) \in \overline{\mathcal{S}}_{\alpha}(c), c \in \mathbb{R}
        \Big\}
        =  \bigcap_{k, c: \overline{p}_{k,c, \Gamma} \le \alpha} 
        \mathcal{H}_{k,c}^{\c}.
    \end{align*}
    Moreover, the resulting confidence set covers the true individual treatment effect vector $\bs{\tau}$ with probability at least $1-\alpha$, i.e., 
    $
        \Pr(
        \bs{\tau} \in \bigcap_{k, c: \overline{p}_{k,c, \Gamma} \le \alpha} 
        \mathcal{H}_{k,c}^{\c}
        ) 
        \ge 1-\alpha. 
    $
\end{enumerate}
\end{theorem}

\begin{remark}\label{rmk:sen_inv_finite_LP}
Theorem \ref{thm:sen_inv_finite} also holds for the easier-to-calculate $p$-value $\overline{p}_{k,c, \Gamma}^{\LP}$. 
\end{remark}
\section{Details for Algorithms Used to  Calculate Valid $p$-values}\label{sec:eff_alg}

In this section we introduce the details for the algorithms used to solve the integer programming in \eqref{eq:integer_program} and the relaxed linear programming in \eqref{eq:linear_program}, which can then provide valid $p$-values for testing null hypotheses on quantiles of individual treatment effects. 
We first introduce the piecewise-linear optimization for the integer programming in \eqref{eq:integer_program}, which can be solved by the Gurobi optimizer. 
We then introduce two algorithms adopted from the literature on the multiple-choice knapsack problem. 
The first is a greedy algorithm that can efficiently solve the relaxed linear programming in \eqref{eq:linear_program} in linear time, 
and the second is the dynamic programming algorithm that can solve the integer programming in \eqref{eq:integer_program} in polynomial time. 

\subsection{Piecewise-linear optimization}\label{sec:pwl}

The integer programming problem in \eqref{eq:integer_program} can be equivalently formulated as a piecewise-linear optimization problem. 
Specifically, 
for each $1\le s\le S$ and $0\le x \le n_s$, define 
\begin{align*}
    f_{s,c}(x) = \I(x=0) t_{s,c}(0) + \sum_{j=0}^{n_s-1} \I(j< x \le j+1) 
    [t_{s,c}(j)+\{t_{s,c}(j+1) - t_{s,c}(j)\} (x-j) ]. 
\end{align*}
The minimum objective value from \eqref{eq:integer_program} is then equivalently 
\begin{align}\label{eq:pwl}
    \min_{l\in \mathcal{K}_{\bs{n}}(N-k)} \sum_{s=1}^S f_{s,c}(l_s) 
    = 
    \min_{l\in \mathcal{R}_{\bs{n}}(N-k)} \sum_{s=1}^S f_{s,c}(l_s), 
\end{align}
where $\mathcal{K}_{\bs{n}}(N-k)$ is defined as in \eqref{eq:setKn}, 
and
$\mathcal{R}_{\bs{n}}(N-k)$ is defined the same as $\mathcal{K}_{\bs{n}}(N-k)$ but without the integer constraints. 
Besides, the objective function in \eqref{eq:pwl} can be non-convex. 
In \S \ref{sec:proof_pwl}, 
we prove the equality in \eqref{eq:pwl} and give a numerical example showing the non-convexity of the objective function. 
We also investigate the performance of Gurobi optimizer for this piecewise-linear optimization in \S \ref{sec:comp_cost}.



\subsection{Greedy Algorithm for solving the relaxed linear programming problem}\label{sec:greedy}


We first introduce some notation to give an equivalent form for the optimization in \eqref{eq:equiv_min_test_stat}. 
By definition, 
we can verify that, 
for each $s$, 
$t_{s, c}(0) \ge t_{s, c}(1) \ge \ldots \ge t_{s, c}(n_s)$. 
Define 
$
\Delta_{s, c}(j) \equiv t_{s, c}(j-1) - t_{s, c}(j) \ge 0 
$
for all $s, j$. 
Then 
the optimization in \eqref{eq:equiv_min_test_stat} reduces to
\begin{align}\label{eq:test_stat_inf_equiv}
    \inf_{\bs{\delta} \in \mathcal{H}_{k,c}} t(\bs{Z},\bs{Y} - \bs{Z}\circ \bs{\delta}) 
    & = 
    \sum_{s=1}^S t_{s, c} (0) - 
    \max_{ (l_1, \ldots, l_S) \in \mathcal{K}_{\bs{n}}(N-k)}  
    \sum_{s=1}^S \sum_{j=1}^{l_s} \Delta_{s, c}(j). 
\end{align}
Consequently, to get the valid $p$-value $p_{k,c}$, 
it suffices to maximize 
$\sum_{s=1}^S \sum_{j=1}^{l_s} \Delta_{s, c}(j)$ 
over $(l_1, \ldots, l_S) \in \mathcal{K}_{\bs{n}}(N-k)$.  
The naive greedy algorithm solves this maximization by making locally optimal choice at each stage when $N-k$ increases from $0$ to $N$. 
However, this may lead to sub-optimal solutions; see 
{\lxr \S \ref{sec:numerical_example}}
for a numerical example. More importantly, it may invalidate the resulting $p$-value.

The naive greedy algorithm fails mainly because $\Delta_{s,c}(j)$ can increase as $j$ increases, 
and thus 
searching only the local optimum will miss the correct solution. 
To overcome 
this drawback, 
we propose to transform the sequence of $\Delta_{s,c}(j)$'s for each $s$ such that the resulting sequence is monotone decreasing and its cumulative sums dominate that of the original sequence. 
We formally define such a transformation below. 

\begin{definition}\label{def:mono_transform}
A function $\Psi(\cdot)$ that maps a vector to a vector of the same length is called a
monotone dominating transformation, 
if for any $m\ge 1$ and any vector $\bs{a} = (a_1, \ldots, a_m)^\top \in \mathbb{R}^m$, 
the transformation 
$\Psi(\bs{a}) = (\Psi_1(\bs{a}), \ldots, \Psi_m(\bs{a}))^\top \in \mathbb{R}^m$
satisfies that 
(i) $\Psi_1(\bs{a}) \ge \Psi_2(\bs{a}) \ge \ldots \ge  \Psi_m(\bs{a})$ 
and (ii) 
$\sum_{i=1}^j \Psi_i(\bs{a}) \ge \sum_{i=1}^j a_i$ for $1\le j \le m$. 
\end{definition}

Let 
$\bs{\Delta}_{s,c}(\overline{j}) = (\Delta_{s,c}(1), \ldots, \Delta_{s,c}(j))$ for $1\le s \le S$, $1\le j \le n_s$ and $c\in \mathbb{R}$. 
From the two properties in Definition \ref{def:mono_transform}, 
for any monotone dominating transformation $\Psi(\cdot)$, 
\begin{align}\label{eq:Psi_monotone}
    \Psi_1\left( \bs{\Delta}_{s,c}
    \left( \overline{n_{sk}} \right) \right) 
    \ge 
    \Psi_2\left( \bs{\Delta}_{s,c}\left(\overline{n_{sk}} \right) \right) 
    \ge 
    \ldots
    \ge 
    \Psi_{n_{sk}}\left( \bs{\Delta}_{s,c}\left(\overline{n_{sk}} \right) \right),  
    \quad (1\le s\le S)
\end{align}
and 
for any $(l_1, \ldots, l_S) \in \mathcal{K}_{\bs{n}}(N-k)$, 
\begin{align}\label{eq:Psi_dominate}
    \sum_{s=1}^S \sum_{j=1}^{l_s} \Delta_{s, c}(\overline{j})
    \le 
    \sum_{s=1}^S \sum_{j=1}^{l_s} 
    \Psi_j\left( \bs{\Delta}_{s,c}\left( \overline{n_{sk}} \right) \right). 
\end{align}
Importantly, 
with the transformed $\Delta_{s,c}(j)$'s, 
\eqref{eq:Psi_monotone} guarantees that the naive greedy algorithm can achieve the global maximum of the objective function on the right hand side of \eqref{eq:Psi_dominate} over $
l \in \mathcal{K}_{\bs{n}}(N-k)$, 
and \eqref{eq:Psi_dominate} guarantees that the achieved global maximum must be 
no less than
the target maximum of the objective function on the left hand side of \eqref{eq:Psi_dominate}. 
This further helps provide conservative but still valid $p$-value for testing $H_{k,c}$. 

The remaining question is how to optimally construct a monotone dominating transformation.
It turns out the optimal transformation has a simple form and is easy to compute. 
Define $\tilde{\Psi}: \bs{a} = (a_1, \ldots, a_m)^\top \rightarrow 
\tilde{\Psi}(\bs{a}) = (\tilde{\Psi}_1(\bs{a}), \ldots, \tilde{\Psi}_m(\bs{a}))$ recursively as:
\begin{align}\label{eq:Psi_optimal}
    \tilde{\Psi}_i(\bs{a})
    & = 
    \begin{cases}
    \max_{1\le j \le m}j^{-1}\sum_{t=1}^j a_t, 
    & \text{for } i = 1, \\
    \max_{i \le j \le m} (j-i+1)^{-1} {\rev \left\{ \sum_{t=1}^j a_t - \sum_{t=1}^{i-1} \tilde{\Psi}_t(\bs{a}) \right\}}, 
    & 
    \text{for } i=2, \ldots, m. 
    \end{cases}
\end{align}
The following proposition establishes the optimality of $\tilde{\Psi}(\cdot)$.  

\begin{proposition}\label{prop:Psi_opt}
The mapping $\tilde{\Psi}(\cdot)$ defined in \eqref{eq:Psi_optimal} is a monotone dominating transformation as in Definition \ref{def:mono_transform}. 
Moreover, for any $m\ge 1$, $\bs{a}\in \mathbb{R}^m$ and any monotone dominating transformation $\Psi(\cdot)$, 
$\sum_{i=1}^j \tilde{\Psi}_i(\bs{a}) \le \sum_{i=1}^j \Psi_i(\bs{a})$ for all $1\le j \le m$. 
\end{proposition}




We are now ready to describe the greedy algorithm with the optimal monotone dominating transformation $\tilde{\Psi}(\cdot)$ in \eqref{eq:Psi_optimal}. 
Surprisingly at the first glance and as discussed in detail later, the greedy algorithm actually solves the linear programming problem in \eqref{eq:linear_program}.
We summarize the algorithm below.

\begin{algo}\label{alg:greedy}
Greedy algorithm for the linear programming problem in \eqref{eq:linear_program}.
\begin{tabbing}
\text{Input}: observed data $(Y_s, Z_s)$ for all $s$, the null hypothesis of interest  $H_{k,c}$, and values of $t_{s,c}(j)$
\\ 
\qquad \ \ \ 
and $\Delta_{s,c}(j)$ for stratum $1\le s\le S$ and $1\le j \le n_{sk}$; 
\\
\text{For} each stratum $s$, perform the transformation $\tilde{\Delta}_{sk,c}(\overline{n_{sk}}) \equiv \tilde{\Psi} (\Delta_{s,c}(\overline{n_{sk}}))$;\\
\text{Pool} all the transformed elements into $\mathcal{T}_{k, c} \equiv \{\tilde{\Delta}_{sk,c}(j): 1\le j \le n_{sk}, 1\le s\le S \}$;\\
\text{Output}: 
$
    t_{k,c}^{\LP} 
    = 
    \sum_{s=1}^S t_{s, c} (0) -  
    \text{sum of the largest $N-k$ elements of } \mathcal{T}_{k, c} .
$
\end{tabbing}
\end{algo}

Below we give several remarks regarding Algorithm \ref{alg:greedy}.
First, 
the final solution has a simple form, involving only sorting and summing elements of $\mathcal{T}_{k, c}$. This is due to the property (i) in Definition \ref{def:mono_transform} for the monotone dominating transformation.
As commented in \S \ref{sec:comp_greedy}, 
the computational complexity of the greedy algorithm 
is $O(N\log N+N\max_s n_s)$, 
while that for solving a general linear programming problem involving $N$ variables is at least $O(N^{2+\epsilon})$ with the latest improvement of $\epsilon=1/18$ \citep{2020arXiv200407470J}.

Second, 
the greedy algorithm here actually solves the linear programming problem in \eqref{eq:linear_program}; we give a proof in 
{\lxr \S \ref{sec:proof_greedy}.}
As demonstrated in 
\S \ref{sec:proof_greedy}, 
for each stratum $s$, 
$\{(j, \sum_{i=1}^j \tilde{\Delta}_{sk,c}(i) ): 0 \le j \le n_{sk}\}$ from the optimal transformation in \eqref{eq:Psi_optimal}
actually forms the upper convex hull of $\{(j, \sum_{i=1}^j \Delta_{s,c}(i) ): 0 \le j \le n_{sk}\}$. 
Consequently, 
the greedy algorithm we introduce here is essentially equivalent to the greedy algorithm for the linear programming relaxation of multiple-choice knapsack problems \citep[][Page 320]{kellerer2004multiple}.
Moreover, the computational complexity of the greedy algorithm can be further reduced to $O(N)$ by employing algorithms in \citet{dyer1984n} and  \citet{zemel1984n}. 
Our R package implements the $O(N)$ algorithm.

Third, 
when the rank transformations $\phi_s(\cdot)$'s in \eqref{eq:strat_rank_sum} is a concave function, the sequence $\{\Delta_{s,c}(j)\}$ will itself be decreasing in $j$, and it will be invariant under the optimal transformation $\tilde{\Psi}(\cdot)$. 
Thus, 
for stratified rank score statistics with concave transformations, 
which include the stratified Wilcoxon rank sum statistic as a special case, 
the greedy algorithm solves exactly the integer linear programming in \eqref{eq:integer_program}. 
However, the rank transformations in stratified Stephenson rank sum statistics are generally not concave. Instead, they are always convex, under which 
the monotone dominating transformation 
is 
necessary to ensure the validity of the resulting $p$-values.  
We relegate the technical details to \S \ref{sec:greedy_concave}. 

\subsection{Dynamic programming for solving the integer linear programming problem}\label{sec:dynamic}

Here we adopt the dynamic programming algorithm that can solve the multiple-choice knapsack problem in pseudopolynomial time as shown by \citet{DUDZINSKI19873}; see also \citet[][Page 329--331]{kellerer2004multiple}. 
Similar to \S \ref{sec:greedy}, we consider the equivalent form in \eqref{eq:test_stat_inf_equiv} and aim to maximize $\sum_{s=1}^S \sum_{j=1}^{l_s} \Delta_{s, c}(j)$ over $(l_1, \ldots, l_S) \in \mathcal{K}_{\bs{n}}(N-k)$. 
For $1 \le \strata \le S$ and nonnegative integer $d$, 
define $\dypr_\strata(d)$ as the maximum value of the objective function but restricted to the first $\strata$ strata and with the constraint that the sum of $l_s$'s is bounded by $d$, i.e., 
\begin{align}\label{eq:Dp}
    \dypr_\strata(d) = \max
    \Big\{\sum_{s=1}^\strata \sum_{j =1}^{l_s}\Delta_{s,c}(j): 
    (l_1, \ldots, l_{\strata})\in \mathbb{Z}^{\strata}, 0\le l_s \le n_s \text{ for all } s, 
    \sum_{s=1}^\strata l_s \le d
    \Big\}. 
\end{align} 
Equivalently, $\sum_{s=1}^\strata t_{s,c}(0) - \dypr_\strata(d)$ is the minimum value of the rank sum test statistic for the first $\strata$ strata under the constraint that there are at most $d$ units with individual effects greater than $c$. 
We further define $\dypr_{0}(d) = 0$ for all nonnegative integer $d$. 
We can verify that $\dypr_\strata(d)$'s satisfy the following recursive formula:
for $1\le \strata \le S$ and nonnegative integer $d$, 
\begin{align}\label{eq:Dp_recur}
    \dypr_\strata(d) = \max
    \big\{\dypr_{\strata-1}(d -i)+\sum_{j=1}^i\Delta_{\strata,c}(j): 
    i=0,1, \ldots, \min\{n_b, d\} \big\}. 
\end{align} 
Importantly, $\dypr_{S}(N-k)$ gives the maximum value of the objective function on the right hand side of \eqref{eq:test_stat_inf_equiv}, which equivalently provides the optimal objective value $t_{k,c}$ for the integer linear programming in \eqref{eq:integer_program}. 
We summarize the algorithm below. 

\begin{algo}\label{alg:dynamic}
Dynamic programming for the integer linear programming problem in \eqref{eq:integer_program}. 
\vspace*{-6pt}
\begin{tabbing}
\text{Input}: observed data $(Y_s, Z_s)$ for all $s$, the null hypothesis of interest  $H_{k,c}$, and values of $t_{s,c}(j)$ 
\\ 
\qquad \ \ \ 
and $\Delta_{s,c}(j)$ for stratum $1\le s\le S$ and $1\le j \le n_{sk}$;
\\
Initialize $\dypr_{0}(0) = \dypr_{0}(1) = \ldots = \dypr_{0}(N-k) = 0$;
\\
\text{For} $s=1$ to $s=S$\\
\qquad Calculate $\dypr_{s}(0), \dypr_{s}(1), \ldots, \dypr_{s}(N-k)$ using the recursive formula in \eqref{eq:Dp_recur};
\\
\text{Output}: 
$
    t_{k,c} 
    = 
    \sum_{s=1}^S t_{s, c} (0) -  \dypr_S(N-k).
$
\end{tabbing}
\end{algo}

In Algorithm \ref{alg:dynamic}, 
at each iteration and for calculating each $m_s(d)$, we calculate at most $n_s+1$ cumulative sums of $\Delta_{s,c}(j)$ and perform at most $n_s+1$ summations to complete the recursion in \eqref{eq:Dp_recur}. 
Therefore, the computational complexity of the dynamic programming is at most of order 
$(N-k)\sum_{s=1}^S (n_s+1) \le 2 (N-k) N \le 2N^2$. 
This indicates that we can solve exactly the integer programming in \eqref{eq:integer_program} in polynomial time, with exponent of the sample size being at most $2$. 
Note that a general multiple-choice knapsack problem is NP-hard. 
Here we have a polynomial-time algorithm because the total cost in our problem in \eqref{eq:integer_program} is always bounded by $N$. 




It is also worth noting that the dynamic programming with $k=0$ calculates simultaneously $t_{j,c}$ for all $0\le j \le N$, 
which then leads to $p$-values $p_{j, c}$ for all $0\le j \le N$. 
As discussed in Theorem \ref{thm:conf_set}(ii), these $p$-values immediately provide confidence sets for $n(c)$, i.e., the number of units with treatment effects greater than $c$.



\section{Simulation studies}\label{sec:simu_study}
\subsection{Computation cost for getting the valid $p$-values}\label{sec:comp_cost}

We consider a SCRE with $S$ ($=100, 500, 1000, 2000, 3000$) strata of equal size $n$ ($=50, 100, 200$), where each stratum has half of its units assigned to treatment. 
We use the stratified Stephenson rank sum statistic with 
$h_s = 6$ for all $s$.
We simulate the potential outcomes $Y_{si}(1)=Y_{si}(0)$ for all $s$ and $i$ as i.i.d.\ samples from the standard Gaussian distribution. 
We consider testing the null hypothesis on the $90\%$ quantile of individual treatment effects, 
or more precisely, 
$H_{k,0}: \tau_{(k)} \le 0$ with $N=nS$ and $k=0.9N$.
{
We consider the following algorithms for achieving the minimum test statistic value $t_{k,0}$ or $t_{k,0}^{\LP}$ from the integer linear programming (ILP) in \eqref{eq:integer_program} or the linear programming (LP) in \eqref{eq:linear_program}: 
\begin{enumerate}[label=(\arabic*)]
    \item the greedy algorithm for the LP in \eqref{eq:linear_program} (denoted by LP-GT) as described in \S \ref{sec:greedy}, 
    
    \item the dynamic programming for the ILP in \eqref{eq:integer_program} (denoted by ILP-DP) as described in \S \ref{sec:dynamic},  
    
    \item Gurobi optimizer for the LP in \eqref{eq:linear_program} (denoted by LP-Gurobi),
    
    \item Gurobi optimizer for the ILP in \eqref{eq:integer_program} (denoted by ILP-Gurobi),
    
    \item Gurobi optimizer for the piecewise linear optimization with integer constraints  (denoted by IPWL-Gurobi) as described in \S \ref{sec:pwl},
    
    \item Gurobi optimizer for the piecewise linear optimization without the integer constraints (denoted by PWL-Gurobi) as described in \S \ref{sec:pwl}.  
\end{enumerate}
{ 
For the integer programming, 
our Gurobi programming involves 
{\lxr $N+S$}
binary integer variables and  $S+1$ additional linear constraints. 
For the piecewise linear optimization without the integer constraints, 
our Gurobi programming involves $S$ variables with given lower and upper bounds, 
$1$ additional linear constraint, 
and $S$ piecewise linear constraints. 
}

Table \ref{tab:comp_cost} shows the run time of the six algorithms 
under different choices of $(n,S)$, taking median over 100 simulated datasets, where 
we exclude the time for getting $t_{s,0}(l)$'s for all $s$ and $l$.
{
It should be noted that the Gurobi optimizer incorporates an error tolerance mechanism, which implies that the solutions provided by Gurobi are accurate within certain predetermined numeric tolerance.
In contrast, both the greedy algorithm and dynamic programming solve the corresponding linear and integer programming problems exactly, subject to the level of numerical precision.  From our simulation, the average absolute values of the relative differences of LP-GT, LP-Gurobi, ILP-Gurobi, IPWL-Gurobi and PWL-Gurobi from ILP-DP, scaled by $10^4$, are, respectively, $3.94$, $3.94$, $3.28$, $3.50$ and $3.32$. These show that the 
solutions 
from the relaxed linear programming are comparable to that from Gurobi optimizations for the integer programming.

Table \ref{tab:comp_cost} shows the run time of the previously listed six algorithms. 
First, LP-GT is much faster than the other five algorithms. Second, LP-Gurobi and ILP-Gurobi take about the same time, and both of them are  much faster than IPWL-Gurobi and PWL-Gurobi. Third, 
the dynamic programming performs well with moderate sample sizes, but its computation time increases significantly as the sample size grows.
However, it is worth mentioning that a single run of the dynamic programming can efficiently calculate the $p$-values $p_{0, c}, p_{1, c}, \ldots,$ and $p_{N,c}$ at the same time, which will then provide confidence sets for $n(c)$. 
Therefore, it can be preferred when we are particularly interested in the number, {\rev or equivalently the proportion}, of units with effects greater than a given threshold; see also the discussion at the end of \S \ref{sec:dynamic}. 
}

}
\begin{table}[htb]
    \centering
    \caption{Run time {\rev in seconds} of the six algorithms, 
    linear programming (LP) solved by greedy algorithm with the optimal transformation (GT) and Gurobi,
    integer linear programming (ILP) solved by dynamic programming (DP) and Gurobi,  and piecewise-linear (PWL) optimization with and without integer constraints solved by Gurobi, for testing null hypothesis 
    on the $90\%$ quantile of individual effects under various values of $(n,S)$, 
    taking median over 100 simulated datasets.
    The simulation is conducted using AMD Ryzen 5 5600X 3.70 GHz processor.}\label{tab:comp_cost}
    \resizebox{0.9\columnwidth}{!}{%
    \begin{tabular}{llrrrrrr}
    \toprule 
    Stratum size & Strata number & LP-GT&LP-Gurobi& ILP-DP& ILP-Gurobi&IPWL-Gurobi&PWL-Gurobi\\
    \midrule
       & $S=100$&0.00& 0.01&  0.01& 0.03&  0.25&  0.24\\
       & $S=500$&0.00& 0.06&  0.28& 0.11&  0.73&  0.58\\
      $n=50$ & $S=1000$&0.01& 0.17&  1.14& 0.22&  1.78&  1.23\\
       & $S=2000$&0.02& 0.44&  4.70& 0.51&  5.85&  3.00\\
       & $S=3000$&0.03& 0.88& 10.58& 0.96& 11.49&  5.66\\
       \midrule
       & $S=100$&0.00& 0.02&  0.03& 0.04&  0.25&  0.24\\
       & $S=500$&0.01& 0.12&  0.72& 0.20&  1.66&  1.44\\
      $n=100$ & $S=1000$&0.02& 0.31&  3.00& 0.41&  3.03&  2.68\\
       & $S=2000$&0.05& 1.01& 12.23& 1.05&  8.69&  6.75\\
       & $S=3000$&0.08& 2.08& 27.60& 2.02& 39.18& 47.03\\
       \midrule
       & $S=100$&0.01& 0.03&  0.08& 0.08&  0.73&  0.68\\
       & $S=500$&0.03& 0.22&  2.21& 0.37&  3.88&  5.58\\
      $n=200$ & $S=1000$&0.07& 0.61&  9.06& 0.93& 11.38&  7.88\\
       & $S=2000$&0.14& 1.93& 36.59& 2.30& 43.44& 50.54\\
       & $S=3000$&0.22& 3.84& 82.78& 4.39& 41.93& 43.13\\
       \bottomrule
    \end{tabular}%
    }
\end{table}
{
We further consider a large dataset from \citet{Lalive2006} on studying the change in unemployment benefits in Austria. 
Similar to \citet{yu2022graded}, we focus on men who were not temporarily laid off, 
pooling the three groups with increased benefits.  
We consider the duration of unemployment as the outcome of interest, and conduct 
$1:6$ 
matching \citep{Ho2011} based on the following covariates: age, wage in prior job, an indicator of at least 3 years of work in the past 5 years, whether the job was an apprenticeship, married or not, divorced or not, education in three levels, whether or not the previous job was a blue collar job, a seasonal job, a manufacturing job. 
{
We allow control units to be matched with multiple treated units. 
The in total 55619 units are then divided into 22111 matched sets, each of which contains 1 treated unit and 6 control units. 
We then construct $80\%$ confidence intervals for the $95\%$, $90\%$, $85\%$, $80\%$ quantiles of individual treatment effects using our approach in \S \ref{sec:inv_test_conf_set} with the Stephenson rank sum statistic and $h_s = 5$ for all matched sets, 
pretending that the matched observational study is a SCRE; we can also conduct sensitivity analysis as in \S \ref{sec:sen_ana}. }
We compare all the algorithms in Table \ref{tab:comp_cost}. 
Table \ref{tab:comp_cost_large} shows the computation time of the algorithms for each of these quantiles, excluding the time for approximating the null distribution of the rank statistic, {\rev which needs to be done only once and can be shared for all quantiles, takes $14.19$ minutes here using $10^5$ Monte Carlo draws, and can also be approximated using a normal approximation as in \S \ref{sec:large_sen}.} 
{From Table \ref{tab:comp_cost_large},  
first, 
the greedy algorithm takes the shortest run time, while IPWL-Gurobi takes the longest run time. 
In particular, the IPWL-Gurobi takes about 1 hour to get the confidence interval for the $80\%$ quantile of individual effects, 
while the greedy algorithm takes about $1$ minute.
Second,
ILP-DP, PWL-Gurobi and IPWL-Gurobi tend to 
take longer run time when we consider smaller quantiles of individual effects, while the run time of the other three algorithms is quite stable across all quantiles.

From the above, the greedy algorithm can be preferable when analyzing large datasets or using statistics with concave rank transformation {\rev for which the greedy algorithm is also exact for the integer programming.}  
When we are particularly interested in the number of units with effects passing a given threshold, 
the dynamic programming algorithm can be preferred. 
In other general cases, we will suggest the Gurobi optimization for the integer programming. 
}
\begin{table}[htb]
    \centering
    \caption{Run time {\rev in minutes} of six algorithms in Table \ref{tab:comp_cost} for obtaining $80\%$ confidence intervals of the $95\%$, $90\%$, $85\%$, $80\%$ quantiles of individual treatment effects.
    }\label{tab:comp_cost_large}
    \begin{tabular}{lrrrrr}
    \toprule 
      Quantile  &$95\%$& $90\%$& $85\%$& $80\%$\\
    \midrule
        LP-GT& 1.13& 1.07& 1.08 &1.07\\
        LP-Gurobi&2.53&2.41&2.43&2.46\\
        ILP-DP& 4.54&  8.04&  11.49 & 15.40\\
        ILP-Gurobi&2.76&2.61&2.60&2.58\\
        PWL-Gurobi&7.18& 5.68& 26.76 &27.42\\
        IPWL-Gurobi&23.98&24.94&83.09&64.01\\
    \bottomrule
    \end{tabular}
\end{table}

}
\subsection{Choice of test statistics}

We consider a SCRE with $S = 200$ strata of equal size $n$ ($=10, 20, 50$), and randomly assign half of the units within each stratum to treatment. 
We simulate the potential outcomes as i.i.d.\ samples from the following model: 
\begin{align}\label{eq:simu_gene}
    Y_{si}(0) \sim \mathcal N(0,1),\quad Y_{si}(1)=\tau+Y_{si}(0)+\sigma\epsilon_{si},\quad \epsilon_{si}\sim\mathcal N(0,1),
    \quad \epsilon_{si}\perp \!\!\! \perp Y_{si}(0),
\end{align}
where we choose $\tau$ to be $0$ or $1$. 
Note that all strata have equal sizes.
We consider stratified Stephenson rank sum statistics 
with $h_1= \ldots=h_S=h$ for some $h\ge 2$.  
To compare the power under various choice of $h$, 
we focus on inferring the number of units with positive effects $n(0)$, 
and use the greedy algorithm to calculate the valid $p$-values. 


\begin{figure}[h]
    \centering
    \begin{subfigure}{\textwidth}
    \centering
    \includegraphics[width=0.6\textwidth]{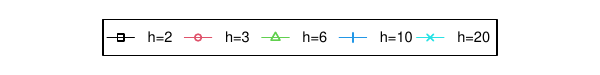}    
    \end{subfigure}
    \subfloat[$n=10$]{\includegraphics[width=0.25\textwidth]{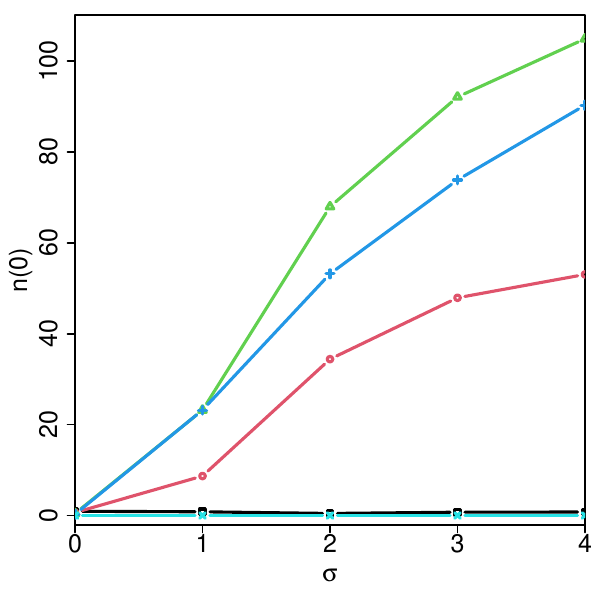}}\hfil 
    \subfloat[$n=20$]{\includegraphics[width=0.25\textwidth]{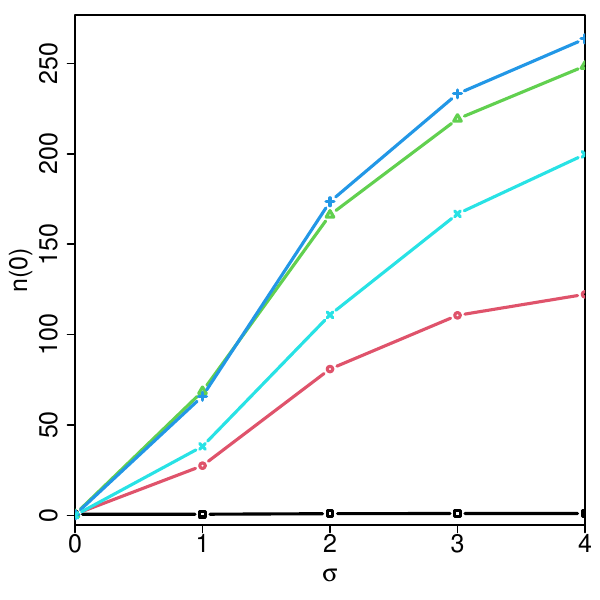}} \hfil
    \subfloat[$n=50$]{\includegraphics[width=0.25\textwidth]{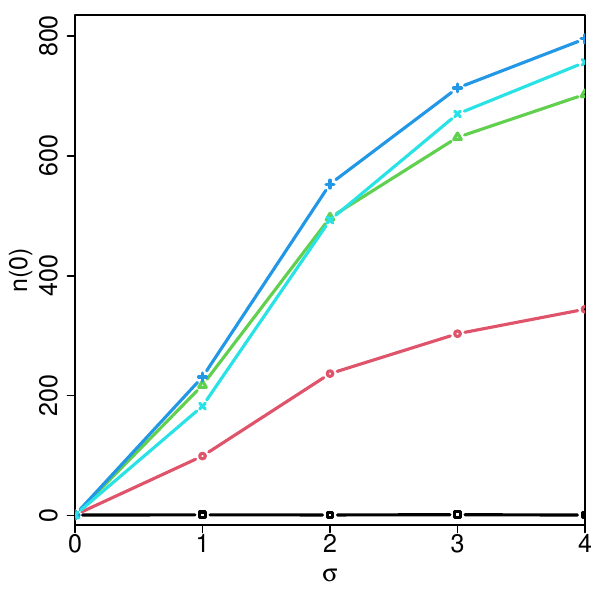}} 
    
    \subfloat[$n=10$]{\includegraphics[width=0.25\textwidth]{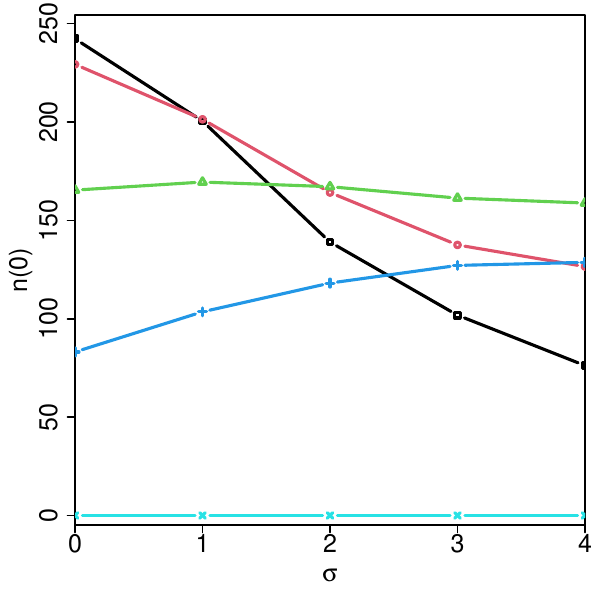}}\hfil
    \subfloat[$n=20$]{\includegraphics[width=0.25\textwidth]{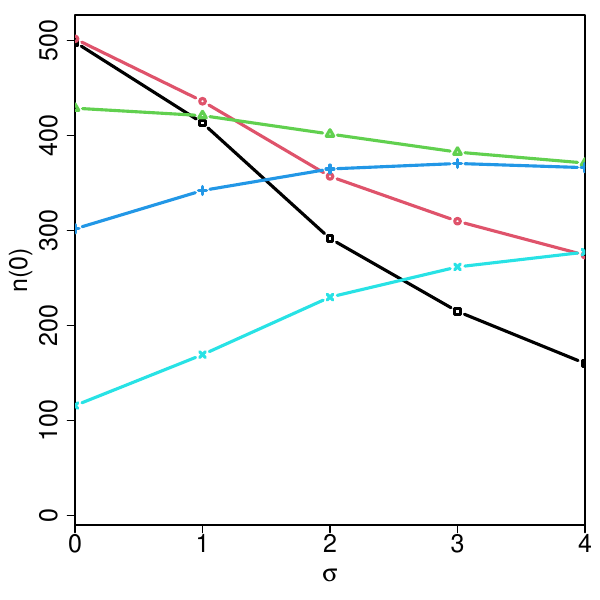}} \hfil
    \subfloat[$n=50$]{\includegraphics[width=0.25\textwidth]{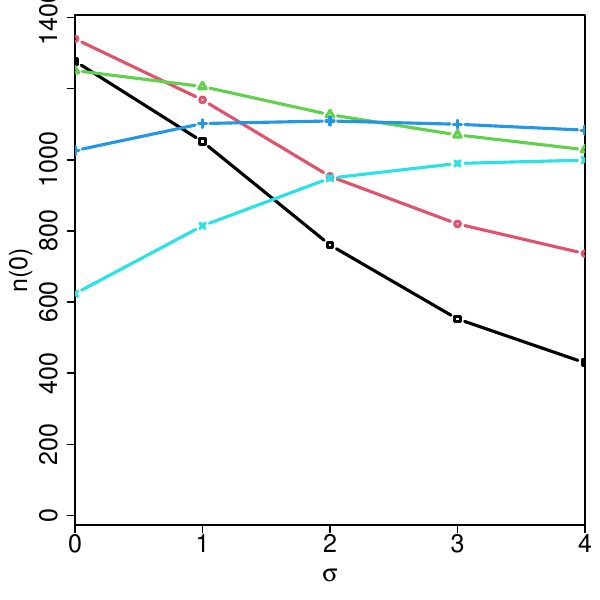}} 
\caption{
$90\%$ lower confidence limits for $n(0)$ under various Stephenson rank statistics and strata sizes, averaging over 100 simulations. 
The top and bottom rows are for $\tau =0$ and 1, respectively. }
\label{fig:tau=2}
\end{figure}
Figure \ref{fig:tau=2} shows the $90\%$ lower confidence limits of $n(0)$ using various  stratified Stephenson rank sum statistics under different strata sizes, averaging over 100 simulated datasets, 
when the expected average effect $\tau$ equals 0 and 1. 
Note that when $h > n$, the Stephenson ranks will become constant zero and the resulting test will lose power. 
We consider only the choice of $h$ such that $h<n$. 
From Figure \ref{fig:tau=2}, moderately large $h$ compared to $n$ can help improve the power for inferring $n(0)$. 
For example, when $n=50$, using the Stephenson rank with $h=6$ can provide 
more informative
lower confidence limit for $n(0)$ on average than using the Stephenson rank with $h=2$, {\rev which reduces to the Wilcoxon rank.} 
However, when $h$ is overly large compared to $n$, it can deteriorate the power of our inference. 
In practice, we suggest to use the Stephenson rank with moderate parameter compared to the stratum size. 
We leave a more detailed theoretical investigation for future study. 

\subsection{Treatment label switching  for matched observational study}

We consider a matched observational study with $S=200$ matched sets, each of which contains 1 treated unit and 9 control units. 
We generate the potential outcomes as i.i.d.\ samples from model \eqref{eq:simu_gene} with 
$(\tau, \sigma) = (1,0)$, 
i.e., the treatment has a constant effect 1. 
Moreover, 
following \citet{Rosenbaum02a}, 
we consider the favorable situation where there is no hidden confounding, 
i.e., the study is essentially a SCRE. 
In this favorable situation with nonzero treatment effects, we wish to detect significant treatment effects and hope the detected effects to be insensitive to hidden confounding. 

To illustrate the power gain from treatment label switching, we consider testing the null hypothesis $H_{k,0}:\tau_{(k)}\le 0$ with $N=2000$ and $k=0.95N$ under sensitivity models. 
In particular, we consider using either the original data or that with switched treatment labels and changed outcome signs. 
We use the greedy algorithm to calculate the valid $p$-values. 
Table \ref{tab:power_switch} shows the empirical power of the test using the stratified Stephenson rank sum statistic with $h_1 = \ldots = h_S = 4$ under various sensitivity models, based on 100 simulated datasets. 
\begin{table} 
\caption{Empirical power of two tests regarding the $90\%$ quantile of individual effects under various sensitivity models, averaging over 100 simulated datasets. One test uses the original data, while the other uses the data with switched treatment labels and changed outcome signs. }
\label{tab:power_switch}
\centering
    \begin{tabular}{ccccccccccc}
    \toprule
      $\Gamma$ &1&1.1&1.2&1.3&1.4&1.5& 1.6&1.7&1.8\\
    \midrule
    Switching labels &0.97&0.91&0.7&0.44&0.20&0.10&0.04&0.02&0.01\\
    No switching & 0.00 & 0.00  &0.00 &0.00 &0.00 &0.00 &0.00&0.00&0.00\\
    \bottomrule
    \end{tabular}%
\end{table}
From Table \ref{tab:power_switch},
switching labels can greatly increase the power of the test. When $\Gamma = 1$, i.e., there is no hidden confounding and the study reduces to a SCRE, the power of the test after label switching  is almost 1, while that using the original data is about 0. 
As $\Gamma$ increases, 
we allow more amount of 
biases
in the treatment assignment, 
and thus the power of the test decreases.
Importantly, 
label switching still provides significant gain in power when $\Gamma>1$, making the study more robust to unmeasured confounding.  
Thus, 
for general matched observational studies with one treated and multiple control units within each matched set, 
we suggest to perform label switching as discussed in Remarks \ref{rmk:switch} and \ref{rmk:switch_sen}, 
which can bring significant power gain as demonstrated by Table \ref{tab:power_switch}.


\section{Proof of Propositions, Theorems and Remarks}\label{sec:proof_prop_thm}

\subsection{Proof of Proposition \ref{prop:dist_free}}

Under the SERE, $\bs Z_1$, $\bs Z_2$, $\ldots$, $\bs Z_S$ are mutually independent, so it suffices to prove that $t_s(\bs{Z}_s,\bs{y}_s)$ is distribution free within each stratum $s$, for $1\le s \le S$.
For any $\bs{y}_s,\bs{y}_s' \in \mathbb{R}^{n_s}$, let $r_{si}=\rank_i(\bs{y_s})$ and $r_{si}'=\rank_i(\bs{y}_{s}')$ for $1\le i \le n_{s}$. 
By the property of the stratified rank score statistic in Definition \ref{def:strat_rank_score}, 
both 
$\bs{r}_{s} = (r_{s1}, r_{s2}, \ldots, r_{sn_s})$ and $\bs{r}'_{s} = (r'_{s1}, r'_{s2}, \ldots, r'_{sn_s})$ are permutations of $\{1, 2, \ldots, n_s\}$. 
Thus, there must exist two permutations $\pi(\cdot)$ and $\pi'(\cdot)$ of $\{1, 2, \ldots, n_s\}$ such that 
$r_{s\pi(i)} = r'_{s\pi'(i)} = i$ for $1\le i \le n_{s}$. Consequently, we have  $t_s(\bs{Z}_s,\bs{y}_s)=\sum_{i=1}^{n_s}Z_{si}\phi_{s}(r_{si})=\sum_{i=1}^{n_s}Z_{s\pi(i)}\phi_{s}(i)$,  and by the same logic $t_s(\bs{Z}_s,\bs{y}_s')=\sum_{i=1}^{n_s}Z_{s\pi'(i)}\phi_{s}(i)$. 
By Definition \ref{def:ERBE} of the SERE, $(Z_{s\pi(1)}, \ldots, Z_{s\pi(n_s)})$ follows the same distribution as $(Z_{s\pi'(1)}, \ldots, Z_{s\pi'(n_s)})$. 
Thus, $t_s(\bs{Z}_s,\bs{y}_s) \sim t_s(\bs{Z}_s,\bs{y}'_s)$.
From the above, Proposition \ref{prop:dist_free} holds.


\subsection{Proof of Theorem \ref{thm:pNkc}}
From \citet[][Theorem 8.3.27]{casella2002statistical}, 
$p_{k,c} = \sup_{\bs{\delta} \in \mathcal{H}_{k,c}} p_{\bs{Z}, \bs{\delta}}$ is a valid $p$-value for testing $H_{k,c}$. 
Below we prove its equivalent forms in \eqref{eq:pNkc}. 
It suffices to show that $t(\bs{Z},\bs{Y} - \bs{Z} \circ \bs{\delta})$ can achieve its infimum over $\bs{\delta} \in \mathcal{H}_{k,c}$. 
From the discussion in \S \ref{sec:simp_opt}, 
the infimum of $t(\bs{Z},\bs{Y} - \bs{Z} \circ \bs{\delta})$ over $\bs{\delta} \in \mathcal{H}_{k,c}$ has an equivalent form shown in \eqref{eq:equiv_min_test_stat}, which immediately implies that there exists $(\tilde{l}_1, \tilde{l}_2, \ldots, \tilde{l}_S) \in \mathcal{K}_{\bs{n}}(N-k)$ such that  
$
\inf_{\bs{\delta} \in \mathcal{H}_{k,c}} t(\bs{Z},\bs{Y} - \bs{Z}\circ \bs{\delta}) = \sum_{s=1}^S t_{s, c}(\tilde{l}_s).
$
From \citet[][Theorem 3]{li2020quantile}, 
there exists $\tilde{\bs{\delta}} = (\tilde{\bs{\delta}}_1^\top, \ldots, \tilde{\bs{\delta}}_S^\top)^\top$ with $\tilde{\bs{\delta}}_s \in \mathcal{H}^s_{n_s, n_s - \tilde{l}_s, c}$ for $1\le s\le S$ such that 
$t_{s, c}(\tilde{l}_s) = t_s(\bs{Z}_s, \bs{Y}_s - \bs{Z}_s \circ \tilde{\bs{\delta}}_s)$ for $1\le s\le S$. 
Therefore, we must have 
$\tilde{\bs{\delta}} \in \mathcal{H}_{k,c}$ by definition, and 
\begin{align*}
    \inf_{\bs{\delta} \in \mathcal{H}_{k,c}} t(\bs{Z},\bs{Y} - \bs{Z}\circ \bs{\delta}) = \sum_{s=1}^S t_{s, c}(\tilde{l}_s) = 
    \sum_{s=1}^S t_s(\bs{Z}_s, \bs{Y}_s - \bs{Z}_s \circ \tilde{\bs{\delta}}_s)
    = t(\bs{Z}, \bs{Y} - \bs{Z} \circ \tilde{\bs{\delta}}). 
\end{align*}
From the above, Theorem \ref{thm:pNkc} holds. 

\subsection{Proof of Theorem \ref{thm:bound_pval} and Remark \ref{rmk:bound_p_lp}}

To prove Theorem \ref{thm:bound_pval} and Remark \ref{rmk:bound_p_lp}, we introduce the following lemma. 

\begin{lemma}\label{lemma:rank_bound}
Consider any given $n\ge 1$, any $z=(z_1, \ldots, z_n)\in \{0,1\}^n$ and any permutation $\pi$ of $\{1,2,\ldots,n\}$. 
We define the following three rank functions that rank units from $1$ to $n$ based on any given vector $y=(y_1, \ldots, y_n)\in \mathbb{R}^n$: for any $1\le i\ne j \le n$, 
\begin{itemize}
    \item[(a)] $\rank_i(y) < \rank_j(y)$ if (i) $y_i < y_j$, or (ii) $y_i=y_j$ and $\pi_i < \pi_j$; 
    \item[(b)] $\overline{\rank}_i(y) < \overline{\rank}_j(y)$ if (i) $y_i < y_j$, (ii) $y_i = y_j$ and $z_i < z_j$, or (iii) $y_i = y_j$, $z_i = z_j$, and $\pi_i < \pi_j$; 
    \item[(c)] $\underline{\rank}_i(y) < \underline{\rank}_j(y)$ if (i) $y_i < y_j$, (ii) $y_i = y_j$ and $z_i > z_j$, or (iii) $y_i = y_j$, $z_i = z_j$, and $\pi_i < \pi_j$.  
\end{itemize}
Let $\phi(\cdot)$ be any increasing function on $\{1,2, \ldots, n\}$,
and define further $d(z, y) = \sum_{i=1}^n z_i \phi(\rank_i(y))$, 
$\overline{d}(z, y) = \sum_{i=1}^n z_i \phi( \overline{\rank}_i(y) )$, 
and $\underline{d}(z, y) = \sum_{i=1}^n z_i \phi( \underline{\rank}_i(y) )$. 
\begin{itemize}
    \item[(1)] For any $y = (y_1, \ldots, y_n)\in \mathbb{R}^n$, 
    $\underline{d}(z, y) \le d(z, y) \le \overline{d}(z, y)$. 
    \item[(2)] For any finite positive real number $a$, $\overline{d}(z, y) \le \underline{d}(z, y + a z)$.
\end{itemize}
Note that the values of $\overline{d}(z, y)$ and $\underline{d}(z, y)$ do not depend on how we rank tied units with the same $z_i$'s;
for simplicity, we use the same ranking as that in $\rank_i(\cdot)$ for these units.  
\end{lemma}

\begin{proof}[Proof of Lemma \ref{lemma:rank_bound}]
We first prove (1).
For any $1\le j\le n$ with $z_j=1$, by definition, 
\begin{align}\label{eq:rank_j}
    \rank_j(y) & = \sum_{i=1}^n \I(y_i < y_j) + \sum_{i=1}^n \I(y_i = y_j) \I(\pi_i \le \pi_j), 
    \nonumber
    \\
    \overline{\rank}_j(y) 
    & = 
    \sum_{i=1}^n \I(y_i < y_j)
    + 
    \sum_{i=1}^n  \I(y_i=y_j) \{(1-z_i)+z_i\I(\pi_i\le \pi_j)\}, 
    \nonumber
    \\
    \underline{\rank}_j(y)
    & = 
    \sum_{i=1}^n \I(y_i < y_j)
    + 
    \sum_{i=1}^n  \I(y_i=y_j) z_i \I(\pi_i \le \pi_j). 
\end{align}
These immediately imply that, for any $j$ with $z_j=1$, 
$\underline{\rank}_j(y) \le \rank_j(y) \le \overline{\rank}_j(y)$. 
Consequently, we must have $\underline{d}(z, y) \le d(z, y) \le \overline{d}(z, y)$.

We then prove (2). 
From \eqref{eq:rank_j}, for any $1\le j \le n$ with $z_j=1$, 
\begin{align*}
    & \quad \ \underline{\rank}_j(y+az)\\
    & = 
    \sum_{i=1}^n \I(y_i + az_i < y_j+a)
    + 
    \sum_{i=1}^n  \I(y_i+az_i=y_j+a) z_i \I(\pi_i \le \pi_j)\\
    & = 
    \sum_{i=1}^n \{z_i+(1-z_i)\} \I(y_i + az_i < y_j+a)
    + 
    \sum_{i=1}^n  \I(y_i=y_j) z_i \I(\pi_i \le \pi_j)\\
    & = 
    \sum_{i=1}^n z_i \I(y_i < y_j)
    + 
    \sum_{i=1}^n (1-z_i) \I(y_i < y_j+a)
    + 
    \sum_{i=1}^n  \I(y_i=y_j) z_i \I(\pi_i \le \pi_j). 
\end{align*}
From \eqref{eq:rank_j}, for any $1\le j \le n$ with $z_j=1$, we then have
\begin{align*}
    & \quad \ \underline{\rank}_j(y+az) - \overline{\rank}_j(y) \\
    & =
    \sum_{i=1}^n (1-z_i) \I(y_i < y_j+a)
    - \sum_{i=1}^n (1-z_i)\I(y_i < y_j) -
    \sum_{i=1}^n (1-z_i) \I(y_i=y_j) {\lxr \I(\pi_i \le \pi_j)} \\
    & {\lxr \ge }
    \sum_{i=1}^n (1-z_i) \{ \I(y_i < y_j+a) - \I(y_i \le y_j) \}
    \ge 0. 
\end{align*}
This immediately implies that 
$\underline{d}(z, y + a z) \ge \overline{d}(z, y)$.

From the above, Lemma \ref{lemma:rank_bound} holds. 
\end{proof}

\begin{proof}[Proof of Theorem \ref{thm:bound_pval}]
For any $\delta\in \mathbb{R}^N$, 
let $\overline{t}_s(Z_s, Y_s - Z_s \circ \delta_s)$ be the rank sum statistic for stratum $s$ using rank function $\overline{\rank}(\cdot)$, 
and $\overline{t}(Z, Y - Z \circ \delta)$ be the stratified rank sum statistic using rank function $\overline{\rank}(\cdot)$. 
Define analogously $\underline{t}_s(Z_s, Y_s - Z_s \circ \delta_s)$ and $\underline{t}(Z, Y - Z \circ \delta)$ using the rank function $\underline{\rank}(\cdot)$.

We first show that $\underline{p}_{k,c}  \le {p}_{k,c} \le \overline{p}_{k,c}$ for all $0\le k\le N$ and $c\in\mathbb R$.
From Lemma \ref{lemma:rank_bound}(1), we can know that, for any $\delta \in \mathbb{R}^N$, 
\begin{align*}
    \underline{t}_s(Z_s, Y_s - Z_s \circ \delta_s)
    \le 
    \sum_{i=1}^{n_s} Z_{si} \phi_s( \rank_i(Y_s - Z_s\circ \delta_s) ) 
    \le 
    \overline{t}_s(Z_s, Y_s - Z_s \circ \delta_s), 
    \quad (1\le s\le S)
\end{align*}
and thus 
$\underline{t}(Z, Y - Z \circ \delta) \le t(Z, Y - Z \circ \delta)\le \overline{t}(Z, Y - Z \circ \delta)$. 
This then implies that, for any $0\le k\le N$ and $c\in \mathbb{R}$, 
\begin{align*}
    G\Big(\inf_{\delta\in \mathcal{H}_{k,c}}\overline{t}(Z, Y - Z \circ \delta) \Big)
    \le  
    G\Big(\inf_{\delta\in \mathcal{H}_{k,c}}t(Z, Y - Z \circ \delta)\Big)
    \le 
    G\Big(\inf_{\delta\in \mathcal{H}_{k,c}}\underline{t}(Z, Y - Z \circ \delta)\Big), 
\end{align*}
i.e., $\underline{p}_{k,c} \le p_{k,c} \le \overline{p}_{k,c}$.

We then prove that, for any $0\le k\le N$ and $c<c'$, $\overline{p}_{k,c}\le \underline{p}_{k,c'}$. 
For any $\delta\in \mathcal H_{k,c}$, we must have $\delta' \equiv \delta + (c'-c) 1_N \in \mathcal H_{k,c'}$. This is because 
\begin{align*}
    \sum_{i=1}^N \I(\delta_i'> c' )
    =
    \sum_{i=1}^N \I\{\delta_i + (c'-c) > c' \} = 
    \sum_{i=1}^N \I(\delta_i  > c ) \le N-k. 
\end{align*}
From Lemma \ref{lemma:rank_bound}(2), for any $1\le s\le S$, 
\begin{align*}
    \overline{t}_{s}(Z_s, Y_s - Z_s \circ \delta_s') 
    \le 
    \underline{t}_{s}(Z_s, Y_s - Z_s \circ \delta_s' + (c'-c) Z_s)
    = \underline{t}_{s}(Z_s, Y_s - Z_s \circ \delta_s), 
\end{align*}
which immediately implies that 
$\overline{t}(Z, Y-Z\circ \delta') \le \underline{t}(Z, Y-Z\circ \delta)$. 
Thus, 
we must have 
$\inf_{\delta \in \mathcal{H}_{k,c'}}\overline{t}(Z, Y-Z\circ \delta) \le \inf_{\delta \in \mathcal{H}_{k,c}} \underline{t}(Z, Y-Z\circ \delta)$, which implies that 
\begin{align*}
    \overline{p}_{k,c} = G\Big( \inf_{\delta \in \mathcal{H}_{k,c}} \underline{t}(Z, Y-Z\circ \delta) \Big) \le 
    G\Big( \inf_{\delta \in \mathcal{H}_{k,c'}}\overline{t}(Z, Y-Z\circ \delta) \Big) = \underline{p}_{k,c'}. 
\end{align*}

From the above, Theorem \ref{thm:bound_pval} holds.
\end{proof}

\begin{proof}[Proof of Remark \ref{rmk:bound_p_lp}]
Define $\overline{t}_{s,c}(l)$ and $\underline{t}_{s,c}(l)$ the same as $t_{s,c}(l)$ in \eqref{eq:t_s,c} but using rank functions $\overline{\rank}(\cdot)$ and $\underline{\rank}(\cdot)$, respectively, 
and $\overline{t}_{k,c}'$ and $\underline{t}_{k,c}'$ as the minimum values of the objective function from the linear programming in \eqref{eq:linear_program} with $t_{s,c}(l)$'s replaced by $\overline{t}_{s,c}(l)$'s and $\underline{t}_{s,c}(l)$'s, respectively. 
From a special case of Theorem \ref{thm:bound_pval} with only one stratum, we can know that, 
for any $1\le s\le S$, $0\le l\le n_s$ and $c, c'\in \mathbb{R}$ with $c<c'$, 
\begin{align*}
    \underline{t}_{s,c}(l) \le t_{s,c}(l) \le \overline{t}_{s,c}(l), 
    \quad 
    \overline{t}_{s,c'}(l) \le \underline{t}_{s,c}(l), 
    \quad (1\le s\le S). 
\end{align*}
Note that the feasible region in the linear programming \eqref{eq:linear_program} does not depend on the choice of rank functions and the value of $c$. 
We must have that, for any $0\le k\le n$ and $c, c'\in \mathbb{R}$ with $c<c'$, 
$\underline{t}_{k,c}' \le t_{k,c}' \le \overline{t}_{k,c}'$, and 
$\overline{t}_{k,c'}' \le \underline{t}_{k,c}'$. 
Therefore, the $p$-values from the linear programming \eqref{eq:linear_program} using rank functions $\overline{\rank}(\cdot)$ and $\underline{\rank}(\cdot)$ must satisfy \begin{align*}
    \underline{p}_{k,c}' \equiv G(\overline{t}_{k,c}') 
    \le 
    p_{k,c}' \equiv G(t_{k,c}') 
    \le 
    \overline{p}_{k,c}' \equiv G(\overline{t}_{k,c}'), 
    \quad 
    ( 0\le k\le N, c\in \mathbb{R} )
\end{align*}
and 
\begin{align*}
    \overline{p}_{k,c}' \equiv
    G( \underline{t}_{k,c}' ) \le G( \overline{t}_{k,c'}' )
    \equiv \underline{p}_{k,c'}', 
    \quad 
    ( 0\le k\le N, \ c<c'\in \mathbb{R}).
\end{align*}
Therefore, Theorem \ref{thm:bound_pval} also holds when we consider the $p$-value $p_{k,c}'$
from the relaxed linear programming in \eqref{eq:linear_program}.
\end{proof}

\subsection{Proof of Theorem \ref{thm:conf_set} and Remarks \ref{rmk:conf_set_plus} and \ref{rmk:ci_order}}
To prove Theorem \ref{thm:conf_set}, we need the following lemma. 

\begin{lemma}\label{lemma:monotone_t}
$
t_{k,c} \equiv \inf_{\bs{\delta} \in \mathcal{H}_{k,c}} t(\bs{Z},\bs{Y} - \bs{Z}\circ \bs{\delta})
$ 
is monotone decreasing in $c$ and increasing in $k$.
\end{lemma}
\begin{proof}[Proof of Lemma \ref{lemma:monotone_t}]
First, we prove that $t_{k,c} $ is decreasing in $c$. By the definition in \eqref{eq:set_H_Nkc}, for any $0 \le k \le N$ and any $c_1,c_2\in \mathbb{R}$ with $c_1\le c_2$, $\mathcal{H}_{k,c_1}\subseteq \mathcal{H}_{k,c_2}$. 
This immediately implies that $t_{k,c_1} \ge t_{k,c_2}$. Thus, $t_{k,c} $ is decreasing in $c$. 

Second, we prove that $t_{k,c} $ is increasing in $k$. For any $c\in \mathbb{R}$ and integers $k_1,k_2\in [0,N]$ with $k_1\le k_2$, 
because $\delta_{(k_1)} \le \delta_{(k_2)}$ for any $\bs{\delta}\in \mathbb{R}^N$, we can derive that $\mathcal{H}_{k_1,c}\supseteq \mathcal{H}_{k_2,c}$. This immediately implies that $t_{k_1,c} \le t_{k_2,c}$. Thus, $t_{k,c} $ is increasing in $k$.

From the above, Lemma \ref{lemma:monotone_t} holds.
\end{proof}

\begin{proof}[Proof of Theorem \ref{thm:conf_set}]
First, by Lemma \ref{lemma:monotone_t} and the fact that $G(\cdot)$ is a monotone decreasing function, we can know that 
$p_{k,c} = G(t_{k,c})$ is increasing in $c$ and decreasing in $k$.

Second, we consider Theorem \ref{thm:conf_set}(i) and (ii). 
The validity of confidence sets $\mathcal{I}_{\alpha}(k)$ and $\mathcal{S}_{\alpha}(c)$ follows from the validity of the $p$-value in \eqref{thm:pNkc}, and their simplified forms follow from the property that the $p$-value $p_{k,c}$ is increasing in $c$ and decreasing in $k$. 
For conciseness, we omit the detailed proof here. 

Third, we prove that 
\begin{align}\label{eq:intersection_k}
    \left\{
    \bs{\delta}\in \mathbb{R}^N:\delta_{(k)}\in\mathcal{I}_\alpha(k), 1 \le k\le N
    \right\} =\bigcap_{k,c:p_{k,c}\le\alpha} \mathcal{H}_{k,c}^\complement. 
\end{align}
Suppose $\bs{\delta}$ is in the set on the left hand side of \eqref{eq:intersection_k}. 
Then by definition, we have $p_{k,\delta_{(k)}}>\alpha$ for all $1\le k \le N$. 
For any $1\le k\le N$ and $c\in\mathbb{R}$ such that $p_{k,c}\le\alpha$, because $p_{k,c}$ is increasing in $c$, we must have $\delta_{(k)}>c$ or equivalently $\bs{\delta}\in\mathcal{H}_{k,c}^\complement$.
This immediately implies that $\bs{\delta}$ is in the set on the right hand side of \eqref{eq:intersection_k}. 
Now suppose $\bs{\delta}$ is not in the set on the left hand side of \eqref{eq:intersection_k}. Then there exists $1\le k_0\le N$ such tht  $\delta_{(k_0)}\notin\mathcal{I}_\alpha(k_0)$ or equivalently $p_{k_0,\delta_{(k_0)}}\le\alpha$. 
This implies that $\bs{\delta}\in\mathcal{H}_{k_0,\delta_{(k_0)}} \subset  \bigcup_{k,c:p_{k,c}\le\alpha}\mathcal{H}_{k,c}$. Consequently, $\bs{\delta}$ is not in the set on the right hand side of \eqref{eq:intersection_k}. Therefore, the two sets on the left and right hand sides of \eqref{eq:intersection_k} must be equal, i.e.,  \eqref{eq:intersection_k} holds.

Fourth, we prove that 
\begin{align}\label{eq:intersection_c}
    \left\{
    \bs{\delta}\in \mathbb{R}^N:\sum_{i=1}^N \I(\delta_i > c) \in\mathcal{S}_\alpha(c),c\in\mathcal{R}
    \right\} =\bigcap_{k,c:p_{k,c}\le\alpha} \mathcal{H}_{k,c}^\complement. 
\end{align}
Suppose $\bs{\delta}$ is in the set on the left hand side of \eqref{eq:intersection_c}. Then by definition, $p_{N-n(c),c}>\alpha$ for any $c\in \mathbb{R}$, where $n(c)=\sum_{i=1}^N \I(\delta_i >  c)$. 
For any $1\le k\le N$ and any $c\in \mathbb{R}$ such that $p_{k,c}\le \alpha$, 
because $p_{k,c}$ is decreasing in $k$, 
we must have $N - n(c) < k$ or equivalently $\bs{\delta}\in\mathcal{H}_{k,c}^\complement$. 
This immediately implies that $\bs{\delta}$ is in the set on the right hand side of \eqref{eq:intersection_c}. 
Now suppose $\bs{\delta}$ is not in the set on the left hand side of \eqref{eq:intersection_c}. 
Then there exists $c_0\in\mathbb{R}$ such that $n(c_0)\notin\mathcal{S}_\alpha(c_0)$ or equivalently $p_{N-n(c_0),c_0}\le\alpha$. 
This implies that $\bs{\delta}\in\mathcal{H}_{N-n(c_0),c_0} \subset \bigcup_{k,c:p_{k,c}\le\alpha}\mathcal{H}_{k,c}$. 
Consequently, $\bs{\delta}$ is not in the set on the right hand side of \eqref{eq:intersection_c}. 
Therefore, the two sets on the left and right hand sides of \eqref{eq:intersection_c} must be equal, i.e.,  \eqref{eq:intersection_c} holds.
    
Fifth, we prove that $\bigcap_{k,c:p_{k,c}\le\alpha}\mathcal{H}_{k,c}^\complement$ is a $1-\alpha$ confidence set for the true treatment effect $\tau$.
If $\bs{\tau}\in \bigcup_{k, c: p_{k,c} \le \alpha} 
\mathcal{H}_{k,c}$, then there must exist a $(k,c)$ such that $p_{k,c}\le \alpha$ and $\bs{\tau}\in\mathcal{H}_{k,c}$, which further implies that $p_{\bs{Z},\bs{\tau}}\le \sup_{\bs{\delta}\in\mathcal{H}_{k,c}}p_{\bs{Z},\bs{\delta}} =  p_{k,c}\le\alpha$. 
Thus, the coverage probability of the set  $\bigcap_{k,c:p_{k,c}\le\alpha}\mathcal{H}_{k,c}^\complement$ satisfies 
\begin{align*}
    \Pr\Big(
    \bs{\tau} \in \bigcap_{k, c: p_{k,c} \le \alpha} 
    \mathcal{H}_{k,c}^{\c}
    \Big) =1- \Pr\Big(
    \bs{\tau} \in \bigcup_{k, c: p_{k,c}  \le \alpha} 
    \mathcal{H}_{k,c} \Big)
    \ge 1-\Pr(p_{\bs{Z,\tau}}\le \alpha)\ge 1-\alpha, 
\end{align*}
where the last inequality holds because $p_{\bs{Z},\bs{\tau}}$ is stochastically larger than or equal to $\Unif(0,1)$ by the validity of the FRT. 
Therefore, $\bigcap_{k,c:p_{k,c}\le\alpha}\mathcal{H}_{k,c}^\complement$ is a $1-\alpha$ confidence set for the true individual treatment effect $\bs \tau$. 

From the above, Theorem \ref{thm:conf_set} holds.
\end{proof}

To prove Remark \ref{rmk:conf_set_plus}, we need the following two lemmas. 

\begin{lemma}\label{lemma:monotone_t_stratum}
For all $1\le s \le S$, $t_{s, c}(l)$ is is decreasing in $c$ and $l$.
\end{lemma}

\begin{proof}[Proof of Lemma \ref{lemma:monotone_t_stratum}]
Lemma \ref{lemma:monotone_t_stratum} is actually a special case of Lemma \ref{lemma:monotone_t} with only one stratum, and thus it is directly implied by  Lemma \ref{lemma:monotone_t}. 
\end{proof}

\begin{lemma}\label{lemma:monotone_t_lp}
$
t^\LP_{k,c} 
$ 
is decreasing in $c$ and increasing in $k$.
\end{lemma}
\begin{proof}[Proof of Lemma \ref{lemma:monotone_t_lp}]
First, we prove that  $t_{N,k,c}^{\LP}$ is decreasing in $c$ for any $0 \le k \le N$. 
For any $c_1, c_2\in\mathbb{R}$ with $c_1\le c_2$ and  any $0 \le l_s \le n_s$,
from Lemma \ref{lemma:monotone_t_stratum}, we have $t_{s,c_1}(l_s) \ge t_{s,c_2}(l_s)$. 
Thus, for any $x_{sl}$'s satisfying the constrains in the linear programming problem \eqref{eq:linear_program}, we must have 
$
    \sum_{s=1}^S\sum_{l=1}^{n_{sk}} t_{s,c_1}(l)\cdot x_{sl}\ge \sum_{s=1}^S\sum_{l=1}^{n_{sk}} t_{s,c_2}(l)\cdot x_{sl}. 
$
This immediately implies that the solution from the linear programming \eqref{eq:linear_program} when $c=c_1$ must be greater than or equal to that when $c=c_2$, i.e. $t_{k,c_1}\ge t_{k,c_2}$.

Second, we prove that $t_{k,c}^{\LP}$ is increasing in $k$ for any given $c\in \mathbb{R}$.
It suffices to prove that $t_{k,c}^{\LP}\ge t_{k-1,c}^{\LP}$ for any $c\in \mathbb{R}$ and $1\le k\le N$. 
Consider any $x_{sl}$'s that satisfy the constraints in the linear programming problem \eqref{eq:linear_program}, i.e., 
\begin{align}
    \sum_{l=0}^{n_{sk}} x_{sl} = 1, 
    \quad 
    & \sum_{s=1}^S \sum_{l=0}^{n_{sk}} l \cdot x_{sl} = N-k, 
    \quad 
    0 \le x_{sl} \le 1, 
    \qquad (1\le l \le n_{sk}; s=1, 2, \ldots, S). 
\end{align}
For each $1\le s\le S$, 
if $n_{s,k-1} = n_{sk}+1$, we additionally define $x_{s, n_{s,k-1}} = 0$. 
We then consider the following transformation of $x_{sl}$'s: 
\begin{align*}
    \tilde{x}_{sl} = \begin{cases}
    x_{sl}-\epsilon_{sl}, & \text{for } l= 0,1,\ldots n_{s,k-1}-1 \\
    x_{sl}+\sum_{i=0}^{n_{s,k-1}-1}\epsilon_{si}, & \text{for } l= n_{s, k-1}, 
    \end{cases}
\end{align*}
for some $\epsilon_{sl}$'s satisfying  $0\le \epsilon_{sl}\le x_{sl}$ for all $s$. 
We can verify that 
\begin{align*}
    \sum_{l=1}^{n_{s,k-1}} \tilde{x}_{sl} = \sum_{l=1}^{n_{s,k-1}} x_{sl} = \sum_{l=1}^{n_{sk}} x_{sl} = 1, \qquad (1\le s\le S)
\end{align*}
and 
\begin{align*}
    0 \le 
    \tilde{x}_{sl} 
    \le 
    \begin{cases}
    x_{sl}, & \text{if } 0 \le l \le n_{s,k-1}-1 \\
    x_{sl}+\sum_{i=0}^{n_{s,k-1}-1}x_{si}, & \text{if } l= n_{s, k-1}, 
    \end{cases}
    \le 1, 
    \quad (1\le l \le n_{s,k-1}; 1\le s \le S). 
\end{align*}
Moreover, 
\begin{align*}
    \sum_{s=1}^S\sum_{l=0}^{n_{s,k-1}} l\cdot \tilde{x}_{sl}
    & = \sum_{s=1}^S\sum_{l=0}^{n_{s,k-1}} l\cdot x_{sl} + \sum_{s=1}^S\sum_{l=1}^{n_{s,k-1}-1} (n_{s,k-1}-l)\cdot \epsilon_{sl}
    \\
    & = \sum_{s=1}^S\sum_{l=0}^{n_{sk}} l\cdot x_{sl} + \sum_{s=1}^S\sum_{l=1}^{n_{s,k-1}-1} (n_{s,k-1}-l)\cdot \epsilon_{sl}
    \\ 
    & = 
    N-k + \sum_{s=1}^S\sum_{l=1}^{n_{s,k-1}-1} (n_{s,k-1}-l)\cdot \epsilon_{sl}
\end{align*}
is a
continuous and increasing function of $\epsilon_{sl}$'s, 
which achieves minimum value $N-k$ when $\epsilon_{sl}$'s are all zero
and achieves maximum value 
$\sum_{s=1}^S n_{s,k-1} \ge N-(k-1)$ when $\epsilon_{sl} = x_{sl}$ for all $s$ and $1\le l \le n_{s,k-1}-1$.  
Therefore, there must exist $\epsilon_{sl}$'s such that 
$\sum_{s=1}^S\sum_{l=0}^{n_{s,k-1}} l\cdot \tilde{x}_{ls} = N-(k-1)$. 
Moreover, the corresponding $\tilde{x}_{ls}$'s must satisfy that 
\begin{align*}
    \sum_{s=1}^S\sum_{l=0}^{n_{s,k-1}} t_{s,c}(l)\cdot \tilde{x}_{sl}&=\sum_{s=1}^S\sum_{l=0}^{n_{s,k-1}} t_{s,c}(l)\cdot x_{sl}+\sum_{s=1}^S\sum_{l=1}^{n_{s,k-1}-1} \left\{t_{s,c}(n_{s,k-1})-t_{s,c}(l)\right\}\cdot \epsilon_{sl}
    \\
    & = \sum_{s=1}^S\sum_{l=0}^{n_{sk}} t_{s,c}(l)\cdot x_{sl}+\sum_{s=1}^S\sum_{l=1}^{n_{s,k-1}-1} \left\{t_{s,c}(n_{s,k-1})-t_{s,c}(l)\right\}\cdot \epsilon_{sl}
    \\
    &\le \sum_{s=1}^S\sum_{l=0}^{n_{sk}} t_{s,c}(l)\cdot x_{sl}, 
\end{align*}
where the last equality holds because $t_{s,c}(l)$ is decreasing in $l$ as shown in Lemma \ref{lemma:monotone_t_stratum}.
From the above, for any possible value of the objective function from the linear programming \eqref{eq:linear_program}, 
we can always find a smaller value of the objective function from the linear programming \eqref{eq:linear_program} with $k$ replaced by $k-1$.
Therefore, the solution from the linear programming \eqref{eq:linear_program} with $k$ replaced by $k-1$ must be smaller than or equal to that from \eqref{eq:linear_program}, i.e., $ t_{k-1,c}^{\LP} \le t_{k,c}^{\LP}$. 

From the above, Lemma \ref{lemma:monotone_t_lp} holds.
\end{proof}

\begin{proof}[Proof of Remark \ref{rmk:conf_set_plus}]
From Lemmas \ref{lemma:monotone_t_lp}, 
Remark \ref{rmk:conf_set_plus} follows
by the same logic as Theorem \ref{thm:conf_set}. 
For conciseness, we omit the detailed proof here. 
\end{proof}
{
\begin{proof}[Proof of Remark \ref{rmk:ci_order}]
From Theorem \ref{thm:bound_pval},
for any $1\le k\le N$, 
$\{c: \underline{p}_{k, c} > \alpha\} \subset \{c: {p}_{k, c} > \alpha\} \subset \{c: \overline{p}_{k, c} > \alpha\}$. 
Moreover, because $\underline{p}_{k, c}$ and $\overline{p}_{k, c}$ are essentially special cases of $p_{k,c}$ when all control units are ordered before or after treated units, 
by the same logic as Theorem \ref{thm:conf_set}, 
all of the three sets are intervals of form $(c, \infty)$ or $[c, \infty)$. 
Below we prove that $\mathcal{C} = \{c: \overline{p}_{k, c} > \alpha\} \setminus \{c: \underline{p}_{k, c} > \alpha\}$ contains at most one point by contradiction.
Suppose that $c_1<c_2$ are both in $\mathcal{C}$. 
Then we have $\underline{p}_{k, c_2} \le \alpha$ and $\overline{p}_{k, c_1} > \alpha$. 
However, this contradicts with Theorem \ref{thm:bound_pval}, which implies that $\overline{p}_{k, c_1} \le \underline{p}_{k, c_2}$. 
Thus, $\mathcal{C}$ contains at most one point.
By the same logic, this will also be true when we consider the $p$-value from the relaxed linear programming. 
From the above, Remark \ref{rmk:ci_order} holds. 
\end{proof}
}

\subsection{Proof of Theorem \ref{thm:sen_ana}}

First, we prove that $\sup_{\bs{u}\in\mathcal U} G_{\bs Z,\bs \delta,\bs u,\Gamma}(c)$ does not depend on $Y_{\bs{Z}, \bs{\delta}}(0)$, and the supremum can be achieved at some $\bs{u} \in \mathcal{U}$. 
By the definition of the stratified rank score statistic in Definition \ref{def:strat_rank_score}, for each $1\le s\le S$, 
there exists a permutation $\{\pi_{s1}, \pi_{s2}, \ldots, \pi_{sn_s}\}$ of $\{1, 2, \ldots, n_s\}$ such that 
$\rank_{\pi_{si}}(\bs{Y}_s - \bs{Z}_s \circ \bs{\delta}_s) = i$ for $1\le i \le n_s$. 
For any $\bs{a}\in \mathcal{Z}$ and $\bs{u} \in \mathcal{U}$, 
define their permutations $\tilde{\bs{a}}$ and $\tilde{\bs{u}}$ such that 
$\tilde{a}_{si} = a_{s\pi_{si}}$ and $\tilde{u}_{si} = u_{s\pi_{si}}$ for all $1\le s\le S$ and $1\le i \le n_s$. 
Let $\bs{y}=(\bs{y}_1^\top, \ldots, \bs{y}_S^\top)^\top$ be a vector such that $\rank_i(\bs{y}_s) = i$ for $1\le i \le n_s$. 
We then have  
\begin{align*}
    t(\bs{a}, \bs{Y}_{\bs{Z},\bs{\delta}}(0))
    & = 
    \sum_{s=1}^S \sum_{i=1}^{n_s} a_{si} \phi_s\left
    (\rank_i(\bs{Y}_s - \bs{Z}_s \circ \bs{\delta}_s) \right)
    = 
    \sum_{s=1}^S \sum_{i=1}^{n_s} a_{s\pi_{si}} \phi_s(i)
    \\
    & = \sum_{s=1}^S \sum_{i=1}^{n_s} a_{s\pi_{si}} \phi_s(\rank_i(\bs{y}_s))
    = \sum_{s=1}^S \sum_{i=1}^{n_s} \tilde{a}_{si} \phi_s(\rank_i(\bs{y}_s)) = 
    t(\tilde{\bs{a}}, \bs{y}), 
\end{align*}
and consequently, 
\begin{align*}
G_{\bs{Z}, \bs{\delta}, \bs{u}, \Gamma} (c) 
& = 
\sum_{\bs{a} \in \mathcal{Z}} 
\prod_{s=1}^S
\frac{\exp(\gamma \sum_{i=1}^{n_s} a_{si} u_{si})}{
\sum_{i=1}^{n_s} \exp(\gamma u_{si})}
\I\left\{
t(\bs{a}, \bs{Y}_{\bs{Z},\bs{\delta}}(0))
\ge 
c
\right\}
\\
& = 
\sum_{\bs{a} \in \mathcal{Z}} 
\prod_{s=1}^S
\frac{\exp(\gamma \sum_{i=1}^{n_s} \tilde{a}_{si} \tilde{u}_{si})}{
\sum_{i=1}^{n_s} \exp(\gamma \tilde{u}_{si})}
\I\left\{
t(\tilde{\bs{a}}, \bs{y})
\ge 
c
\right\}\\
& = 
\sum_{\bs{a} \in \mathcal{Z}} 
\prod_{s=1}^S
\frac{\exp(\gamma \sum_{i=1}^{n_s} a_{si} \tilde{u}_{si})}{
\sum_{i=1}^{n_s} \exp(\gamma \tilde{u}_{si})}
\I\left\{
t(\bs{a}, \bs{y})
\ge 
c
\right\}. 
\end{align*}
Thus, 
\begin{align}\label{eq:G_sup_over_u_proof}
    \sup_{\bs{u}\in \mathcal{U}}G_{\bs{Z}, \bs{\delta}, \bs{u}, \Gamma} (c)  
    & = 
    \sup_{\bs{u}\in \mathcal{U}}
    \sum_{\bs{a} \in \mathcal{Z}} 
    \prod_{s=1}^S
    \frac{\exp(\gamma \sum_{i=1}^{n_s} a_{si} \tilde{u}_{si})}{
    \sum_{i=1}^{n_s} \exp(\gamma \tilde{u}_{si})}
    \I\left\{
    t(\bs{a}, \bs{y})
    \ge 
    c
    \right\}
    \nonumber
    \\
    & = 
    \sup_{\bs{u}\in \mathcal{U}}
    \sum_{\bs{a} \in \mathcal{Z}} 
    \prod_{s=1}^S
    \frac{\exp(\gamma \sum_{i=1}^{n_s} a_{si} u_{si})}{
    \sum_{i=1}^{n_s} \exp(\gamma u_{si})}
    \I\left\{
    t(\bs{a}, \bs{y})
    \ge 
    c
    \right\}, 
\end{align}
which does not depend on the imputed potential outcomes $\bs{Y}_{\bs{Z}, \bs{\delta}}(0)$. 
Moreover, it is not hard to see that the function in \eqref{eq:G_sup_over_u_proof} is continuous in $\bs{u}\in \mathcal{U}$. Because $\mathcal{U}$ is compact, the supremum in \eqref{eq:G_sup_over_u_proof} can be achieved at some $\bs{u} \in \mathcal{U}$. 
Therefore, \eqref{eq:G_Gamma} holds. 


Second, we prove the equivalent forms of $p_{k,c, \Gamma}$ in Theorem \ref{thm:sen_ana}. 
From the proof of Theorem \ref{thm:pNkc}, 
$t(\bs{Z}, \bs{Y}_{\bs{Z}, \bs{\delta}}(0))$ can achieve its infimum over $\bs{\delta} \in \mathcal{H}_{k,c}$ at some $\tilde{\bs{\delta}} \in \mathcal{H}_{k,c}$. 
Let $\tilde{c} = t(\bs{Z}, \bs{Y}_{\bs{Z}, \tilde{\bs{\delta}}}(0))$. 
We then have, for any $\bs{\delta} \in \mathcal{H}_{k,c}$ and $\bs{u} \in \mathcal{U}$, 
\begin{align*}
    p_{\bs{Z}, \bs{\delta}, \bs{u}, \Gamma} \equiv G_{\bs{Z}, \bs{\delta}, \bs{u}, \Gamma} \left( t(\bs{Z}, \bs{Y}_{\bs{Z}, \bs{\delta}}(0)) \right) 
    \le 
    G_{\bs{Z}, \bs{\delta}, \bs{u}, \Gamma} \left( \tilde{c} \right) 
    \le 
    G_{\Gamma}( \tilde{c} ). 
\end{align*}
This immediately implies that $p_{k,c, \Gamma} \equiv \sup_{\bs{\delta} \in \mathcal{H}_{k,c}, \bs{u} \in \mathcal{U}} p_{\bs{Z}, \bs{\delta}, \bs{u}, \Gamma} \le G_{\Gamma}( \tilde{c} )$. 
From the discussion before, there must exists $\tilde{\bs{u}}\in \mathcal{U}$ such that $G_{\bs{Z}, \tilde{\bs{\delta}}, \tilde{\bs{u}}, \Gamma} ( \tilde{c} ) = G_{\Gamma}( \tilde{c} )$. This implies that 
\begin{align*}
    p_{k,c, \Gamma} 
    & \equiv \sup_{\bs{\delta} \in \mathcal{H}_{k,c}, \bs{u} \in \mathcal{U}} p_{\bs{Z}, \bs{\delta}, \bs{u}, \Gamma} 
    = 
    \sup_{\bs{\delta} \in \mathcal{H}_{k,c}, \bs{u} \in \mathcal{U}}
    G_{\bs{Z}, \bs{\delta}, \bs{u}, \Gamma} \left( t(\bs{Z}, \bs{Y}_{\bs{Z}, \bs{\delta}}(0)) \right) 
    \\
    & \ge 
    G_{\bs{Z}, \tilde{\bs{\delta}}, \tilde{\bs{u}}, \Gamma} \left( t(\bs{Z}, \bs{Y}_{\bs{Z}, \tilde{\bs{\delta}}}(0)) \right)
    = G_{\bs{Z}, \tilde{\bs{\delta}}, \tilde{\bs{u}}, \Gamma}(\tilde{c})
    = G_{\Gamma}( \tilde{c} ). 
\end{align*}
Therefore, we must have $p_{k,c, \Gamma}  = G_{\Gamma}(\tilde{c})$.

Third, we prove the validity of the $p$-value $p_{k,c, \Gamma}$ for testing $H_{k,c}$ in \eqref{eq:H_Nkc} under the sensitivity model with bias at most $\Gamma$. 
Let $\bs{u}^{*}$ be the true value of the unmeasured confounding associated with the sensitivity model with bias at most $\Gamma$. 
If the null hypothesis $H_{k,c}$ holds, we have $\bs{\tau} \in \mathcal{H}_{k,c}$, and thus 
$p_{k,c,\Gamma} \ge p_{\bs{Z}, \bs{\tau}, \bs{u}^*, \Gamma}$, 
which compares the observed value of test statistic $t(\bs{Z}, \bs{Y}(0))$ to its true distribution under the sensitivity model with bias $\Gamma$ and unmeasured confounding $\bs{u}^{*}$. 
Therefore, $p_{k,c,\Gamma} \ge p_{\bs{Z}, \bs{\tau}, \bs{u}^*, \Gamma}$ must be stochastically larger than or equal to $\Unif(0,1)$. 

From the above, Theorem \ref{thm:sen_ana} holds.

\subsection{Proof of Theorem \ref{thm:sen_ana_large_sample}}
To prove Theorem \ref{thm:sen_ana_large_sample}, we need the following three lemmas. 

\begin{lemma}\label{lemma:lindfellerclt}
For each $n$, let $\xi_{n,m}$, $1\leq m\leq n$, be independent random variables with $\E \xi_{n,m}=0$. Suppose 
\begin{itemize}
\item[(i)] $\sum_{m=1}^{n} \E \xi_{n,m}^2 \rightarrow \sigma^2 > 0$, 
\item[(ii)] For all $c > 0$, 
$\lim_{n\rightarrow \infty} \sum_{m=1}^{n} \E [\xi^2_{n,m}1\{|\xi_{n,m}|>c\}] = 0$.
\end{itemize}
Then $\Xi_n = \xi_{n,1} + \ldots + \xi_{n,n} \converged \mathcal{N}(0, \sigma^2)$ as $n\rightarrow \infty$.
\end{lemma}
\begin{proof}[Proof of Lemma \ref{lemma:lindfellerclt}]
Lemma \ref{lemma:lindfellerclt} is the Lindeberg-Feller theorem; see, e.g., \citet[][Theorem 3.4.5]{durrett2010probability}. 
\end{proof}

\begin{lemma}\label{lemma:match_set_clt}
    Let $\bs{Z}\in \mathcal{Z}$ be a treatment assignment vector following the sensitivity model with bias at most $\Gamma\ge 1$ and unmeasured confounding $\bs{u}\in \mathcal{U}$, 
    $\bs{y} \in \mathbb{R}^N$ be any fixed constant vector, 
    and $t(\bs{Z}, \bs{y})$ be the stratified rank score statistic. 
    If Condition \ref{cond:match_set_clt} holds, then 
    \begin{align*}
    \frac{t(\bs{Z}, \bs{y}) - \E_{\bs{u}, \Gamma}\{t(\bs{Z}, \bs{y})\}}{
    \sqrt{\Var_{\bs{u}, \Gamma}\{t(\bs{Z}, \bs{y})\}}
    }
    \converged \mathcal{N}(0, 1), 
    \end{align*}
where $\E_{\bs{u}, \Gamma}(\cdot)$ and $\Var_{\bs{u}, \Gamma}(\cdot)$ denotes the mean and variance under the sensitivity model with bias at most $\Gamma$ and unmeasured confounding $\bs{u}$. 
\end{lemma}
\begin{proof}[Proof of Lemma \ref{lemma:match_set_clt}]
For $1\le s \le S$, 
let $T_s=t(\bs Z_s, \bs y_s)$ be the rank score statistic for stratum $s$, 
and 
$\mu_s=\E_{\bs{u}, \Gamma}(T_s)$ and $v_s^2=\Var_{\bs{u}, \Gamma}(T_s)$
be the mean and variance of $T_s$, both of which depend on the unmeasured confounding $\bs{u}$. 
Let $T = \sum_{s=1}^S T_s$ be the stratified rank score statistic. 
By the mutually independence of $T_s$'s, $T$ has mean $\mu \equiv \E_{\bs{u}, \Gamma}(T) = \sum_{s=1}^S \mu_s$ and variance $\sigma^2 \equiv \Var_{\bs{u}, \Gamma}(T) = \sum_{s=1}^S v_s^2$. 
Define $X_s = (T_s - \mu_s)/\sigma$ for $1\le s\le S$. Obviously, $X_s$ has mean zero for all $s$, and $\sum_{s=1}^S \E_{\bs{u}, \Gamma}(X_s^2) = 1$. 

Below we consider bounding $X_s^2$ for all $s$. 
By construction, $\phi_s(1) \le T_s \le \phi_s(n_s)$. This implies that $\phi_s(1) \le  \mu_s \le \phi_s(n_s)$ and consequently $|T_s - \mu_s| \le \phi_s(n_s)-\phi_s(1) = R_s$. 
Let $1\le j_{s1}, j_{s2} \le n_s$ be the indices such that $\rank_{j_{s1}}(\bs{y}_s) = 1$ and $\rank_{j_{s2}}(\bs{y}_s) = n_s$, 
and 
let $\Pr_{\bs{u}, \Gamma}(\cdot)$ denote the probability measure for the treatment assignment vector $Z$ under the sensitivity model with bias at most $\Gamma$ and unmeasured confounding $\bs{u}$. 
We then have 
\begin{align*}
    v_s^2 & = \E_{\bs{u}, \Gamma}\big\{ \left( T_s - \mu_s \right)^2 \big\}
    = \sum_{i=1}^{n_s} \Pr_{\bs{u}, \Gamma}(Z_{si}=1)\{\rank_{i}(\bs{y}_s)-\mu_s\}^2
    \\
    & \ge 
    \Pr_{\bs{u}, \Gamma}(Z_{sj_{s1}}=1)\{\phi_s(1)-\mu_s\}^2 + \Pr_{\bs{u}, \Gamma}(Z_{sj_{s2}}=1)\{\phi_s(n_s)-\mu_s\}^2. 
\end{align*}
Let $p_{s1} = \Pr_{\bs{u}, \Gamma}(Z_{sj_{s1}}=1)$ and $p_{s2} = \Pr_{\bs{u}, \Gamma}(Z_{sj_{s2}}=1)$. 
We can then bound $\mu_s$ by 
\begin{align*}
    p_{s2} \phi_s(n_s) + (1-p_{s2}) \phi_s(1)
    \le 
    \mu_s \le p_{s1} \phi_s(1) + (1-p_{s1}) \phi_s(n_s), 
\end{align*}
which implies that 
\begin{align*}
    \mu_s - \phi_s(1) \ge p_{s2} \{ \phi_s(n_s) - \phi_s(1)\} = p_{s2} R_s, 
    \quad 
    \phi_s(n_s) - \mu_s \ge p_{s1}  \{ \phi_s(n_s) - \phi_s(1)\} = p_{s1} R_s.  
\end{align*}
Thus, we can bound $v_s^2$ by 
\begin{align*}
    v_s^2 
    \ge p_{s1}\{\phi_s(1)-\mu_s\}^2 + p_{s2}\{\phi_s(n_s)-\mu_s\}^2
    \ge p_{s1} p_{s2}^2 R_s^2 + p_{s2} p_{s1}^2 R_s^2 
    = p_{s1} p_{s2} (p_{s1} + p_{s2})R_s^2. 
\end{align*}
By definition, we can verify that, under the sensitivity model with bias at most $\Gamma$, both $p_{s1}$ and $p_{s2}$ must be greater than or equal to $1/\{1 + (n_s-1)\Gamma\} \ge 1/(n_s\Gamma)$. 
Consequently, 
\begin{align*}
    v_s^2 \ge p_{s1} p_{s2} (p_{s1} + p_{s2})R_s^2 \ge \frac{2}{n_s^3\Gamma^3}R_s^2. 
\end{align*}
We can then bound $X_s^2$ by 
\begin{align}\label{eq:bound_X_s2}
    X_s^2 = \frac{(T_s - \mu_s)^2}{\sigma^2} = \frac{(T_s - \mu_s)^2}{\sum_{s=1}^S v_s^2} 
    \le 
    \frac{\Gamma^3}{2}
    \frac{R_s^2}{ \sum_{s=1}^S R_s^2/n_s^3 }. 
\end{align}

Finally, we prove the asymptotic Gaussianity of $T$. 
From \eqref{eq:bound_X_s2}, 
under Condition \ref{cond:match_set_clt}, as $S \rightarrow \infty$,  
\begin{align*}
    \max_{1\le s \le S} X_s^2  = \max_{1\le s \le S}\frac{(T_s-\mu_s)^2}{\sigma^2} \le \max_{1\le s \le S}
    \frac{\Gamma^3}{2}
    \frac{R_s^2}{\sum_{s=1}^S R_s^2/n_s^3}
    \rightarrow 0. 
\end{align*}
Thus, for any $\epsilon>0$, there exists $\underline{S}$ such that $\max_{1\le s \le S} X_s^2 < \epsilon^2$ when $S\ge \underline{S}$. 
Consequently, 
$\sum_{s=1}^{S} \E_{\bs{u}, \Gamma} [X^2_s \I\{|X_s|>\epsilon\}]=0$ when $S\ge \underline{S}$. 
This implies that, for any $\epsilon>0$, 
$\E_{\bs{u}, \Gamma} [X^2_s \I\{|X_s|>\epsilon\}]$ converges to zero as $S \rightarrow \infty$. 
Therefore, by the Lindeberg--Feller central limit theorem in Lemma \ref{lemma:lindfellerclt}, as $S \rightarrow \infty$, 
\begin{align*}
    \frac{T - \mu}{\sigma} = \frac{\sum_{s=1}^S(T_s-\mu_s)}{\sqrt{\sum_{s=1}^S v_s^2}} = 
    \sum_{s=1}^S X_s \converged \mathcal N(0,1). 
\end{align*}

From the above, Lemma \ref{lemma:match_set_clt} holds. 
\end{proof}

\begin{lemma}\label{lemma:mu_sigma_tilde_greater}
Under the sensitivity model with bias at most $\Gamma$, 
if Condition \ref{cond:sen_asymp_conservative} holds, then for any $k \ge 0$, 
$\tilde{\mu}_{\Gamma} + k \tilde{\sigma}_{\Gamma}$ is greater than or equal to  $\mu + k \sigma$ when $S$ is sufficiently large. 
\end{lemma}

\begin{proof}[Proof of Lemma \ref{lemma:mu_sigma_tilde_greater}]
Recall that $\mathcal{A} = \{s: v_{s}^2 > \tilde{v}_{s,\Gamma}^2, 1\le s\le S \} $. 
By the construction in \eqref{eq:mu_s_Gamma} and  \eqref{eq:v_s_Gamma}, we must have $\mu_{s}<\tilde{\mu}_{s,\Gamma}$ for $s\in \mathcal{A}$. 
By definition, we then have 
\begin{align*}
    \tilde{\mu}_{\Gamma} - \mu & = 
    \sum_{s=1}^S (\tilde{\mu}_{s, \Gamma}-\mu_{s}) \ge \sum_{s\in \mathcal{A}} (\tilde{\mu}_{s, \Gamma}-\mu_{s}) = 
    |\mathcal{A}| \cdot \frac{1}{|\mathcal{A}|} \sum_{s\in \mathcal{A}} (\tilde{\mu}_{s, \Gamma}-\mu_{s})
\end{align*}
and
\begin{align*}
    \sigma - \tilde{\sigma}_{\Gamma} & = 
    \sqrt{\sum_{s=1}^S v_{s}^2} - \sqrt{\sum_{s=1}^S \tilde{v}_{s,\Gamma}^2} 
    =
    \sqrt{\sum_{s\in \mathcal{A}} v_{s}^2+ \sum_{s\notin \mathcal{A}} v_{s}^2} - \sqrt{\sum_{s\in \mathcal{A}} \tilde{v}_{s,\Gamma}^2 + \sum_{s\notin \mathcal{A}} \tilde{v}_{s,\Gamma}^2}
    \\
    &\le
    \sqrt{\sum_{s\in \mathcal{A}} v_{s}^2+ \sum_{s\notin \mathcal{A}} \tilde{v}_{s,\Gamma}^2} - \sqrt{\sum_{s\in \mathcal{A}} \tilde{v}_{s,\Gamma}^2 + \sum_{s\notin \mathcal{A}} \tilde{v}_{s,\Gamma}^2}
    \\& \le \frac{1}{2\tilde{\sigma}_{\Gamma}}\sum_{s\in \mathcal{A}} ( v_{s}^2 - \tilde{v}_{s,\Gamma}^2)
    =  
    \frac{|\mathcal{A}|}{2\tilde{\sigma}_{\Gamma}} \cdot \frac{1}{|\mathcal{A}|}\sum_{s\in \mathcal{A}} ( v_{s}^2 - \tilde{v}_{s,\Gamma}^2), 
\end{align*}
where the last inequality holds due to \citet[][Lemma 1]{Rosenbaum2000sep}. 
Thus, for any $k \ge 0$, 
\begin{align*}
    & \quad \ ( \tilde{\mu}_{\Gamma} + k \tilde{\sigma}_{\Gamma} ) - (\mu + k \sigma) 
    \\
    & = (\tilde{\mu}_{\Gamma} - \mu) - k(\sigma - \tilde{\sigma}_{\Gamma})
    \ge |\mathcal{A}| \cdot \frac{1}{|\mathcal{A}|} \sum_{s\in \mathcal{A}} (\tilde{\mu}_{s, \Gamma}-\mu_{s}) - \frac{|\mathcal{A}|}{2\tilde{\sigma}_{\Gamma}} \cdot \frac{k}{|\mathcal{A}|}\sum_{s\in \mathcal{A}} ( v_{s}^2 - \tilde{v}_{s,\Gamma}^2)\\
    & =
    |\mathcal{A}| \frac{\Delta_{\Gamma}(v^2) }{\tilde{\sigma}_{\Gamma}}
    \left\{ 
    \frac{\Delta_{\Gamma}(\mu)}{\Delta_{\Gamma}(v^2)} \tilde{\sigma}_{\Gamma} - \frac{k}{2}
    \right\}. 
\end{align*}
From the proof of Lemma \ref{lemma:match_set_clt}, we know that  $\tilde{\sigma}_{\Gamma}^2 \ge 2/\Gamma^3 \cdot \sum_{s=1}^S R_s^2/n_s^3$. 
Therefore, 
\begin{align*}
    ( \tilde{\mu}_{\Gamma} + k \tilde{\sigma}_{\Gamma} ) - (\mu + k \sigma) 
    & \ge 
    |\mathcal{A}|
    \frac{\Delta_{\Gamma}(v^2) }{\tilde{\sigma}_{\Gamma}}
    \left\{ 
    \frac{\Delta_{\Gamma}(\mu)}{\Delta_{\Gamma}(v^2)}\tilde{\sigma}_{\Gamma} - \frac{k}{2} 
    \right\}
    \\
    & \ge |\mathcal{A}|
    \frac{\Delta_{\Gamma}(v^2) }{\tilde{\sigma}_{\Gamma}}
    \left\{ 
    \frac{\sqrt{2}}{\Gamma^{3/2}}
    \frac{\Delta_{\Gamma}(\mu)}{\Delta_{\Gamma}(v^2)}\sqrt{\sum_{s=1}^S R_s^2/n_s^3} - \frac{k}{2} 
    \right\}. 
\end{align*}
Under Condition \ref{cond:sen_asymp_conservative}, 
there must exists $\underline{S}$ such that when $S\ge \underline{S}$, 
$\sqrt{\sum_{s=1}^S R_s^2/n_s^3} \cdot \Delta_{\Gamma}(\mu)/\Delta_{\Gamma}(v^2) \ge \Gamma^{3/2}k/(2\sqrt{2})$ or  $\Delta_{\Gamma}(v^2) = 0$. 
Consequently, when $S\ge \underline{S}$, we must have have 
$( \tilde{\mu}_{\Gamma} + k \tilde{\sigma}_{\Gamma} ) - (\mu + k \sigma) \ge 0$. 
Therefore, the difference between $\tilde{\mu}_{\Gamma} + k \tilde{\sigma}_{\Gamma}$ and $\mu + k \sigma$ must be nonnegative when $S$ is sufficiently large. 
\end{proof}

\begin{proof}[Proof of Theorem \ref{thm:sen_ana_large_sample}]
Let $T = t(\bs{Z}, \bs{Y}(0))$ be the stratified rank score statistic using the true control potential outcomes, 
and $T_s = t(\bs{Z}_s, \bs{Y}_s(0))$ be the corresponding rank score statistic for each stratum $s$. 
By definition, $T = \sum_{s=1}^S T_s$. Recall that $\mu_s = \E(T_s)$ and $v^2_s = \Var(T_s)$ are the true mean and variance of $T_s$. 
Thus, 
by the mutual independence of treatment assignment across all strata, 
the true mean and variance of $T$ are, respectively, $\mu = \sum_{s=1}^S \mu_s$ and $\sigma^2 = \sum_{s=1}^S v_s^2$. 

First, we study the tail probability of $T$ evaluated at $\tilde{\mu}_{\Gamma} + k \tilde{\sigma}_{\Gamma}$ for $k\ge 0$. 
Under the sensitivity model with bias at most $\Gamma$, 
from Lemma \ref{lemma:mu_sigma_tilde_greater}, for any $k\ge 0$, there must exists $\underline{S}_k$ such that $( \tilde{\mu}_{\Gamma} + k \tilde{\sigma}_{\Gamma} ) - (\mu + k \sigma) \ge 0$ when $S\ge \underline{S}_k$. 
Consequently, we have 
$\Pr(T\ge \tilde{\mu}_{\Gamma} + k \tilde{\sigma}_{\Gamma}) \le \Pr(T\ge \mu + k \sigma)$ for $S\ge \underline{S}_k$,  
which immediately implies that
$
    \limsup_{S\rightarrow \infty}\Pr(T\ge \tilde{\mu}_{\Gamma} + k \tilde{\sigma}_{\Gamma}) \le 
    \limsup_{S\rightarrow \infty} \Pr(T\ge \mu + k \sigma).
$
From Condition \ref{cond:match_set_clt} and Lemma \ref{lemma:match_set_clt}, 
$(T-\mu)/\sigma \converged \mathcal{N}(0,1)$. 
Therefore, for any $k\ge 0$, we have 
\begin{align*}
    \limsup_{S\rightarrow \infty}\Pr
    \{ (T-\tilde{\mu}_{\Gamma})/{\tilde{\sigma}_{\Gamma}} \ge k \}
    \le 
    \limsup_{S\rightarrow \infty} \Pr
    \{ (T-\mu)/{\sigma} \ge k \} = 1 - \Phi(k), 
\end{align*}
where $\Phi(\cdot)$ is the distribution function for the standard Gaussian distribution. 

Second, we prove that $\limsup_{S\rightarrow \infty} \Pr( \tilde{p}_{k,c, \Gamma} \le \alpha ) \le \alpha$ for any $\alpha \in (0, 0.5]$ under the null hypothesis $H_{k,c}$. 
From the discussion in the first part, for any fixed $\alpha\in (0, 0.5]$, 
\begin{align*}
    \limsup_{S\rightarrow \infty} \Pr( \tilde{G}_{\Gamma}(T) 
    \le \alpha ) 
    =
    \limsup_{S\rightarrow \infty} \Pr
    \Big\{ 
    \frac{T-\tilde{\mu}_{\Gamma}}{\tilde{\sigma}_{\Gamma}} \ge \Phi^{-1}(1-\alpha) 
    \Big\}
    \le 1 - \Phi\left( \Phi^{-1}(1-\alpha) \right) = \alpha, 
\end{align*}
where $\Phi^{-1}(\cdot)$ is the quantile function for the standard Gaussian distribution. 
Thus, 
when the null hypothesis $H_{k,c}$ holds, i.e., $\bs{\tau} \in \mathcal{H}_{k,c}$, 
we have, for any fixed $\alpha\in (0, 0.5]$,
\begin{align*}
    \limsup_{S\rightarrow \infty} \Pr( \tilde{p}_{k,c, \Gamma} \le \alpha ) = 
    \limsup_{S\rightarrow \infty} \Pr\{ \tilde{G}_{\Gamma}(t_{k,c}) \le \alpha \}
    \le 
    \limsup_{S\rightarrow \infty} \Pr\{ \tilde{G}_{\Gamma}(T) \le \alpha \} \le \alpha, 
\end{align*}
where the second last inequality holds because $t_{k,c} = \inf_{\bs \delta\in\mathcal H_{k,c}}t(\bs Z, \bs Y-\bs Z \circ \bs \delta) \le t(\bs Z, \bs Y-\bs Z \circ \bs \tau) = T$ and $\tilde{G}_{\Gamma}(\cdot)$ is a monotone decreasing function.

Third, we consider the $p$-value $\tilde{p}_{k,c, \Gamma}^\LP$. 
Because $t_{k,c} \ge t_{k,c}^\LP$ and $\tilde{G}_{\Gamma}(\cdot)$ is a monotone decreasing function, 
$\tilde{p}_{k,c, \Gamma} \le \tilde{p}_{k,c, \Gamma}^\LP$. 
From the proof in the second part, under the sensitivity model with bias at most $\Gamma$ and the null hypothesis $H_{k,c}$, for any $0<\alpha \le 0.5$, 
\begin{align*}
    \limsup_{S\rightarrow \infty}\Pr(\tilde{p}_{k,c, \Gamma}^{\LP} \le \alpha) 
    \le 
    \limsup_{S\rightarrow \infty} \Pr(\tilde{p}_{k,c, \Gamma} \le \alpha)
    \le \alpha. 
\end{align*}

From the above, Theorem \ref{thm:sen_ana_large_sample} holds.
\end{proof}

\subsection{Proof of Theorem \ref{thm:sen_inv}}

Below we first give a complete version of Theorem \ref{thm:sen_inv}. 

\begin{theorem}\label{thm:sen_inv_comp}
Under the sensitivity model with bias at most $\Gamma$,
the 
randomization 
$p$-value 
$\tilde{p}_{k,c, \Gamma}$ 
is increasing in $c$ and decreasing in $k$. 
Moreover, 
if Conditions \ref{cond:match_set_clt} and \ref{cond:sen_asymp_conservative} hold, then 
for any $\alpha\in (0, 0.5]$, 
\begin{enumerate}[label=(\roman*)]
    \item $\tilde{\mathcal{I}}_{\alpha}(k) \equiv \{c: \tilde{p}_{k,c, \Gamma} > \alpha, c \in \mathbb{R}\}$ is an asymptotic $1-\alpha$ confidence set for $\tau_{(k)}$, and it must be an interval of the form $(\underline{c}, \infty)$ or $[\underline{c}, \infty)$, with $\underline{c} = \inf\{c: \tilde{p}_{k,c, \Gamma} > \alpha, c \in \mathbb{R}\}$; 
    \item $\tilde{\mathcal{S}}_{\alpha}(c) \equiv \{N-k: \tilde{p}_{k,c, \Gamma} > \alpha, 0 \le k \le N \}$ is an asymptotic $1-\alpha$ confidence set for $n(c)$, and it must have the form of $\{j: N-\overline{k} \le j \le N\}$, 
    with 
    $\overline{k} = \sup\{k: \tilde{p}_{k,c, \Gamma} > \alpha, 0 \le k \le N\}$;
    \item the intersection of all asymptotic $1-\alpha$ confidence sets for $\tau_{(k)}$'s, viewed as a confidence set for $\bs{\tau}$, is the same as that for $n(c)$'s, 
    and they have the following equivalent forms:
    \begin{align*}
        \big\{ \bs{\delta} \in \mathbb{R}^N: \delta_{(k)} \in \tilde{\mathcal{I}}_{\alpha}(k), 1\le k \le N \big\}
        & = 
        \Big\{
        \bs{\delta} \in \mathbb{R}^N: 
        \sum_{i=1}^N \I(\delta_i > c) \in \tilde{\mathcal{S}}_{\alpha}(c), c \in \mathbb{R}
        \Big\}
        \\
        & 
        =  \bigcap_{k, c: \tilde{p}_{k,c, \Gamma} \le \alpha} 
        \mathcal{H}_{k,c}^{\c}.
    \end{align*}
    Moreover, the resulting confidence set covers the true individual treatment effect vector $\bs{\tau}$ with probability asymptotically at least $1-\alpha$, i.e., 
    \begin{align*}
        \liminf_{S \rightarrow \infty}\Pr\Big(
        \bs{\tau} \in \bigcap_{k, c: \tilde{p}_{k,c, \Gamma} \le \alpha} 
        \mathcal{H}_{k,c}^{\c}
        \Big) 
        \ge 1-\alpha. 
    \end{align*}
    \item the above results also hold for the $p$-value $\tilde{p}_{k,c, \Gamma}^\LP$ from the relaxed linear programming. 
\end{enumerate}
\end{theorem}

\begin{proof}[Proof of Theorem \ref{thm:sen_inv_comp}]
First, note that $\tilde{p}_{k,c,\Gamma}$ is the tail probability $\tilde{G}_{\Gamma}(\cdot)$ evaluated at $t_{k,c}$. 
The monotonicity of $\tilde{p}_{k,c,\Gamma}$ in $k$ and $c$ follows immediately from Lemma \ref{lemma:monotone_t}.

Second, Theorem \ref{thm:sen_inv_comp}(i) and (ii) on the asymptotic validity of confidence sets for the quantiles of individual effects and number of units with effects greater than $c$ follow immediately from Theorem \ref{thm:sen_ana_large_sample}.

Third, we prove Theorem \ref{thm:sen_inv_comp}(iii). 
The equivalent form of the confidence set follows by the same logic as the proof of Theorem \ref{thm:conf_set}. 
Below we prove the asymptotic validity of this confidence set. 
Suppose that $\bs{\tau} \notin \bigcap_{k, c: \tilde{p}_{k,c, \Gamma} \le \alpha} \mathcal{H}_{k,c}^{\c}$. 
Then there must exist $k$ and $c$ such that 
$\bs{\tau} \in \mathcal{H}_{k,c}$ and $\tilde{p}_{k,c, \Gamma} \le \alpha$. 
By definition, we must have 
$t_{k,c} = \inf_{\bs{\delta}\in \mathcal{H}_{k,c}} t(\bs{Z}, \bs{Y} - \bs{Z}\circ \bs{\delta})
\le t(\bs{Z}, \bs{Y} - \bs{Z}\circ \bs{\tau}) = t(\bs{Z}, \bs{Y}(0)),
$
and thus 
$
   \alpha \ge  \tilde{p}_{k,c, \Gamma} 
   = 
   \tilde{G}_{\Gamma} (t_{k,c}) 
   \ge 
   \tilde{G}_{\Gamma} (t(\bs{Z}, \bs{Y}(0))). 
$
By the definition of $\tilde{G}_{\Gamma}(\cdot)$, this further implies that 
$t(\bs{Z}, \bs{Y}(0)) \ge \tilde{\mu}_{\Gamma} + \Phi^{-1}(1-\alpha) \cdot \tilde{\sigma}_{\Gamma}$. 
From Lemma \ref{lemma:mu_sigma_tilde_greater}, 
there exists $\underline{S}_{\alpha}$ such that 
when $S$ is sufficiently large, 
$
\tilde{\mu}_{\Gamma} + \Phi^{-1}(1-\alpha) \cdot \tilde{\sigma}_{\Gamma}
\ge \mu + \Phi^{-1}(1-\alpha) \cdot \sigma. 
$
Thus, for sufficiently large $S$, 
$\bs{\tau} \notin \bigcap_{k, c: \tilde{p}_{k,c, \Gamma} \le \alpha} 
        \mathcal{H}_{k,c}^{\c}$ 
must imply that $t(\bs{Z}, \bs{Y}(0)) \ge \mu + \Phi^{-1}(1-\alpha) \cdot \sigma$. 
Consequently, for sufficiently large $S$, we have 
\begin{align*}
    & \quad \ \Pr\Big(
        \bs{\tau} \in \bigcap_{k, c: \tilde{p}_{k,c, \Gamma} \le \alpha} 
        \mathcal{H}_{k,c}^{\c}
        \Big) 
        \\
    & = 
    1 - \Pr\Big(
        \bs{\tau} \notin \bigcap_{k, c: \tilde{p}_{k,c, \Gamma} \le \alpha} 
        \mathcal{H}_{k,c}^{\c}
        \Big) 
    \ge
    1 - \Pr\Big\{
    t(\bs{Z}, \bs{Y}(0)) \ge \mu + \Phi^{-1}(1-\alpha) \cdot \sigma
     \Big\}
     \\
     & 
    =
    1 - \Pr\Big\{
    \frac{ t(\bs{Z}, \bs{Y}(0)) - \mu}{\sigma} \ge \Phi^{-1}(1-\alpha) \Big\}. 
\end{align*}
From Lemma \ref{lemma:match_set_clt}, letting $S\rightarrow \infty$, we then have  
\begin{align*}
    \liminf_{S\rightarrow \infty}\Pr\Big(
        \bs{\tau} \in \bigcap_{k, c: \tilde{p}_{k,c, \Gamma} \le \alpha} 
        \mathcal{H}_{k,c}^{\c}
        \Big)
    & \ge 
    1 - \limsup_{S\rightarrow \infty}\Pr\Big\{
    \frac{ t(\bs{Z}, \bs{Y}(0)) - \mu}{\sigma} \ge \Phi^{-1}(1-\alpha) \Big\}
    \\
    & = \Phi(\Phi^{-1}(1-\alpha)) = 1 - \alpha. 
\end{align*}

Fourth, 
from Lemma \ref{lemma:monotone_t_lp}, 
Theorem \ref{thm:sen_inv_comp}(iv) follows
by the same logic as Theorem \ref{thm:sen_inv_comp}(i)--(iii).

From the above, Theorem \ref{thm:sen_inv_comp} holds. 
\end{proof}

\subsection{Comment on switching treatment labels and change outcome signs for sensitivity analysis in matched observation studies}

Recall that $\bs Z$ and $\bs Y$ are our original treatment assignment and outcome vectors, 
$\bs{Y}(1)$, $\bs{Y}(0)$ and $\bs{\tau}$ are the corresponding treatment potential outcome, control potential outcome and individual effect vectors. 
After switching the treatment labels and changing the signs of the outcomes, 
the treatment assignment and observed outcome vectors become $\check{\bs Z} = \bs{1}-\bs Z$, $\check{\bs Y} = -\bs Y$, 
the corresponding potential outcomes become 
$\check{\bs{Y}}(1) = - \bs{Y}(0)$ and $\check{\bs{Y}}(0) = - \bs{Y}(1)$, and the corresponding individual effect vector becomes $\check{\bs{\tau}} = \check{\bs{Y}}(1) - \check{\bs{Y}}(0) =  \bs{\tau}$. 
We can verify that 
$\bs{\check{Y}} =  \check{\bs{Z}} \circ \check{\bs{Y}}(1) + (\bs{1}-\check{\bs{Z}}) \circ \check{\bs{Y}}(0)$. 

First, we consider the tail probability of $t\{\check{\bs{Z}}, \check{\bs{Y}}(0)\}$ when $\bs{Z}$ follows the sensitivity model with bias at most $\bs{\Gamma}$. 
By definition, 
for $1\le s \le S$,
\begin{align*}
t_s(\check{\bs Z_s}, \check{\bs Y_s}(0))
& = 
\sum_{i=1}^{n_s} \I(\check{Z}_{si}=1) \phi_s(\rank_i(\check{\bs Y}_s(0))) 
= 
\sum_{i=1}^{n_s} \I(Z_{si}=0) \phi_s(\rank_i(\check{\bs Y}_s(0))) 
\\
& = 
\sum_{i=1}^{n_s} \phi_s(i) 
- 
\sum_{i=1}^{n_s} \I(Z_{si}=1) \phi_s(\rank_i(\check{\bs Y}_s(0))) 
\\
& = 
\sum_{i=1}^{n_s} \phi_s(i) + 
\sum_{i=1}^{n_s} \I(Z_{si}=1) \check{\phi}_s(\rank_i(\check{\bs Y}_s(0))), 
\end{align*}
where $\check{\phi}_s(j) = -\phi_s(j)$ for $1\le j\le n_s$. Consequently, 
\begin{align*}
    t(\check{\bs{Z}}, \check{\bs{Y}}(0))
    & = 
    \sum_{s=1}^S \sum_{i=1}^{n_s} \phi_s(i) + 
    \sum_{s=1}^S \sum_{i=1}^{n_s} \I(Z_{si}=1) \check{\phi}_s(\rank_i(\check{\bs Y}_s(0))). 
\end{align*}
Let $\check{\mu}$ and $\check{\sigma}$ be the true mean and standard deviation of $\sum_{s=1}^S \sum_{i=1}^{n_s} \I(Z_{si}=1) \check{\phi}_s(\rank_i(\check{\bs Y}_s(0)))$, 
and let $\check{\tilde{\mu}}_{\Gamma}$ and $\check{\tilde{\sigma}}_{\Gamma}$ be the maximized mean and maximized standard deviation (given that mean is maximized) of $\sum_{s=1}^S \sum_{i=1}^{n_s} \I(Z_{si}=1) \check{\phi}_s(i)$ given that $\bs{Z}$ follows the sensitivity model with bias at most $\Gamma$, which can be calculated similarly as in \eqref{eq:mu_v2_Gamma}. 
Then by the same logic as Lemmas \ref{lemma:match_set_clt} and \ref{lemma:mu_sigma_tilde_greater}, under certain regularity conditions, 
\begin{align}\label{eq:check_clt}
    \frac{
    \sum_{s=1}^S \sum_{i=1}^{n_s} \I(Z_{si}=1) \check{\phi}_s(\rank_i(\check{\bs Y}_s(0))) - \check{\mu}
    }{\check{\sigma}}
    \converged \mathcal{N}(0,1), 
\end{align}
and for any $k \ge 0$, when $S$ is sufficiently large, 
\begin{align}\label{eq:check_greater}
    \check{\tilde{\mu}}_{\Gamma} + k \check{\tilde{\sigma}}_{\Gamma}
    \ge \check{\mu} + k \check{\sigma}. 
\end{align}


Second, we construct $p$-values and demonstrate their large-sample validity at significance level $\alpha\in (0, 0.5]$ given that $\bs{Z}$ follows the sensitivity model with bias at most $\Gamma$. 
Let $\check{\tilde{G}}_{\Gamma}(\cdot)$ denote the tail probability for the Gaussian distribution with mean $\sum_{s=1}^S \sum_{i=1}^{n_s} \phi_s(i) + \check{\tilde{\mu}}_{\Gamma}$ and standard deviation $\check{\tilde{\sigma}}_{\Gamma}$.
Let $\check{t}_{k,c}$ be the minimum of $t(\check{\bs{Z}}, \check{\bs{Y}} - \check{\bs{Z}}\circ \bs{\delta})$ over $\bs{\delta} \in \mathcal{H}_{k,c}$, which is equivalent to an integer linear programming as in \eqref{eq:integer_program}, 
and $\check{t}_{k,c}^{\LP}$ be the minimum objective value from the relaxed linear programming as in \eqref{eq:linear_program}. 
Define 
\begin{align}\label{eq:p_check_tilde}
    \check{\tilde{p}}_{k,c, \Gamma}
    \equiv 
    \check{\tilde{G}}_{\Gamma}(\check{t}_{k,c})
    \le 
    \check{\tilde{p}}^{\LP}_{k,c, \Gamma}
    \equiv 
    \check{\tilde{G}}_{\Gamma}(\check{t}^{\LP}_{k,c}),
\end{align}
where the inequalities holds because $\check{t}_{k,c} \ge \check{t}_{k,c}^{\LP}$ and  $\check{\tilde{G}}_{\Gamma}(\cdot)$ is a monotone decreasing function. 
To prove the asymptotic validity of the $p$-values in \eqref{eq:p_check_tilde}, it suffices to show the validity for the smaller one $\check{\tilde{p}}_{k,c, \Gamma}$. By definition, for any $\alpha\in (0, 0.5]$, 
\begin{align*}
    \Pr( \check{\tilde{p}}_{k,c, \Gamma} \le \alpha )
    & 
    = 
    \Pr\big\{ \check{\tilde{G}}_{\Gamma}(\check{t}_{k,c}) \le \alpha \big\}
    = 
    \Pr\Big\{ 
    \check{t}_{k,c} - \sum_{s=1}^S \sum_{i=1}^{n_s} \phi_s(i) \ge  \check{\tilde{\mu}}_{\Gamma} + \Phi^{-1}(1-\alpha) \check{\tilde{\sigma}}_{\Gamma}
    \Big\}
    \\
    & \le 
    \Pr\Big\{
    \sum_{s=1}^S \sum_{i=1}^{n_s} \I(Z_{si}=1) \check{\phi}_s(\rank_i(\check{\bs Y}_s(0))) \ge  \check{\tilde{\mu}}_{\Gamma} + \Phi^{-1}(1-\alpha) \check{\tilde{\sigma}}_{\Gamma}
    \Big\}. 
\end{align*}
From \eqref{eq:check_greater}, when $S$ is sufficiently larger, we have 
\begin{align*}
    \Pr( \check{\tilde{p}}_{k,c, \Gamma} \le \alpha )
    & 
    \le 
    \Pr\Big\{ 
    \sum_{s=1}^S \sum_{i=1}^{n_s} \I(Z_{si}=1) \check{\phi}_s(\rank_i(\check{\bs Y}_s(0))) \ge  \check{\tilde{\mu}}_{\Gamma} + \Phi^{-1}(1-\alpha) \check{\tilde{\sigma}}_{\Gamma}
    \Big\}
    \\
    & 
    \le 
    \Pr\Big\{  
    \sum_{s=1}^S \sum_{i=1}^{n_s} \I(Z_{si}=1) \check{\phi}_s(\rank_i(\check{\bs Y}_s(0))) \ge  \check{\mu} + \Phi^{-1}(1-\alpha) \check{\sigma}
    \Big\}
    \\
    & = 
    \Pr\Big\{  
    \frac{\sum_{s=1}^S \sum_{i=1}^{n_s} \I(Z_{si}=1) \check{\phi}_s(\rank_i(\check{\bs Y}_s(0))) - \check{\mu}}{\check{\sigma}} \ge \Phi^{-1}(1-\alpha) 
    \Big\}. 
\end{align*}
From \eqref{eq:check_clt}, letting $S\rightarrow \infty$, we then have
\begin{align*}
    & \quad \ \limsup_{S\rightarrow \infty}\Pr( \check{\tilde{p}}_{k,c, \Gamma} \le \alpha )
    \\
    & \le 
    \limsup_{S\rightarrow \infty}
    \Pr\Big\{ 
    \frac{\sum_{s=1}^S \sum_{i=1}^{n_s} \I(Z_{si}=1) \check{\phi}_s(\rank_i(\check{\bs Y}_s(0))) - \check{\mu}}{\check{\sigma}} \ge \Phi^{-1}(1-\alpha) 
    \Big\}\\
    & = 1 - \Phi(\Phi^{-1}(1-\alpha)) = \alpha. 
\end{align*}

Fourth, by the same logic as the proof of Theorem \ref{thm:sen_inv_comp}, 
the resulting confidence sets by inverting the $p$-values $\check{\tilde{p}}_{k,c, \Gamma}$ or $\check{\tilde{p}}_{k,c, \Gamma}^{\LP}$ 
for quantiles of individual effects and numbers of units with effects greater than $c$
are asymptotically simultaneously valid. 

\subsection{Proof of Proposition \ref{prop:Psi_opt}}
We first prove that $\tilde{\Psi}(\cdot)$ is a monotone dominating transformation.
Let $\tilde{\bs{a}} = (\tilde{a}_1, \ldots, \tilde{a}_n)^\top \equiv \tilde{\Psi}(\bs{a})$. 
First, we show $\tilde{\bs{a}}$ satisfies property (i) in Definition \ref{def:mono_transform}. 
For any $2 \le m \le n$, 
by definition, 
there must exist $m \le j_m \le n$ such that
$
\tilde{a}_m = (j_m-m+1)^{-1} ( \sum_{i=1}^{j_m} a_i - \sum_{i=1}^{m-1}\tilde{a}_i ), 
$
which implies that 
\begin{align*}
    & \quad \ (j_m-m+1)(\tilde{a}_m-\tilde{a}_{m-1})
    \\
    &
    =
    \sum_{i=1}^{j_m} a_i - \sum_{i=1}^{m-1}\tilde{a}_i - (j_m-m+1) \tilde{a}_{m-1}
    =\sum_{i=1}^{j_m}a_i-\sum_{i=1}^{m-2}\tilde{a}_i-(j_m-m+2)\tilde{a}_{m-1}
    \\
    &=
    (j_m-m+2)\left\{ \frac{\sum_{i=1}^{j_m}a_i-\sum_{i=1}^{m-2}\tilde{a}_i}{j_m-(m-1)+1} -\tilde{a}_{m-1}\right\}
    \\
    & = 
    (j_m-m+2)\left\{ \frac{\sum_{i=1}^{j_m}a_i-\sum_{i=1}^{m-2}\tilde{a}_i}{j_m-(m-1)+1} -
    \max_{m-1\le j \le n} 
    \frac{\sum_{i=1}^{j}a_i-\sum_{i=1}^{m-2}\tilde{a}_i}{j-(m-1)+1}
    \right\}
    \\
    &\le 0.
\end{align*}
Thus, $\tilde{a}_m \le \tilde{a}_{m-1}$ for $2\le m \le n$, i.e., $\tilde{a}_1 \ge \tilde{a}_2 \ge \ldots \ge \tilde{a}_n$ is a monotone decreasing sequence. 
Second, we show $\tilde{\bs{a}}$ satisfies property (ii) in Definition \ref{def:mono_transform}. 
It is obvious that $\tilde{a}_1= \max_{1\le j \le n}j^{-1}\sum_{t=1}^j a_t\ge a_1$. 
For $2\le m \le n$, by definition, 
\begin{align*}
    \tilde{a}_m=\max_{m \le j \le n} \frac{\sum_{i=1}^j a_i - \sum_{i=1}^{m-1} \tilde{a}_i}{j-m+1} 
    \ge \frac{\sum_{i=1}^{m} a_i - \sum_{i=1}^{m-1}\tilde{a}_i}{m-m+1}
    = \sum_{i=1}^{m} a_i - \sum_{i=1}^{m-1}\tilde{a}_i, 
\end{align*}
which immediately implies that $\sum_{i=1}^{m} \tilde{a}_i \ge \sum_{i=1}^{m} a_i$ for $2\le m\le n$.
Therefore, $\tilde{\Psi}(\cdot)$ is a monotone dominating transformation.

We then prove the optimality of $\tilde{\Psi}(\cdot)$. 
For any monotone dominating transformation $\Psi(\cdot)$ and any $\bs{a}\in \mathbb{R}^n$, 
we prove by induction that $\sum_{i=1}^j \tilde{\Psi}_i(\bs{a}) \le \sum_{i=1}^j \Psi_i(\bs{a})$ for all $1\le j \le n$.
First, 
by the two properties in Definition \ref{def:mono_transform}, 
\begin{align*}
     \sum_{i=1}^j a_i\le \sum_{i=1}^j \Psi_i(\bs{a}) \le j\Psi_1(\bs{a}),\ \forall 1\le j \le n \Longrightarrow \Psi_1(\bs{a})\ge \max_j \frac 1j\sum_{i=1}^ja_i=\tilde{\Psi}_1(\bs{a}). 
\end{align*}
Second, assume that $\sum_{i=1}^t {\Psi}_i(\bs{a}) \ge \sum_{i=1}^t \tilde{\Psi}_i(\bs{a})$, for $1\le t \le m-1$. When $t=m$, for any $m \le j \le n$, by the two properties in Definition \ref{def:mono_transform}, 
\begin{align*}
    (j-m+1){\Psi}_m(\bs{a})+\sum_{i=1}^{m-1}{\Psi}_i(\bs{a})\ge \sum_{i=m}^{j}{\Psi}_i(\bs{a})+\sum_{i=1}^{m-1}{\Psi}_i(\bs{a})=\sum_{i=1}^{j}{\Psi}_i(\bs{a}) \ge \sum_{i=1}^j a_i.
\end{align*}
This implies that ${\Psi}_m(\bs{a}) \ge \max_{m \le j \le n} (j-m+1)^{-1} \{ \sum_{i=1}^j a_i - \sum_{i=1}^{m-1}\Psi_i(\bs{a})\}$. 
Consequently, 
\begin{align*}
   \sum_{i=1}^m {\Psi}_i(\bs{a}) &= \sum_{i=1}^{m-1} {\Psi}_i(\bs{a}) + {\Psi}_m(\bs{a})
   \ge 
   \sum_{i=1}^{m-1} {\Psi}_i(\bs{a}) + \max_{m \le j \le n}
   \frac{\sum_{i=1}^j a_i - \sum_{i=1}^{m-1}\Psi_i(\bs{a})}{j-m+1}
   \\& =
   \max_{m \le j \le n} \left( \frac{\sum_{i=1}^j a_i}{j-m+1}  + \frac{(j-m)\sum_{i=1}^{m-1} {\Psi}_i(\bs{a})}{j-m+1}  \right)
   \\& \ge
   \max_{m \le j \le n} \left( \frac{\sum_{i=1}^j a_i}{j-m+1}  + \frac{(j-m)\sum_{i=1}^{m-1}\tilde{\Psi}_i(\bs{a})}{j-m+1} \right)
   \\& =
   \sum_{i=1}^{m-1} \tilde{\Psi}_i(\bs{a}) + 
   \max_{m \le j \le n} \frac{\sum_{i=1}^j a_i - \sum_{i=1}^{m-1}\tilde{\Psi}_i(\bs{a})}{j-m+1}
   \\& = 
   \sum_{i=1}^m \tilde{\Psi}_i(\bs{a}), 
\end{align*}
where the last inequality holds because 
$\sum_{i=1}^{m-1} \Psi_i(\bs{a}) \ge \sum_{i=1}^{m-1} \tilde{\Psi}_i(\bs{a})$ 
and the last equality holds by the definition of $\tilde{\Psi}(\cdot)$ in \eqref{eq:Psi_optimal}. 

From the above, Proposition $\ref{prop:Psi_opt}$ holds.

\subsection{Proof of Theorem \ref{thm:sen_conservative_finte_sample}}
We first show that $G_{\Gamma}(c) \le \overline{G}_{\Gamma}(c)$ for all $c\in \mathbb{R}$. 
Consider any fixed $\bs{y} \in \mathbb{R}^N$.  
For each $1\le s\le S$ and $c\in \mathbb{R}$, 
let $\mathcal{Q}_{s}(c) = \{i: \phi_s( \rank_i(\bs{y}_s) ) \ge c, 1\le i \le n_s\}$ be the set of indices in stratum $s$ with transformed rank greater than or equal to $c$. 
Let $\Pr_{\bs{u}, \Gamma}(\cdot)$ denote the probability measure for the treatment assignment vector $Z$ under the sensitivity model with bias at most $\Gamma$ and unmeasured confounding $\bs{u}$. 
Then 
the tail probability of the statistic $t(\bs{Z}_s, \bs{y}_s)$ under the sensitivity model with bias at most $\Gamma$ can be bounded by 
\begin{align*}
    {\rev \Pr_{\bs{u}, \Gamma}\left\{t(\bs{Z}_s, \bs{y}_s) \ge c\right\}}
    & = 
    \frac{\sum_{i\in \mathcal{Q}_{s}(c)} \exp(\gamma u_{si}) }{\sum_{i\notin \mathcal{Q}_{s}(c)} \exp(\gamma u_{si}) + \sum_{i\in \mathcal{Q}_{s}(c)} \exp(\gamma u_{si})}
    \\
    & \le 
    \frac{\sum_{i\in \mathcal{Q}_{s}(c)} \exp(\gamma u_{si}) }{n_s - |\mathcal{Q}_{s}(c)| + \sum_{i\in \mathcal{Q}_{s}(c)} \exp(\gamma u_{si})}
    \le 
    \frac{|\mathcal{Q}_{s}(c)| \Gamma}{n_s - |\mathcal{Q}_{s}(c)| + |\mathcal{Q}_{s}(c)| \Gamma}, 
\end{align*}
where the last inequality holds because the function $x/(1+x)$ is increasing in $x\ge 0$. 
Note that the tail probability of $\overline{T}_s$ has the following equivalent forms:
\begin{align*}
    \Pr(\overline{T}_s \ge c)
    & = 
    \begin{cases}
    1, & \text{if } c\le \xi_{s1} \\
    g_{si}\Gamma/\{ (n_s-g_{si}) + g_{si}\Gamma \}, 
    & \text{if } \xi_{s, i-1} < c \le \xi_{si}, 2\le i \le m_s
    \\
    0, & \text{if } c > \xi_{sm_s}
    \end{cases}
    \\
    & = 
    \frac{|\mathcal{Q}_{s}(c)| \Gamma}{n_s - |\mathcal{Q}_{s}(c)| + |\mathcal{Q}_{s}(c)| \Gamma}. 
\end{align*}
Thus, $t(\bs{Z}_s, \bs{y}_s)$ must be stochastically smaller than or equal to $\overline{T}_s$. 
Because both $t(\bs{Z}_s, \bs{y}_s)$'s and $\overline{T}_s$'s are mutually independent across all $s$, 
$t(\bs{Z}, \bs{y}) = \sum_{s=1}^S t(\bs{Z}_s, \bs{y}_s)$ must be stochastically smaller than or equal to $\sum_{s=1}^S \overline{T}_s$. 
Therefore, by definition and from \eqref{eq:G_Gamma}, 
for any $c\in \mathbb{R}$, 
\begin{align*}
    G_{\Gamma}(c) 
    = 
    \max_{\bs{u} \in \mathcal{U}} {\rev
    \Pr_{\bs{u}, \Gamma}\left\{ t(\bs{Z}, \bs{y}) \ge c \right\}}
    \le 
    \Pr\Big( \sum_{s=1}^S \overline{T}_s \ge c \Big) = \overline{G}_{\Gamma}(c)
\end{align*}
This implies that 
$\overline{p}_{k,c, \Gamma} \equiv \overline{G}_{\Gamma}(t_{k,c}) \ge G_{\Gamma}(t_{k,c}) = p_{k,c, \Gamma}$. 
From Theorem \ref{thm:sen_ana}, 
$\overline{p}_{k,c, \Gamma}$ is a valid $p$-value for testing $H_{k,c}$ under the sensitivity model with bias at most $\Gamma$. 

Because $t_{k, c} \ge t_{k, c}^\LP$ and $\overline{G}_{\Gamma}(\cdot)$ is a monotone decreasing function, 
$\overline{p}_{k, c, \Gamma}^\LP \ge \overline{p}_{k, c, \Gamma}$, and it must also be a valid $p$-value for testing $H_{k,c}$ under the sensitivity model with bias at most $\Gamma$.

From the above, Theorem \ref{thm:sen_conservative_finte_sample} holds.

\subsection{Comment on equality of $G_\Gamma(\cdot)$ and $\overline{G}_\Gamma(\cdot)$ in matched pair study}\label{sec:match_pair_Gamma}

In matched pair studies, we have exactly two units within each stratum. 
For any $\bs{y}\in \mathbb{R}^N$, define $\bs{u} = (u_{11}, u_{12}, \ldots, u_{S2})$ such that $u_{s1} = 1- u_{s2} =  \I(y_{s1} \ge y_{s2})$ for all $1\le s \le S$. 
Then by the definition in \eqref{eq:G_Gamma}, for any $c\in \mathbb{R}$, we have 
\begin{align*}
   G_{\Gamma}(c) 
   & \ge 
   \sum_{\bs{a} \in \mathcal{Z}} 
    \prod_{s=1}^S
    \frac{\exp(\gamma \sum_{i=1}^{n_s} a_{si} u_{si})}{
    \sum_{i=1}^{n_s} \exp(\gamma u_{si})}
    \I\{
    t( \bs{a}, \bs{y} )
    \ge 
    c
    \}
    = 
    \sum_{\bs{a} \in \mathcal{Z}} \Pr_{\bs{u}, \bs{\Gamma}}(\bs{A}=\bs{a}) \I \{ t(\bs{a}, \bs{y}) \ge c \}
    \\
    & =
    \Pr_{\bs{u}, \bs{\Gamma}} \{ t(\bs{A}, \bs{y}) \ge c \}
    = 
    \Pr_{\bs{u}, \bs{\Gamma}} \Big\{ \sum_{s=1}^S t_s(\bs{A}_s, \bs{y}_s) \ge c \Big\}, 
\end{align*}
where 
$\Pr_{\bs{u}, \Gamma}(\cdot)$ denotes the probability measure for the treatment assignment vector from the sensitivity model with bias at most $\Gamma$ and unmeasured confounding $\bs{u}$, 
$\bs{A}$ denotes the treatment assignment with probability mass function 
\begin{align*}
    \Pr_{\bs{u}, \bs{\Gamma}}(\bs{A}=\bs{a})
    & = \prod_{s=1}^S
    \frac{\exp(\gamma \sum_{i=1}^{n_s} a_{si} u_{si})}{
    \sum_{i=1}^{n_s} \exp(\gamma u_{si})},
\end{align*}
and $\bs{A}_s$ and $\bs{y}_s$ are subvectors of $\bs{A}$ and $\bs{y}$ corresponding to stratum $s$. 
We can verify that, for any $s$, 
if $\phi_s(1)\ne \phi_s(2)$, then 
$t_s(\bs{A}_s, \bs{y}_s)$ takes value $\phi_s(2)$ with probability $\Gamma/(1+\Gamma)$ and value $\phi_s(1)$ with probability $1/(1+\Gamma)$; otherwise $t_s(\bs{A}_s, \bs{y}_s)$ is a constant $\phi_s(1) = \phi_s(2)$. 
This implies that $t_s(\bs{A}_s, \bs{y}_s)$ has the same distribution as $\overline{T}_s$. 
Note that both $t_s(\bs{A}_s, \bs{y}_s)$ and $\overline{T}_s$ are mutually independent across all $s$. We then have, for any $c\in \mathbb{R}$, 
\begin{align*}
    G_{\Gamma}(c) & \ge \Pr_{\bs{u}, \bs{\Gamma}} \Big\{ \sum_{s=1}^S t_s(\bs{A}_s, \bs{y}_s) \ge c \Big\}
    = 
    \Pr\Big(\sum_{s=1}^S \overline{T}_s \ge c \Big) = \overline{G}_{\Gamma}(c).
\end{align*}
Note that from the proof of Theorem \ref{thm:sen_conservative_finte_sample}, 
$G_{\Gamma}(c) \le \overline{G}_{\Gamma}(c)$ for all $c\in \mathbb{R}$.
Therefore, we must have $G_{\Gamma}(c) = \overline{G}_{\Gamma}(c)$ for all $c\in \mathbb{R}$.

\subsection{Proof of Theorem \ref{thm:sen_inv_finite} and Remark \ref{rmk:sen_inv_finite_LP}}

\begin{proof}[Proof of Theorem \ref{thm:sen_inv_finite}]
First, note that $\overline{p}_{k,c,\Gamma}$ is the tail probability $\overline{G}_{\Gamma}(\cdot)$ evaluated at $t_{k,c}$. 
The monotonicity of $\overline{p}_{k,c,\Gamma}$ in $k$ and $c$ follows immediately from Lemma \ref{lemma:monotone_t}.

Second, Theorem \ref{thm:sen_inv_finite}(i) and (ii) on the validity of confidence sets for the quantiles of individual effects and number of units with effects greater than $c$ follow immediately from Theorem \ref{thm:sen_conservative_finte_sample}. 

Third, the simultaneous validity in Theorem \ref{thm:sen_inv_finite}(iii) holds by almost the same logic as 
Theorem \ref{thm:conf_set}(iii), and is thus omitted here. 
\end{proof}

\begin{proof}[Proof of Remark \ref{rmk:sen_inv_finite_LP}]
From Lemma \ref{lemma:monotone_t_lp}, 
Remark \ref{rmk:conf_set_plus} follows
by the same logic as Theorem \ref{thm:sen_inv_finite}. 
\end{proof}

\section{Additional technical details for optimizations}\label{sec:add_opi}

\subsection{A numerical example illustrating the drawback of the naive greedy algorithm}\label{sec:numerical_example}

Consider a stratified randomized experiment with three strata, where each stratum contains three treated and three control units, 
and suppose the observed data are as 
in Table \ref{tab:greedy_counter}. 
Based on the observed data, we calculate the values of the $\Delta_{s, 0}(j)$'s using the stratified Stephenson rank sum statistic with $h_1 = h_2 = h_3 = 4$, 
and display them in Table \ref{tab:greedy_counter}, 
where $\Delta_{s, 0}(j)$'s for $j \ge 4$ are all zero and are thus omitted from the table. From Table \ref{tab:greedy_counter}, 
to maximize the value of $\sum_{s=1}^3 \sum_{j=1}^{l_s} \Delta_{s, 0}(j)$ under the constraint that 
$l_1+l_2+l_3 = 2$, 
the naive greedy algorithm will pick first $\Delta_{3,0}(1) = 6$ and then $\Delta_{1,0}(1) = 3$, resulting a value of $\Delta_{3,0}(1) + \Delta_{1,0}(1) =9$. 
However, 
from Table \ref{tab:greedy_counter}, 
it is obvious that the maximum value 
is actually 
$\Delta_{1,0}(1) + \Delta_{1,0}(2) = 10$, 
greater 
than the one from the naive greedy algorithm. 
Table \ref{tab:greedy_example_pval} shows the maximum value of $\sum_{s=1}^3 \sum_{j=1}^{l_s} \Delta_{s, c}(j)$ over 
$(l_1, l_2, l_3) \in \mathcal{K}_{\bs{n}}(N-k)$ achieved by the naive greedy algorithm, integer linear programming and linear programming at various values of $N-k$, 
as well as the corresponding $p$-values, which are the values of $G(\cdot)$ in \eqref{eq:G} evaluated at the minimum value of  $\sum_{s=1}^3 t_{s, c} (0)-\sum_{s=1}^3 \sum_{j=1}^{l_s} \Delta_{s, c}(j)$ achieved by these three algorithms. 
From Table \ref{tab:greedy_example_pval}, 
the greedy algorithm sometimes cannot reach the global optimum, and thus provides $p$-value strictly less than the valid $p$-value $p_{k,c}$, which can cause the tests to be statistically invalid. 
As a side note, the linear programming and the greedy algorithm with the optimal transformation always provide equal $p$-values that are larger than or equal to $p_{k,c}$, and thus the corresponding tests are still valid, as discussed in 
\S \ref{sec:ILP_LP} and \S \ref{sec:greedy}. 
\begin{table}[htbp]
    \centering
    \caption{A numerical example 
    for illustrating 
    the drawback of the greedy algorithm. There are in total three strata, as indicated in the 1st column. 
    The 2nd and 3nd columns show the observed outcomes for the three treated and control units, respectively, within each stratum. 
    The 4th to 6th columns show the values of $\Delta_{s,0}(j)$'s for $1\le j\le 3$. 
    Note that 
    here $\Delta_{s,0}(j) = 0$ for $j \ge 4$. 
    }\label{tab:greedy_counter}
    \begin{tabular}{cccccc}
    \toprule
    stratum     & treated outcomes & control outcomes & $\Delta_{s, 0}(1)$ & $\Delta_{s, 0}(2)$ & $\Delta_{s, 0}(3)$\\
    \midrule
    1 & $2.9,\  2.3,\  1.1$ & $-0.5,\  1.0,\  1.9$ & 3 & 7 & 4 \\
    2 & $1.4,\  2.4,\  2.1$ & $0.3,\  -0.8,\  0.1$ & 1 & 4 & 10 \\
    3 & $3.3,\  0.5,\  1.8$ & $-0.1,\  -0.8,\ 2.0$ & 6 & 1 & 4 \\
    \bottomrule
    \end{tabular}
\end{table}
\begin{table}[htb]
    \centering
    \caption{The maximum of $\sum_{s=1}^3 \sum_{j=1}^{l_s} \Delta_{s, c}(j)$ over 
$(l_1, l_2, l_3) \in \mathcal{K}_{\bs{n}}(N-k)$ achieved by different optimization methods under various values of $N-k$
and the resulting $p$-value from the tail probability $G(\cdot)$ in \eqref{eq:G} evaluated at the minimum value of $\sum_{s=1}^S t_{s, c} (0)-\sum_{s=1}^3 \sum_{j=1}^{l_s} \Delta_{s, c}(j)$ achieved by different optimization methods for the observed data in Table \ref{tab:greedy_counter}. 
Specifically, ILP stands for integer linear programming, LP stands for linear programming, 
Greedy stands for the naive greedy algorithm, 
and GT stands for the greedy algorithm with the optimal monotone dominating transformation.}

    \label{tab:greedy_example_pval}
    \begin{tabular}{ccccccccccc}
    \toprule
    $N-k$ & & \multicolumn{4}{c}{$\sum_{s=1}^3 \sum_{j=1}^{l_s} \Delta_{s, 0}(j)$} & & \multicolumn{4}{c}{$p$-value} \\
    \midrule
        & & ILP & LP & Greedy&GT & & ILP & LP & Greedy&GT\\
    \midrule
    1 & & 6 & 6 & 6 &6 && 0.11 & 0.11 & 0.11& 0.11 \\
    2 & & 10 & 11 & 9 & 11&& 0.21 & 0.27 & 0.15 & 0.27\\
    3 & & 16 & 16 & 16 &16 && 0.47 & 0.47 & 0.47 & 0.47\\
    4 & & 21 & 21 & 20 &21 && 0.69 & 0.69 & 0.57 & 0.69 \\
    5 & & 25 & 26 & 21 &26&& 0.84 & 0.85 & 0.63 & 0.85\\
    6 & & 31 & 31 & 25 &31&& 0.95 & 0.95 & 0.79 & 0.95 \\
    7 & & 35 & 35 & 35 &35&& 0.98 & 0.98 & 0.98& 0.98\\
    8 & & 36 & 37.5 & 36 &37.5 && 0.99 & 1.00 & 0.99  & 1.00\\
    9 & & 40 & 40 & 40 & 40 && 1.00 & 1.00 & 1.00  & 1.00\\
    \bottomrule
    \end{tabular}
\end{table}


\subsection{Proof for the equivalence of the piecewise-linear optimizations in \eqref{eq:pwl}}\label{sec:proof_pwl}

\begin{lemma}\label{lemma:KR_pwl}
Let $n=(n_1, \ldots, n_S)$ be any vector consisting of $S\ge 1$ nonnegative integers, 
$k$ be any nonnegative integer less than or equal to $\sum_{s=1}^S n_s$, recall the definition of $\mathcal{K}_{n}(k)$ in \eqref{eq:setKn}, and define 
\begin{align*}
    \mathcal{R}_n(k) = \Big\{ (k_1, k_2, \ldots, k_S) \in \mathbb{R}^S: \ \sum_{s=1}^S k_s = k, \text{ and }\  0\le k_s \le n_s, 1\le s \le S
    \Big\}, 
\end{align*}
which is the same as $\mathcal{K}_{n}(k)$ except that we drop the integer constraints. 
Let
$\{b_{sj}: 1\le s\le S, 0\le j\le n_{s}\}$ be any finite real numbers, 
and 
$g_s(\cdot)$ be the piecewise-linear function on $[0, n_s]$ interpolating points
$(j, b_{sj})$ for $j=0,1, \ldots, n_s$, i.e., 
\begin{align*}
    g_s(x) = \I(x=0) b_{s0} + \sum_{j=0}^{n_s-1} \I(j < x\le j+1) [b_{sj} + (b_{s,j+1} - b_{sj})(x-j)], 
    \quad (1\le s\le S). 
\end{align*}
Define 
$
g(l_1, \ldots, l_S) \equiv \sum_{s=1}^S g_s(l_s). 
$
Then 
\begin{align}\label{eq:KR_pwl}
    \min_{(l_1, \ldots, l_S) \in \mathcal{K}_{n}(k)} g(l_1, \ldots, l_S)
    = \min_{(l_1, \ldots, l_S) \in \mathcal{R}_{n}(k)} g(l_1, \ldots, l_S),
\end{align}
and, by the same logic, 
\begin{align}\label{eq:KR_pwl_max}
    \max_{(l_1, \ldots, l_S) \in \mathcal{K}_{n}(k)} g(l_1, \ldots, l_S)
    = \max_{(l_1, \ldots, l_S) \in \mathcal{R}_{n}(k)} g(l_1, \ldots, l_S). 
\end{align}
\end{lemma}
\begin{proof}[Proof of Lemma \ref{lemma:KR_pwl}]
Note that $g(l_1, \ldots, l_S)$ is continuous in $(l_1, \ldots, l_S)$, and $\mathcal{R}_{n}(k)$ is a compact set. 
Thus, the quantities on the right hand side of \eqref{eq:KR_pwl} and \eqref{eq:KR_pwl_max} are well-defined. 
Moreover, because $-g(l_1, \ldots, l_S) = \sum_{s=1}^S -g_s(l_s)$ is still a summation of piecewise-linear functions,  
to prove Lemma \ref{lemma:KR_pwl}, it suffices to prove \eqref{eq:KR_pwl}. 
Because $\mathcal{K}_{n}(k)\subset \mathcal{R}_{n}(k)$,
it suffices to prove that, 
for any given $(l_1, \ldots, l_S) \in \mathcal{R}_{n}(k)$, 
\begin{align}\label{eq:KR_suff}
    g(l_1, \ldots, l_S)
    \ge 
     \min_{(l_1, \ldots, l_S) \in \mathcal{K}_{n}(k)} g(l_1, \ldots, l_S). 
\end{align}
Let $\zeta_s = l_s - \lfloor l_s \rfloor \in [0, 1)$ for all $s$. 
If all the $\zeta_s$'s are zero, then \eqref{eq:KR_suff} holds obviously. 
Otherwise, 
$\mathcal{C} = \{s: \zeta_s > 0, 1\le s\le S\}$ is a nonempty set, 
and $\sum_{s\in \mathcal{C}} \zeta_s = \sum_{s=1}^S l_s - \sum_{s=1}^S \lfloor l_s \rfloor = k - \sum_{s=1}^S \lfloor l_s \rfloor$ must be a positive integer. 
Let $\tilde{s} = \argmin_{s\in \mathcal{C}} ( b_{s, \lfloor l_s \rfloor+1} - b_{s, \lfloor l_s \rfloor})$. 
Because $\sum_{s\in \mathcal{C}} \zeta_s \ge 1$ and equivalently $1 - \zeta_{\tilde{s}} \le \sum_{s\in \mathcal{C}, s\ne \tilde{s}} \zeta_s$, 
we must be able to define 
$(\zeta_1', \ldots, \zeta_S')$ such that 
(i) $\zeta_{\tilde{s}}' = 1$, 
(ii) $0 \le \zeta_s' \le \zeta_s$ for $s\in \mathcal{C}$ and $s\ne \tilde{s}$, 
(iii) $\zeta_s' = \zeta_s = 0$ for $s\notin \mathcal{C}$, 
and 
(iv) $\sum_{s=1}^S \zeta_s' = \sum_{s=1}^S \zeta_s$. 
We then have 
\begin{align*}
    & \quad \ \sum_{s=1}^S \zeta_s' (b_{s, \lfloor l_s \rfloor+1} - b_{s, \lfloor l_s \rfloor}) - \sum_{s=1}^S \zeta_s ( b_{s, \lfloor l_s \rfloor+1} - b_{s, \lfloor l_s \rfloor} )
    \\
    & = 
    (\zeta_{\tilde{s}}' - \zeta_{\tilde{s}} ) ( b_{\tilde{s}, \lfloor l_{\tilde{s}} \rfloor+1} -  b_{\tilde{s}, \lfloor l_{\tilde{s}} \rfloor} )- 
    \sum_{s\in \mathcal{C}, s\ne \tilde{s}} (\zeta_s - \zeta_s') ( b_{s, \lfloor l_s \rfloor+1} - b_{s, \lfloor l_s \rfloor} )\\
    & \le (\zeta_{\tilde{s}}' - \zeta_{\tilde{s}} ) ( b_{\tilde{s}, \lfloor l_{\tilde{s}} \rfloor+1} -  b_{\tilde{s}, \lfloor l_{\tilde{s}} \rfloor} ) - 
    \sum_{s\in \mathcal{C}, s\ne \tilde{s}} (\zeta_s - \zeta_s') ( b_{\tilde{s}, \lfloor l_{\tilde{s}} \rfloor+1} -  b_{\tilde{s}, \lfloor l_{\tilde{s}} \rfloor} ) = 0,
\end{align*}
where the second last equality holds due to the definition of $\tilde{s}$ and property (ii) in the construction of $(\zeta_1', \ldots, \zeta_S')$, 
and the last equality holds due to properties (iii) and (iv) in the construction of $(\zeta_1', \ldots, \zeta_S')$. 
Define $l_s' = \lfloor l_s \rfloor + \zeta_s'$ for all $s$. We then verify that 
\begin{align*}
    g(l_1, \ldots, l_S) & = \sum_{s=1}^S 
    \left\{ b_{s, \lfloor l_s \rfloor}
    + (l_s - \lfloor l_s \rfloor) ( b_{s, \lfloor l_s \rfloor+1} - b_{s, \lfloor l_s \rfloor} )
    \right\}
    \\
    & 
    = 
    \sum_{s=1}^S 
    \left\{ b_{s, \lfloor l_s \rfloor}
    + \zeta_s ( b_{s, \lfloor l_s \rfloor+1} - b_{s, \lfloor l_s \rfloor} )
    \right\}
    \ge 
    \sum_{s=1}^S 
    \left\{ b_{s, \lfloor l_s \rfloor}
    + \zeta_s' ( b_{s, \lfloor l_s \rfloor+1} - b_{s, \lfloor l_s \rfloor} )
    \right\}\\
    & = 
    \sum_{s=1}^S 
    \left\{ b_{s, \lfloor l_s' \rfloor}
    + (l_s' - \lfloor l_s'\rfloor) ( b_{s, \lfloor l_s' \rfloor+1} - b_{s, \lfloor l_s' \rfloor} )
    \right\}
    = g(l_1', \ldots, l_S'). 
\end{align*}
Moreover, $(l_1', \ldots, l_S')$ is also in $\mathcal{R}_n(k)$, and it contains strictly less non-integer elements than $(l_1, \ldots, l_S)$. 
As long as $(l_1', \ldots, l_S')$ contains non-integer elements, 
by the same logic, we can further find $(l_1'', \ldots, l_S'')$, which contains strictly less non-integer elements than $(l_1', \ldots, l_S')$ and satisfies 
$g(l_1', \ldots, l_S') \ge g (l_1'', \ldots, l_S'')$. 
Eventually, we can find $(\overline{l}_1, \ldots, \overline{l}_S) \in \mathcal{K}_n(k)$ such that $g(l_1, \ldots, l_S) \ge g(\overline{l}_1, \ldots, \overline{l}_S)$. 
This immediately implies \eqref{eq:KR_suff}. 
From the above, Lemma \ref{lemma:KR_pwl} holds. 
\end{proof}

\begin{remark}
Here we give a remark showing that the objective function in the piecewise-linear optimization in \eqref{eq:pwl} can be non-convex. 
Specifically, we consider the data set in \S  \ref{sec:numerical_example} or more precisely Table \ref{tab:greedy_counter}. 
Figure \ref{fig:pwl_nonconvex} shows the piecewise-linear functions $f_{s,0}(\cdot)$'s for the three strata, defined as in \S \ref{sec:pwl}. 
The objective function for the piecewise-linear optimization is then $f_0(l) = \sum_{s=1}^3 f_{s,0}(l_s)$ for $l=(l_1, l_2, l_3)$. 
Let $a = (0, 3, 3)$, $b=(2, 5, 3)$, and $c= (a+b)/2 = (1, 4, 3)$. 
By definition, we can verify that 
$f_0(a) = 14$, $f_0(b) = 4$, and $f_0(c) = 11$. 
Obviously, $(f_0(a)+f_0(b))/2 = 9 < 11 = f_0(c) = f_0((a+b)/2)$. 
Therefore, $f_0(\cdot)$ is not a convex function. 
\end{remark}

\begin{figure}[h]
    \centering
    \subfloat[$f_{1,0}(x)$]{\includegraphics[width=0.25\textwidth]{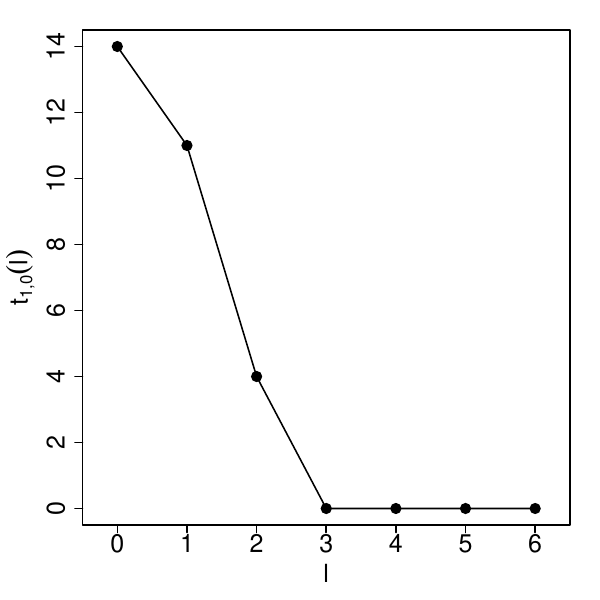}}\hfil 
    \subfloat[$f_{2,0}(x)$]{\includegraphics[width=0.25\textwidth]{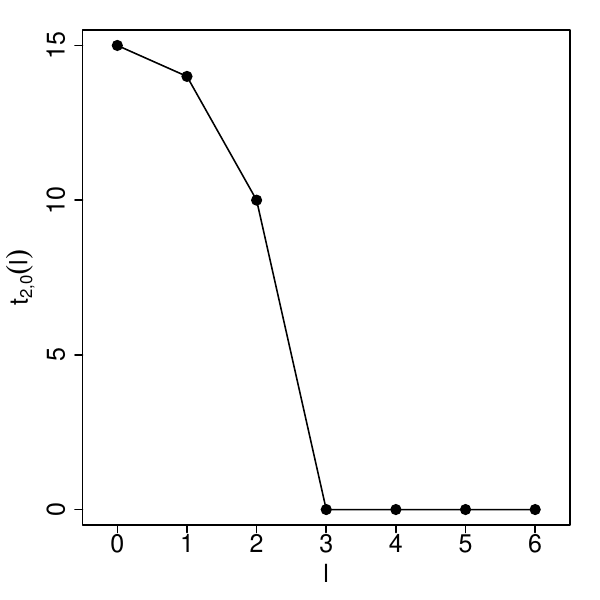}} \hfil
    \subfloat[$f_{3,0}(x)$]{\includegraphics[width=0.25\textwidth]{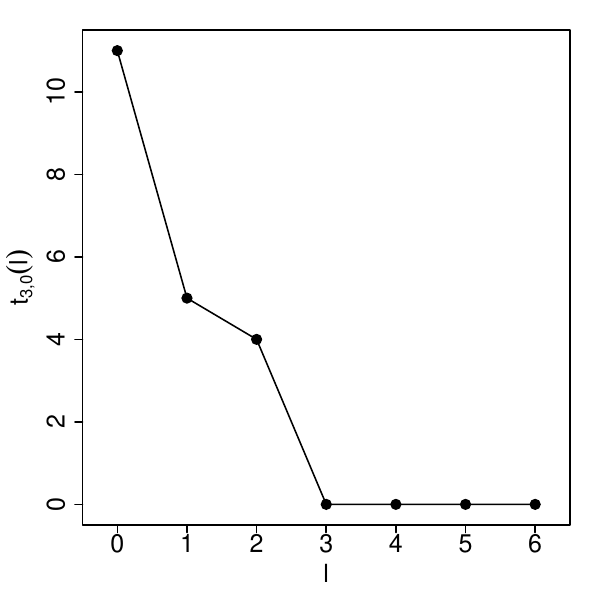}} 
\caption{
Piecewise-linear functions $f_{s,0}(\cdot)$'s for all strata, defined as in \S \ref{sec:pwl}, for the observed data in \S \ref{sec:numerical_example} or more precisely Table \ref{tab:greedy_counter}. 
}
\label{fig:pwl_nonconvex}
\end{figure}

\subsection{Greedy algorithm with the optimal monotone dominating transformation solves the linear programming problem}\label{sec:proof_greedy}
We first introduce the following key lemma. 

\begin{lemma}\label{lemma:opt_trans_ineq}
 For any $n\ge 1$, 
$b=(b_0, \ldots, b_n) \in \mathbb{R}^{n+1}$, 
$a=(a_1, \ldots, a_n) \in \mathbb{R}^n$ with 
$a_i = b_{i-1} - b_i$ for all $i$, 
and any real number $k \in [0, n]$, 
\begin{align}\label{eq:proof_monotone_strata}
    \sum_{i=1}^{\lfloor k \rfloor} 
    \tilde{\Psi}_i ( a ) 
    + 
    (k - \lfloor k \rfloor) \tilde{\Psi}_{\lfloor k \rfloor+1} ( a )
    = 
    \max_{(x_0, \ldots, x_n)\in \mathcal{S}_{n,k}}
    \sum_{l=0}^{n}  x_{l}( b_0- b_l ), 
\end{align}
where $\lfloor k \rfloor$ denotes the largest integer less than or equal to $k$, and 
$\mathcal{S}_{n, k} = \{ (x_0, \ldots, x_n)\in [0,1]^{n+1}: \sum_{l=0}^{n} x_{l}=1, \sum_{l=0}^{n} l x_{l} = k \}$. 
\end{lemma}

\begin{proof}[Proof of Lemma \ref{lemma:opt_trans_ineq}]
We first note that $\sum_{l=0}^{n} x_{l}( b_0- b_l )$ is a continuous function of $(x_0,\ldots, x_n)$, and $\mathcal{S}_{n, k}$ is a compact set. 
Thus, the maximum on the right hand side of \eqref{eq:proof_monotone_strata} is well-defined. 
We will prove \eqref{eq:proof_monotone_strata} by induction. 
When $k \in [0,1]$, the left hand side of  \eqref{eq:proof_monotone_strata} must equal $k \tilde{\Psi}_1(a)$. 
By definition, $\tilde{\Psi}_1(a) = (b_0 - b_{j_1})/j_1$, where $ j_1 = \argmax_{1\le j\le n} (b_0 - b_j)/j$. 
This implies that, 
for any $(x_0,\ldots, x_n)\in\mathcal{S}_{n,k}$, 
\begin{align*}
    \sum_{l=0}^{n}  x_{l}( b_0- b_l ) = \sum_{l=1}^{n}  x_{l}( b_0- b_l ) 
    \le \sum_{l=1}^{n}  l x_{l} \frac{b_0-b_l}{l}  \le \tilde{\Psi}_1(a) \sum_{l=1}^{n}  l x_{l}
    = k  \tilde{\Psi}_1(a). 
\end{align*}
Moreover, define $x_0 = 1 - k/j_1$, $x_{j_1} = k/j_1$, and $x_l = 0$ for $1\le l\le n$ and $l\ne j_1$. 
We can verify that $(x_0,\ldots, x_n)\in\mathcal{S}_{n,k}$, and 
$\sum_{l=0}^{n}  x_{l}( b_0- b_l ) = x_{j_1}( b_0- b_{j_1} ) = k ( b_0- b_{j_1} )/j_1 =k  \tilde{\Psi}_1(a)$. 
Therefore, \eqref{eq:proof_monotone_strata} holds for $k\in [0,1]$. 
Assume that \eqref{eq:proof_monotone_strata} holds when $k \in (m-1, m]$ for some $1\le m < n$. 
Below we consider the case where $k \in (m, m+1]$.
Let $\zeta = k - m \in (0, 1]$. 
By definition, the left hand side of \eqref{eq:proof_monotone_strata} has the following equivalent forms:
\begin{align}\label{eq:proof_monotone_max}
    & \quad \ \sum_{i=1}^{\lfloor k \rfloor} 
    \tilde{\Psi}_i ( a ) 
    + 
    (k - \lfloor k \rfloor) \tilde{\Psi}_{\lfloor k \rfloor+1} ( a )
    \nonumber
    \\
    &= 
    \sum_{i=1}^{m} 
    \tilde{\Psi}_i(a) + \zeta \tilde{\Psi}_{m+1}(a)
    = \sum_{i=1}^{m} 
    \tilde{\Psi}_i(a) + \zeta  \max_{m+1\le j\le n} \frac{(b_0-b_j) - \sum_{i=1}^m \tilde{\Psi}_i(a)}{j-m}
    \nonumber
    \\
    & = 
    \max_{m+1 \le j \le n} \left\{ \frac{ \zeta (b_0-b_j)}{j-m}+\frac{j-m-\zeta}{j-m}\sum_{i=1}^{m}  \tilde{\Psi}_i( a) \right\}. 
\end{align}

We first prove that the left hand side of \eqref{eq:proof_monotone_strata} is less than or equal to the right hand side of \eqref{eq:proof_monotone_strata}. 
Let $j_{m+1} \in \{m+1, \ldots, n\}$ be the maximizer of the quantity in \eqref{eq:proof_monotone_max}. 
By our assumption, there must exist $(x_0,\ldots, x_n)\in\mathcal{S}_{n,m}$ such that  
$\sum_{i=1}^{m}  \tilde{\Psi}_i( a) = \sum_{l=0}^{n} x_{l}( b_0-b_l)$. 
Consequently, 
\begin{align*}
    & \quad \ \sum_{i=1}^{\lfloor k \rfloor} 
    \tilde{\Psi}_i ( a ) 
    + 
    (k - \lfloor k \rfloor) \tilde{\Psi}_{\lfloor k \rfloor+1} ( a ) \\
    &= \frac{ \zeta( b_0-b_{j_{m+1}})}{j_{m+1}-m}+\frac{j_{m+1}-m-\zeta}{j_{m+1}-m}\sum_{i=1}^{m}  \tilde{\Psi}_i( a)
    = \frac{ \zeta( b_0-b_{j_{m+1}})}{j_{m+1}-m}+\frac{j_{m+1}-m-\zeta}{j_{m+1}-m}\sum_{l=0}^{n} x_{l}( b_0-b_l)
    \\
    & = \sum_{l=0}^{n} x'_{l}( b_0-b_l), 
\end{align*}
where 
\begin{align*}
    x'_{l} \equiv 
    \begin{cases}
    (j_{m+1}-m-\zeta)x_l/(j_{m+1}-m), & \text{for } l\ne j_m\\
    (j_{m+1}-m-\zeta)x_l/(j_{m+1}-m) + \zeta/(j_{m+1}-m), & \text{for } l=j_{m+1}.
    \end{cases}
\end{align*}
We can verify that $0\le x_l'\le 1$ for all $l$, and $\sum_{l=0}^n x_l' = 1$. 
Moreover, 
\begin{align*}
    \sum_{l=0}^n l x_l' =  \frac{j_{m+1}-m-\zeta}{j_{m+1}-m} \sum_{l=0}^n l x_l + \frac{\zeta j_{m+1}}{j_{m+1}-m}
    = \frac{(j_{m+1}-m-\zeta)m + \zeta j_{m+1}}{j_{m+1}-m} = m+\zeta = k. 
\end{align*}
Thus, $(x_0',\ldots,x_n')\in \mathcal S_{n,k}$. 
This implies that 
\begin{align*}
    \sum_{i=1}^{\lfloor k \rfloor} 
    \tilde{\Psi}_i ( a ) 
    + 
    (k - \lfloor k \rfloor) \tilde{\Psi}_{\lfloor k \rfloor+1} ( a )
    = 
    \sum_{l=0}^{n} x'_{l}( b_0-b_l)
    \le     \max_{(x_0, \ldots, x_n)\in \mathcal{S}_{n,k}}
    \sum_{l=0}^{n}  x_{l} ( b_0- b_l ). 
\end{align*}

We then prove that the left hand side of \eqref{eq:proof_monotone_strata} is greater than or equal to the right hand side of \eqref{eq:proof_monotone_strata}. 
From \eqref{eq:proof_monotone_max}, it suffices to prove that, for any given $(x_0,\ldots, x_n)\in\mathcal S_{n,k}$, 
\begin{align}\label{eq:upper_x}
    \sum_{l=0}^{n} x_{l} ( b_0- b_l ) \le \max_{m+1 \le j \le n} \left\{ \frac{ \zeta (b_0-b_j)}{j-m}+\frac{j-m-\zeta}{j-m}\sum_{i=1}^{m}  \tilde{\Psi}_i( a) \right\}.
\end{align}
Let $\Omega \equiv \sum_{i=m+1}^n (i-m) x_i \ge \sum_{i=0}^n (i-m) x_i = \sum_{i=0}^n i x_i - m \sum_{i=0}^n x_i = k - m = \zeta > 0$, 
$\lambda_i \equiv (i-m) x_i / \Omega \ge 0$ for $m+1\le i\le n$, 
and 
\begin{align}\label{eq:Lambda}
    \Lambda \equiv \sum_{i=m+1}^n \frac{\lambda_i}{i-m} =  \frac{\sum_{i=m+1}^n x_i}{\Omega} = \frac{\sum_{i=m+1}^n x_i}{\sum_{i=m+1}^n (i-m) x_i} \le 1. 
\end{align}
If $\zeta \Lambda = 1$, then $\zeta = 1$ and $\Lambda = 1$. 
We then have $k=m+\zeta = m+1$, and $\sum_{i=m+1}^n (i-m-1) x_i = 0$, which further implies that 
$x_{i} = 0$ for $m+1<i\le n$. 
Consequently, 
$m+1 = k = \sum_{i=0}^n i x_i = \sum_{i=0}^{m+1} i x_i \le (m+1) \sum_{i=0}^{m+1} x_i \le m+1$. 
Thus, we must have $x_{m+1} = 1$, and 
\begin{align*}
    \sum_{l=0}^{n} x_{l} ( b_0- b_l ) & = b_0 - b_{m+1} 
    = \frac{ \zeta( b_0-b_{m+1})}{(m+1)-m}+\frac{(m+1)-m-\zeta}{(m+1)-m}\sum_{i=1}^{m}  \tilde{\Psi}_i( a)\\
    & \le 
    \max_{m+1 \le j \le n} \left\{ \frac{ \zeta (b_0-b_j)}{j-m}+\frac{j-m-\zeta}{j-m}\sum_{i=1}^{m}  \tilde{\Psi}_i( a) \right\},
\end{align*}
i.e., \eqref{eq:upper_x} holds. 
If $\zeta \Lambda < 1$, we define $(x_0', x_1', \ldots, x_n')$ as 
\begin{align}\label{eq:trans_m}
    x'_l =     \begin{cases}
    {x_l}/(1-\zeta \Lambda), & \text{for } 0 \le l\le m\\
    \{x_l-\zeta \lambda_l/(l-m)\}/(1-\zeta\Lambda) = x_l (1-\zeta \Omega^{-1})/(1-\zeta \Lambda), &\text{for } m < l \le n. 
    \end{cases}
\end{align}
Below we verify that $(x'_0,\ldots, x'_n)\in \mathcal S_{n,m}$. 
Because $\zeta \le 1$,  $\Omega \ge 1$ and $\sum_{i=m+1}^n x_i\le 1$, from \eqref{eq:Lambda}, we have 
\begin{align*}
    1 - \zeta \Lambda =  1 - \frac{\zeta \sum_{i=m+1}^n x_i}{\Omega} 
    \begin{cases}
    \ge 1 - \sum_{i=m+1}^n x_i = \sum_{i=0}^m x_i, 
    \\
    \ge 1 - \zeta \Omega^{-1}. 
    \end{cases}
\end{align*}
Thus, for $0 \le l \le m$, $0 \le x_l' = x_l/(1-\zeta \Lambda) \le \sum_{i=0}^m x_i /(1-\zeta \Lambda) \le 1$, 
and for $m<l\le n$, 
$
    0 \le x_l' = x_l (1-\zeta \Omega^{-1})/(1-\zeta \Lambda) \le x_l \le 1. 
$
Moreover, from \eqref{eq:Lambda}, 
\begin{align*}
    \sum_{i=0}^nx'_i=\left(\sum_{i=0}^n{x_i}-\zeta \sum_{i=m+1}^n \frac{\lambda_i}{i-m}\right)/(1-\zeta\Lambda)
    =
    (1-\zeta \Lambda)/(1-\zeta \Lambda)
    =1, 
\end{align*}
and 
\begin{align*}
     \sum_{i=0}^n i x'_i & =\left(\sum_{i=0}^n{i x_i}- \frac{\zeta\sum_{i=m+1}^nix_i}{\Omega}\right)/(1-\zeta \Lambda) 
     =\left(m+\zeta- \frac{\zeta\sum_{i=m+1}^nix_i}{\Omega}\right)/(1-\zeta \Lambda) 
     \\
     & = \left(m- \zeta\frac{\sum_{i=m+1}^nix_i - \Omega}{\Omega}\right)/(1-\zeta \Lambda) 
     = 
     \left(m- \zeta \frac{m\sum_{i=m+1}^n x_i}{\Omega}\right)/(1-\zeta \Lambda)
     \\
     & =
     (m-\zeta m\Lambda)/(1-\zeta \Lambda)
     =
     m. 
\end{align*}
Therefore,  $(x'_0,\ldots, x'_n)\in \mathcal S_{n,m}$.
Note that $\sum_{i=m+1}^n \lambda_i = 1$. 
By the definition in \eqref{eq:trans_m}, 
\begin{align*}
    x_l = x_l' (1-\zeta \Lambda)  =  x_l' \left( \sum_{i=m+1}^n \lambda_i - \zeta \sum_{i=m+1}^n \frac{\lambda_i}{i-m} \right)  
    = x_l' \sum_{i=m+1}^n  \frac{\lambda_i(i-m-\zeta)}{i-m}, \quad (0\le l\le m), 
\end{align*}
and by the same logic, 
\begin{align*}
    x_l = (1-\zeta \Lambda) x_l' + \frac{\zeta \lambda_l}{l-m}
    = x_l' \sum_{i=m+1}^n  \frac{\lambda_i(i-m-\zeta)}{i-m} + \frac{\zeta \lambda_l}{l-m}, \quad (m< l\le n). 
\end{align*}
The left hand side of \eqref{eq:upper_x} then has the following equivalent forms: 
\begin{align*}
    \sum_{l=0}^{n} x_l( b_0- b_l )
    & = 
    \sum_{l=0}^{n} ( b_0- b_l ) \left\{ x'_l \sum_{i=m+1}^n  \frac{\lambda_i(i-m-\zeta)}{i-m} \right\}
    + 
    \sum_{l=m+1}^{n}( b_0- b_l )\frac{\zeta \lambda_l}{l-m}
    \\
    & = 
    \sum_{i=m+1}^n \frac{\lambda_i(i-m-\zeta)}{i-m} \left\{ \sum_{l=0}^{n} x'_l ( b_0- b_l ) \right\} + \sum_{l=m+1}^{n}( b_0- b_l )\frac{\zeta \lambda_l}{l-m}
    \\
    & = 
    \sum_{j=m+1}^n \lambda_j \left[ \frac{\zeta(b_0-b_j)}{j-m} + \frac{j-m-\zeta}{j-m} \cdot \left\{ \sum_{l=0}^{n} x'_l ( b_0- b_l ) \right\}\right].
\end{align*}
By our assumption and the fact that $(x'_0,\ldots, x'_n)\in \mathcal S_{n,m}$, we must have $\sum_{l=0}^{n}  x'_l ( b_0- b_l ) \le \sum_{i=1}^m \tilde{\Psi}_i ( a )$. Consequently, we can bound the left hand side of \eqref{eq:upper_x} by 
\begin{align*}
    \sum_{l=0}^{n}  x_l( b_0- b_l ) 
    & \le 
    \sum_{j=m+1}^n \lambda_j \left\{ \frac{\zeta(b_0-b_j)}{j-m} + \frac{j-m-\zeta}{j-m}  \sum_{i=1}^m \tilde{\Psi}_i ( a ) \right\}
    \\
    & \le \max_{m+1 \le j \le n} \left\{ \frac{ \zeta (b_0-b_j)}{j-m}+\frac{j-m-\zeta}{j-m}\sum_{i=1}^{m}  \tilde{\Psi}_i( a) \right\},
\end{align*}
where the last inequality holds because $\sum_{j=m+1}^n \lambda_j=1$ and $\lambda_j \ge 0$ for $m<j\le n$. 

From the above, \eqref{eq:proof_monotone_strata} must hold when 
{\lxr $k \in (m, m+1].$}
Therefore, Lemma \ref{lemma:opt_trans_ineq} holds.
\end{proof}

Below we prove that the Greedy algorithm in \S \ref{sec:greedy} solves the linear programming problem in \S \ref{sec:ILP_LP}. 
Let $t_{k,c}^\LP$ denote the optimal value from the linear programming in \eqref{eq:linear_program}, and $t_{k,c}^\GT$ denote the optimal value from the Greedy algorithm in Algorithm \ref{alg:greedy}. 
Define $\mathcal{S}_{n, k}$ the same as in Lemma \ref{lemma:opt_trans_ineq}. 

\begin{proof}[Proof for $t_{k,c}^{\LP} \le t_{k,c}^{\GT}$]
Let 
$
(\tilde{l}_1, \ldots, \tilde{l}_S) \equiv  \argmax_{(l_1, \ldots, l_S) \in \mathcal{K}_{\bs{n}}(N-k)} \sum_{s=1}^S \sum_{i=1}^{l_s} 
\tilde{\Psi}_i\left( \bs{\Delta}_{s,c}(\overline{n_{sk}} ) \right). 
$
Recall that, by definition, $\bs{\Delta}_{s,c}(i) = t_{s,c}(i-1) - t_{s,c}(i)$ for $1\le i \le n_{sk}$. 
From Lemma \ref{lemma:opt_trans_ineq}, for each $1\le s\le S$, there exists $(x_{s0}, \ldots, x_{sn_{sk}}) \in \mathcal{S}_{n, \tilde{l}_s}$ such that 
\begin{align*}
    \sum_{i=1}^{\tilde{l}_s} 
    \tilde{\Psi}_i \left( \bs{\Delta}_{s,c}(\overline{n_{sk}} ) \right)
    & = \sum_{l=0}^{n_{sk}}  x_{sl} \{ t_{s,c}(0)- t_{s,c}(l) \}
    = t_{s,c}(0)-  \sum_{l=0}^{n_{sk}}  x_{sl} t_{s,c}(l), 
\end{align*}
where the last equality holds because $\sum_{l=0}^{n_{sk}}  x_{sl}=1$. 
Consequently, 
\begin{align*}
    t_{k,c}^{\GT} & = \sum_{s=1}^S t_{s,c}(0) - \sum_{s=1}^S \sum_{j=1}^{\tilde{l}_s} 
    \tilde{\Psi}_j\left( \bs{\Delta}_{s,c}(\overline{n_{sk}} ) \right)
    = \sum_{s=1}^S \sum_{l=0}^{n_{sk}}  x_{sl} t_{s,c}(l). 
\end{align*}
By the definition of $\mathcal{S}_{n, \tilde{l}_s}$'s, 
(i) $x_{sl}\in [0,1]$ for all $l, s$, 
(ii) $\sum_{l=0}^{n_{sk}}  x_{sl}=1$ for all $s$, 
and (iii) 
$
    \sum_{s=1}^S \sum_{l=0}^{n_{sk}} x_{sl} l = \sum_{s=1}^S \tilde{l}_s = N-k. 
$
Thus, $x_{sl}$'s satisfy the constraints in the linear programming \eqref{eq:linear_program}. 
This immediately implies that the solution from the linear programming \eqref{eq:linear_program} is 
bounded by $t_{k,c}^{\LP} \le t_{k,c}^{\GT}$. 
\end{proof}

\begin{proof}[Proof for  $t_{k,c}^{\LP} \ge t_{k,c}^{\GT}$]
To prove $t_{k,c}^{\LP} \ge t_{k,c}^{\GT}$, it suffices to prove that, for any given  $x_{sl}$'s satisfying that
$0\le x_{sl}\le 1$ for all $s, l$, 
$\sum_{l=0}^{n_{sk}}x_{sl}=1$ for all $s$, and $\sum_{s=1}^S\sum_{l=0}^{n_{sk}}x_{sl} l=N-k$, 
\begin{align*}
    \sum_{s=1}^S \sum_{l=0}^{n_{sk}}  x_{sl} t_{s,c}(l) \ge t_{k,c}^{\GT} 
    \Longleftrightarrow
    \max_{(l_1, \ldots, l_S) \in \mathcal{K}_{\bs{n}}(N-k)} \sum_{s=1}^S \sum_{j=1}^{l_s} 
    \tilde{\Psi}_j\left( \bs{\Delta}_{s,c}(\overline{n_{sk}} ) \right) \ge
    \sum_{s=1}^S \sum_{l=0}^{n_{sk}}x_{sl}\{ t_{s,c}(0)- t_{s,c}(l)\}.  
\end{align*}
For $1\le s\le S$ and real number $l_s\in [0, n_{sk}]$, define
\begin{align*}
    g_s (l_s) = \sum_{j=1}^{\lfloor l_s \rfloor} 
    \tilde{\Psi}_j\left( \bs{\Delta}_{s,c}(\overline{n_{sk}} ) \right) 
    + 
    (l_s - \lfloor l_s \rfloor)\tilde{\Psi}_{\lfloor l_s \rfloor+1}\left( \bs{\Delta}_{s,c}(\overline{n_{sk}} ) \right), 
\end{align*}
and define further $g(l_1, \ldots, l_S) = \sum_{s=1}^S g_s(l_s)$.
From Lemma \ref{lemma:KR_pwl},
it suffices to prove that 
\begin{align}\label{eq:suff_g}
    \max_{(l_1, \ldots, l_S) \in \mathcal{R}_{\bs{n}}(N-k)} g(l_1, \ldots, l_S) \ge
    \sum_{s=1}^S \sum_{l=0}^{n_{sk}}x_{sl}\{ t_{s,c}(0)- t_{s,c}(l)\}, 
\end{align}
where $\mathcal{R}_{\bs{n}}(N-k)$ is defined as in Lemma \ref{lemma:KR_pwl}.
Let $\overline{l}_s \equiv \sum_{l=0}^{n_{sk}}x_{sl} l$ for all $s$. 
By the property of $x_{sl}$'s, we can know that 
$(\overline{l}_1, \ldots, \overline{l}_S) \in \mathcal{R}_{\bs{n}}(N-k)$, 
and 
$(x_{s0}, \ldots, x_{sn_{sk}}) \in \mathcal{S}_{n_{sk}, \overline{l}_s}$. 
From Lemma \ref{lemma:opt_trans_ineq}, 
we can know that
$\sum_{l=0}^{n_{sk}}x_{sl}\{ t_{s,c}(0)- t_{s,c}(l)\} \le g_s(\overline{l}_s)$ for all $s$. 
Consequently, 
\begin{align*}
    \sum_{s=1}^S \sum_{l=0}^{n_{sk}}x_{sl}\{ t_{s,c}(0)- t_{s,c}(l)\} 
    \le \sum_{s=1}^S g_s(\overline{l}_s) = g(\overline{l}_1, \ldots, \overline{l}_S) 
    \le \max_{(l_1, \ldots, l_S) \in \mathcal{R}_{\bs{n}}(N-k)} g(l_1, \ldots, l_S), 
\end{align*}
i.e., \eqref{eq:suff_g} holds. 
Therefore, we must have $t_{k,c}^{\LP} \ge t_{k,c}^{\GT}$. 
\end{proof}

From the above, the Greedy algorithm indeed solves the linear programming problem. 

\subsection{Comment on computation cost of the greedy algorithm with the optimal transformation}\label{sec:comp_greedy}

The greedy algorithm with the optimal transformation mainly consists of two parts. First, we need to carry out the optimal transformation of the $\Delta_{s,c}$'s for each stratum $s$, as shown in \eqref{eq:Psi_optimal}. 
For each stratum, we can first calculate the cumulative summations of $\Delta_{s,c}$'s, which has complexity $O(n_s)$. 
Then, to calculate the $i$th term of the optimal transformed sequence, we need to calculate the summation of the first $i-1$ terms of the optimal transformed sequence, 
and then take the maximum of $O(n_s)$ linear combinations of known terms, which has complexity $O(n_s)$. 
Therefore, the total complexity for each stratum is $O(n_s^2)$, 
and the total complexity for all strata is $O(\sum_{s=1}^S n_s^2) = O(N\max_{s}n_s)$ .
Second, we apply the greedy algorithm to those monotone decreasing sequences, which is equivalent to combining these sequences into one sequence and summing up the largest $N-k$ elements. 
Thus, the complexity for this part is
at most that for sorting a sequence of length $N$, which is 
$O(N\log N)$.
From the above, 
the complexity of the greedy algorithm 
with the optimal transformation 
is $O(N\max_{s}n_s+N\log N)$.

\subsection{Greedy algorithm with the optimal monotone dominating transformation solves the integer linear programming problem under concave rank transformations}\label{sec:greedy_concave}

To prove that the greedy algorithm in \S \ref{sec:greedy} solves the integer linear programming problem in \eqref{eq:integer_program} under concave rank transformations, 
we need the following two lemmas. 

\begin{lemma}\label{lemma:concave}
Let $\phi$ be a concave function on an interval $\mathcal{I}$. Then for any $x_1, x_2, y_1, y_2\in \mathcal{I}$ with $x_2 - x_1 = y_2 - y_1 \ge 0$ and $x_1 \ge y_1$, we must have 
$\phi(x_2) - \phi(x_1) \le \phi(y_2) - \phi(y_1)$. 
\end{lemma}
\begin{proof}[Proof of Lemma \ref{lemma:concave}]
Lemma \ref{lemma:concave} follows from the standard property of concave functions, and we give a brief proof below.
Note that $y_1\le \min\{y_2,x_1\} \le \max\{y_2, x_1\} \le x_2$. Then by the property of concave function, we have
\begin{align*}
    \phi(x_1) = \phi\left( \frac{x_2-x_1}{x_2-y_1}y_1+\frac{x_1-y_1}{x_2-y_1}x_2\right)\ge
    \frac{x_2-x_1}{x_2-y_1}\phi\left( y_1\right)+
    \frac{x_1-y_1}{x_2-y_1}\phi\left(x_2\right),
\end{align*}
and similarly we have 
\begin{align*}
    \phi(y_2) = \phi\left( \frac{x_2-y_2}{x_2-y_1}y_1+\frac{y_2-y_1}{x_2-y_1}x_2\right)\ge
    \frac{x_2-y_2}{x_2-y_1}\phi\left( y_1\right)+
    \frac{y_2-y_1}{x_2-y_1}\phi\left(x_2\right),
\end{align*}
Combining the two inequalities, we have
\begin{align*}
    \phi(x_1)+\phi(y_2)\ge \frac{2x_2-(x_1+y_2)}{x_2-y_1}\phi\left( y_1\right)+
    \frac{(x_1+y_2)-2y_1}{x_2-y_1}\phi\left(x_2\right) = \phi(y_1)+\phi(x_2),
\end{align*}
where the last equality holds because $x_2-x_1 = y_2-y_1$. 
Therefore, Lemma \ref{lemma:concave} holds. 
\end{proof}

\begin{lemma}\label{lemma:concave_transform}
Let $\{r_i(0):i=1,2,\ldots, m\}$ be a strictly monotone increasing sequence of positive integers, and define $\{r_i(l):i=1,2,\ldots, m\}$ for $l\ge 1$ recursively as 
$r_1(l) = 1$ and $r_{i}(l) = r_{i-1}(l-1)+1$ for $2\le i \le m$. 
Let $\phi(\cdot)$ be an increasing and concave function defined on $[1, \infty)$, and define  $\Delta(l) = \sum_{i=1}^m \{\phi(r_i(l-1))-\phi(r_i(l))\}$ for $l\ge 1$. Then we must have 
$
    \Delta(l)\ge 
     \Delta(l+1)
$
for all $l\ge 1$.
\end{lemma}
\begin{proof}[Proof of Lemma \ref{lemma:concave_transform}]
Lemma \ref{lemma:concave_transform} holds obviously when $m=1$. Below we prove only the case where $m\ge 2$. 
For descriptive convenience, we define $r_0(l) = 0$ for any $l\ge 0$. 
Then by definition, we have $r_i(l) = r_{i-1}(l-1)+1$ for $1\le i \le m$ and $l\ge 1$. 
By definition, we can verify that for $l\ge 1$ and $1 \le i \le m$, 
\begin{align*}
    & \text{(i) } r_i(l) - r_{i-1}(l) \ge 1, 
    \qquad 
    \text{(ii) } r_{i}(l-1) = r_{i+1}(l) - 1 \ge r_{i}(l), 
    \\
    & 
    \text{(iii) } r_i(l) - r_{i}(l+1) = r_{i-1}(l-1) - r_{i-1}(l).
\end{align*}
By definition, Lemma \ref{lemma:concave} and the monotone increasing and concave properties of $\phi(\cdot)$, we then have 
\begin{align*}
    \Delta(l+1) & = \sum_{i=1}^m \{\phi(r_i(l))-\phi(r_i(l+1))\}
    \le 
     \sum_{i=1}^m \{\phi(r_{i-1}(l-1))-\phi(r_{i-1}(l))\}
    \\
    & = \sum_{i=1}^{m-1} \{\phi(r_{i}(l-1))-\phi(r_{i}(l))\}
    \le \sum_{i=1}^{m-1} \{\phi(r_{i}(l-1))-\phi(r_{i}(l))\} 
    + \phi(r_{m}(l-1))-\phi(r_{m}(l)) \\
    & = \Delta(l). 
\end{align*}
Therefore, Lemma \ref{lemma:concave_transform} holds.
\end{proof}



As discussed in \S \ref{sec:greedy}, 
the following Proposition implies that the greedy algorithm indeed solves the integer programming in \eqref{eq:integer_program}. 

\begin{proposition}\label{prop:concave_transform}
For the stratified rank score statistic in \eqref{eq:strat_rank_sum}, 
    if the rank transformation $\phi_s(\cdot)$ is a concave function for all $1\le s\le S$, 
    then for any observed data $(\bs{Z}, \bs{Y})$ and any $c\in \mathbb{R}$, 
    the resulting 
    $\Delta_{s,c}(j)$'s defined in \S \ref{sec:greedy} satisfy that, for all $s$, 
    $
        \Delta_{s,c}(1) \ge \Delta_{s,c}(2) \ge \ldots 
        \ge \Delta_{s,c}(n_s). 
    $
\end{proposition}
\begin{proof}[Proof of Proposition \ref{prop:concave_transform}]
The proof of Proposition \ref{prop:concave_transform} relies on the closed form of $t_{s,c}(l_s)$ given by \citet{li2020quantile}, which is briefly explained at the end of Section \ref{sec:simp_opt}. 
Recall that 
$t_{s, c}(l_s) = t_s(\bs{Z}_s, \bs{Y}_s - \bs{Z}_s \circ \bs{\xi}_{s,c}(l_s))$, 
where $\bs{\xi}_{s,c}(l_s) = (\xi_{s,c,1}(l_s), \xi_{s,c,2}(l_s), \ldots, \xi_{s,c,n_s}(l_s))^\top \in \mathbb{R}^{n_s}$, 
and $\xi_{s, c, i}(l_s)$ equals $\infty$
if $Z_{si}=1$ and $\rank_i(\bs{Y}_s)$ is among the $l_s$ largest elements  of  $\{\rank_j(\bs{Y}_s): Z_{sj} = 1, 1\le j \le n_s\}$, 
and $c$ otherwise.

Let $\mathcal T_s$ denote the set of units receiving treatment in stratum $s$, i.e, $\mathcal T_s = \{i:Z_{si}=1, 1\le i \le n_s\}$, and $r_{si}(l_s)=\rank_i(\bs{Y}_s - \bs{Z}_s \circ \bs{\xi}_{s,c}(l_s))$. Then $t_{s, c}(l_s) = \sum_{i\in\mathcal{T}_s}\phi_s(r_{si}(l_s))$. 
Moreover, 
we can verify that 
the sorted sequence of $\{r_{si}(l_s): i\in \mathcal{T}_s\}$, for $l_s\ge 1$,  satisfies the recursion in Lemma \ref{lemma:concave_transform}. 
Therefore, from Lemma \ref{lemma:concave_transform}, we can immediately derive Proposition \ref{prop:concave_transform}.
\end{proof}

\begin{remark}
The Stephenson rank transformations are always convex functions, and we give a brief proof here. 
We first extend the domain of the Stephenson rank transformation $\phi_s$ as follows:
\begin{align*}
    \phi_s(r) = 
    \begin{cases}
    (r-1)(r-2)\cdots \{(r-1)-(h_s-1)+1\}/(h_s-1)!, & \text{if } r \ge h_s-1,\\
    0, & \text{otherwise}. 
    \end{cases}
\end{align*}
Below we 
show that $\phi_s(\cdot)$ is convex on $\mathbb{R}$. 
First, $\phi_s(\cdot)$ is convex on $[h_s-1,\infty)$\footnote{See, e.g., https://math.stackexchange.com/q/1923193, for a proof.}. 
Second, by definition, $\phi_s(r)$ takes value zero at $r \le h_s-1$, is continuous at $r = h_s-1$, and is increasing at $r\ge h_s-1$. 
Therefore, we can verify that $\phi_s(\cdot)$ is convex on the whole real line. 
\end{remark}

\end{document}